\documentclass[11pt]{amsart}
\ifdefined\pdfminorversion\pdfminorversion=7\fi
\usepackage[margin=1.1in]{geometry}
\usepackage[T1]{fontenc}
\usepackage[utf8]{inputenc}
\usepackage{lmodern}
\usepackage{microtype}
\usepackage{amsmath,amssymb,amsfonts,amsthm,mathtools}
\usepackage{bm}
\usepackage{booktabs,tabularx,array}
\usepackage{graphicx}
\graphicspath{{figures/}}
\usepackage{float}
\usepackage{cite,xurl}
\usepackage{enumitem}
\usepackage{xcolor}
\usepackage[colorlinks=true,linkcolor=blue!50!black,citecolor=blue!50!black,urlcolor=blue!50!black]{hyperref}
\usepackage[capitalise,nameinlink,noabbrev]{cleveref}
\usepackage{aliascnt}
\allowdisplaybreaks
\newtheorem{theorem}{Theorem}[section]
\newaliascnt{proposition}{theorem}
\newtheorem{proposition}[proposition]{Proposition}
\aliascntresetthe{proposition}
\newaliascnt{corollary}{theorem}
\newtheorem{corollary}[corollary]{Corollary}
\aliascntresetthe{corollary}
\newaliascnt{lemma}{theorem}
\newtheorem{lemma}[lemma]{Lemma}
\aliascntresetthe{lemma}
\theoremstyle{definition}
\newaliascnt{definition}{theorem}
\newtheorem{definition}[definition]{Definition}
\aliascntresetthe{definition}
\newaliascnt{assumption}{theorem}
\newtheorem{assumption}[assumption]{Assumption}
\aliascntresetthe{assumption}
\newaliascnt{remark}{theorem}
\newtheorem{remark}[remark]{Remark}
\aliascntresetthe{remark}
\newaliascnt{example}{theorem}
\newtheorem{example}[example]{Example}
\aliascntresetthe{example}
\newtheorem*{uremark}{Remark}

\crefname{theorem}{Theorem}{Theorems}
\Crefname{theorem}{Theorem}{Theorems}
\crefname{proposition}{Proposition}{Propositions}
\Crefname{proposition}{Proposition}{Propositions}
\crefname{corollary}{Corollary}{Corollaries}
\Crefname{corollary}{Corollary}{Corollaries}
\crefname{lemma}{Lemma}{Lemmas}
\Crefname{lemma}{Lemma}{Lemmas}
\crefname{definition}{Definition}{Definitions}
\Crefname{definition}{Definition}{Definitions}
\crefname{assumption}{Assumption}{Assumptions}
\Crefname{assumption}{Assumption}{Assumptions}
\crefname{remark}{Remark}{Remarks}
\Crefname{remark}{Remark}{Remarks}
\crefname{example}{Example}{Examples}
\Crefname{example}{Example}{Examples}

\newcommand{\R}{\mathbb{R}}
\newcommand{\E}{\mathbb{E}}
\newcommand{\Prob}{\mathbb{P}}
\newcommand{\Law}{\mathcal{L}}
\newcommand{\1}{\mathbf{1}}
\newcommand{\Wone}{\mathcal{W}_1}
\newcommand{\dd}{\mathrm{d}}
\newcommand{\supp}{\mathrm{supp}}

\newcommand{\ind}{\mathrm{ind}}

\newcommand{\norm}[1]{\left\lVert #1 \right\rVert}
\newcommand{\abs}[1]{\left\lvert #1 \right\rvert}

\title[Low-rank and graphon limits for distress contagion]{Low-rank and graphon limits for dynamic threshold distress contagion in heterogeneous financial networks}
\author{Pengbin Feng}\thanks{An earlier version of the model appeared in the author's Ph.D. dissertation \cite{feng2021thesis}. The author thanks Professors Jin Ma and Jianfeng Zhang for their guidance on his Ph.D. thesis and on the development of this work, and Professor Zimu Zhu for helpful discussions.}
\email{pengbinf@alumni.usc.edu}
\date{}
\hypersetup{
  pdftitle={Low-rank and graphon limits for dynamic threshold distress contagion in heterogeneous financial networks},
  pdfauthor={Pengbin Feng},
  pdfsubject={Low-rank and graphon limits for distress contagion},
  pdfkeywords={systemic risk, recoverable distress, financial networks, occupation times, low-rank approximation, directed graphons, Osgood condition}
}

\begin{document}

\begin{abstract}
We study how financial distress propagates through a directed network when institutions can recover and become distressed again. Each institution's current buffer determines its loss transmission rate, so the duration of distress affects its creditors' cumulative losses. We show that the dynamics of a network with $K$ sender--receiver factor pairs can be represented exactly by a closed system of $K$ feedback equations. For bounded Lipschitz loss rates, we prove stability with respect to population and exposure changes, using Wasserstein distance for type distributions and aligned norms for kernels. These estimates quantify the error in approximating a directed graphon model by finite factor networks. Hard thresholds introduce a further difficulty: when many institutions lie close to the threshold, the same initial data can admit several contagion paths. For bounded nonnegative kernels and bounded initial profiles, we construct the greatest cumulative-distress solution and show that vanishing positive-side regularization selects it. We prove uniqueness, quantitative stability, and convergence of sampled nonnegative networks under an Osgood condition on the mass near the threshold along a reference solution. Rank-one examples establish sharpness among uniqueness guarantees based only on this mass. Numerical experiments illustrate how recovery interrupts new loss transmission while past losses persist. Using the disclosed sovereign portfolios of 116 banks under prescribed stress scenarios, we also examine how the choice of retained channels affects approximation accuracy. In the baseline scenario, retaining the two most strongly shocked channels gives a maximum buffer-path error about one fifteenth that of a twelve-mode singular-value truncation.
\end{abstract}



\maketitle

\section{Introduction}

When an institution interrupts payments, its creditors lose income and their financial buffers deteriorate. Some creditors may then enter distress and transmit losses to other institutions. Recovery can interrupt this process: an institution that resumes payments stops generating new losses, although the losses it has already transmitted remain in its creditors' buffers. To describe the resulting contagion, we must keep track of who is exposed to whom and how long each institution remains distressed.

Many financial networks organize these exposures through a small number of channels. Tiered interbank markets, central clearing, and common-asset holdings give institutions distinct roles in transmitting and receiving losses \cite{boss2004,craig2014,aldasoro2017,veraart2025,bcbs2025nbfi}. Such structure suggests that a few aggregate variables may suffice to describe the dynamics of a large network. We ask when this reduction is exact and how closely a reduced model reproduces institutional buffer paths when its population or exposures only approximate those of the original network.

We formulate a model in which an institution's current buffer determines the rate at which it transmits losses. Buffers continue to evolve during distress, allowing an institution to recover and later become distressed again. The mechanism is already visible in the homogeneous system
\begin{equation}
\label{eq:intro-meanfield}
X_t^{i,N}=x_i^N+\mu t-\frac1N\sum_{j=1}^N\int_0^t\ell(X_s^{j,N})\,\dd s,
\qquad 1\le i\le N,
\end{equation}
where $x_i^N$ is an initial buffer, $\mu$ is an outside-book drift, and $\ell$ is a bounded nonnegative loss rate. With $\ell(x)=\mathbf1_{\{x\le0\}}$, the integral records the total time counterparty $j$ has spent in distress. Each additional episode therefore adds to the losses borne by its creditors; recovery ends the current episode without reversing earlier losses.

In this homogeneous system, every institution receives the same cumulative loss, and one aggregate feedback path determines all buffer trajectories. We extend this description to heterogeneous channels by assigning each institution two loadings per channel: a sender loading weights its contribution to aggregate distress, and a receiver loading measures its exposure to that aggregate. We show that, when the exposure matrix is a sum of $K$ such factor pairs, the dynamics close exactly through $K$ feedback equations. Each institution's buffer can then be reconstructed from these feedback paths, its initial buffer, and its loadings. For a general directed kernel, we approximate the exposure structure by factors and the population by finitely many institutions. These two approximations introduce distinct sources of error.

We quantify both errors when the loss rate is bounded and Lipschitz. First, we compare populations whose institutions differ in their initial buffers and sender--receiver loadings. We bound the distance between their evolving state--factor distributions by the Wasserstein distance between their initial type distributions. Under uniform factor bounds and the $1/K$ normalization, the stability constant does not depend on $K$. Second, we keep institutional labels fixed and bound the $L^1$ distance between two evolving state profiles, uniformly over time, by the $L^1$ differences between their kernels and initial buffer profiles. Combining these estimates gives a finite-network-to-graphon bound for admissible bounded-factor approximations. Keeping the labels aligned ensures that we compare the exposures and buffers of the same institutional types.

A hard distress threshold creates an additional difficulty. If many institutions have buffers close to zero, a small loss can move enough of them into distress to generate a further round of losses. Even a single channel can then admit several responses to the same initial data. In a continuum population $u\in[0,1]$, take identical exposures, no outside drift, and initial buffers $x_0(u)=u^2$. A fraction $\sqrt a$ of the population lies within distance $a$ of zero. The aggregate cumulative distress satisfies $\dot\beta=\sqrt\beta$ while $\beta\le1$. For any delay $\tau\ge0$, the path can remain at zero until $\tau$ and then follow $\beta(t)=(t-\tau)^2/4$ for $\tau\le t\le\tau+2$. All these paths satisfy the same closed-threshold rule with the same initial data. An exact factor reduction therefore preserves the dynamics, but need not make their response unique.

We first resolve the selection question for nonnegative exposures. For every bounded nonnegative kernel and bounded initial profile, we construct a greatest cumulative-distress solution. We then show that positive-side ramps converge to this solution as their width vanishes. Thus this regularization specifies a particular response even when the hard-threshold system admits several paths.

We next identify conditions under which every hard solution must coincide with a given reference path. Two nearby state paths can disagree about an institution's distress status only when its reference buffer lies near zero. We use the mass of these institutions to bound the additional losses caused by a perturbation. An Osgood condition on this mass yields uniqueness and quantitative stability, as well as convergence of sampled nonnegative networks. Rank-one examples show that the condition is sharp among uniqueness guarantees based only on threshold mass. When the mass grows at most linearly with the width of the threshold interval, we obtain optimal power-law bounds for deterministic perturbations and a quantitative sampling bound. We also give conditions on initial buffers and kernel sections that allow this regularity to be checked from the model inputs.

Our model draws on the balance-sheet mechanisms developed in static clearing and cascade models \cite{eisenberg2001,cont2010,amini2016,elliott2014,glasserman2015}. These models describe bilateral loss transmission and terminal outcomes. Dynamic approaches include diffusion models of interbank borrowing and lending \cite{carmona2015}, heterogeneous default-time contagion \cite{nadtochiy2020,feinstein2021}, and dynamic clearing \cite{banerjee2025}. Feinstein and S{\o}jmark \cite{feinsteinsojmark2023} study heterogeneous impact and exposure with absorbing first-passage losses. We instead accumulate losses during each institution's current distress episodes. Allowing recovery changes both the financial interpretation of the loss process and the questions of solution selection and threshold regularity that we must address.

To compare populations, we use optimal transport to couple institutions through their initial buffers and sender--receiver loadings \cite{villani2009}. We then propagate this coupling through the dynamics using Dobrushin's method \cite{dobrushin1979}. Feng and Peng \cite{fengpeng2014} study structural properties of optimal couplings in the Monge--Kantorovich problem.

Our kernel formulation is related to work on graphon population limits and concentration estimates \cite{bayraktar2023,bayraktar2022,bayraktar2023b,allmeier2025,coppini2025}, graphon estimation \cite{bayraktar2025np}, and graphon control and games \cite{gao2020,parise2018,caineshuang2020,amini2023graphon}. Related threshold mechanisms arise in irreversible graphon contagion \cite{erol2023} and cascades on inhomogeneous random networks \cite{detering2019,amini2024}. In our comparisons, we retain the institutional meaning of each label and measure exposure errors in aligned kernel norms. For hard thresholds, we also track how much of the population approaches distress as the system evolves.

\Cref{sec:model} derives the model from balance-sheet flows, using the stopped default system as a benchmark for recoverable distress. \Cref{sec:finite-rank-limit,sec:graphon} develop the factor limit and kernel theory, respectively. \Cref{sec:financial-interpretations,sec:numerics} give financial interpretations and numerical illustrations. These include a recovery-and-relapse example and common-exposure factors constructed from the disclosed sovereign holdings of 116 banks under prescribed stress scenarios. We collect the proofs in \cref{app:proofs} and auxiliary representations and approximation results in \cref{app:theory-supplement}.

\section{Dynamic contagion model and low-rank factorization}
\label{sec:model}

We derive the buffer dynamics from contractual payments and losses on counterparty claims. We first consider a single clearing date, then allow debt service to be impaired over time, and finally identify when the network dynamics can be represented by a small system of feedback equations.

\subsection{Balance-sheet derivation and the stopped benchmark}
\label{sec:derivation}

An institution receives payments on its interbank assets and owes payments on its interbank liabilities. The difference contributes to its buffer, while counterparty distress reduces the receipts on the affected claims.

In the bilateral interpretation, $e_{ij}\ge0$ is the raw amount owed by institution $j$ to institution $i$, with $e_{ii}=0$. Thus row $i$ records interbank assets, column $i$ records interbank liabilities, and their totals are $\sum_j e_{ij}$ and $\sum_j e_{ji}$. The initial quantity $x_i$ is net wealth outside the interbank book in the static cascade. For the dynamic models, it is the initial buffer, whose subsequent changes are measured against contractual debt-service flows. Section~\ref{sec:finite-network} also gives the diagonal convention for common-exposure channels.

\smallskip
\emph{One clearing date: the static cascade.} Let $\theta_0\in[0,1)$ be the recovery rate and $\theta:=1-\theta_0$ the loss rate. When institution $j$ defaults, its creditor $i$ receives $\theta_0e_{ij}$, while the full obligations owed \emph{by} $i$ are deducted from $i$'s pre-resolution buffer. We take $\mathbb D_0:=\{i:x_i\le0\}$ as the initial resolution set, using the outside-book buffer to identify the initial defaults. For the remaining institutions, we recalculate the buffer from initial outside-book wealth, the net interbank position, and losses on defaulted counterparties. After each update, institutions with nonpositive buffers enter resolution:
\begin{equation}
\label{eq:cascade}
X_i^{(k+1)}
:=x_i+\sum_{j=1}^N\bigl(e_{ij}-e_{ji}\bigr)-\theta\sum_{j=1}^N e_{ij}\,\1_{\mathbb D_k}(j)
\quad\text{for } i\notin\mathbb D_0,
\qquad
\mathbb D_{k+1}:=\bigl\{i:X_i^{(k+1)}\le0\bigr\},
\end{equation}
with $X_i^{(k+1)}:=x_i$ for $i\in\mathbb D_0$. The initially resolved institutions retain their original buffers. Institutions that default in later rounds continue to recognize losses on their counterparties; their own default, in turn, causes losses for their creditors. The default sets grow monotonically and the buffers decrease from the first round onward. Every strict change adds at least one institution to the default set, so this set stabilizes by round $N$. One further buffer update gives $X^{(N+1)}$, the \emph{greatest} solution of the fixed-point system
\begin{equation}
\label{eq:cascade-fixedpoint}
X_i=x_i+\1_{\{x_i>0\}}\Bigl[\sum_{j=1}^N\bigl(e_{ij}-e_{ji}\bigr)-\theta\sum_{j=1}^N e_{ij}\,\1_{\{X_j\le0\}}\Bigr],
\qquad 1\le i\le N;
\end{equation}
the short induction argument is recorded in \cref{prop:cascade-maxsol}. This fixed-recovery cascade follows the balance-sheet loss mechanism studied in the static literature \cite{cont2010,amini2016}.

\smallskip
\emph{From one date to a horizon.} Now consider the same exposures over $[0,T]$ and let $r>0$ be the interest rate. We take all interbank maturities beyond $T$, so the interbank cash flows during the horizon are debt-service flows. A counterparty above the threshold pays the contractual interest on the position $e_{ij}$ in full. A counterparty at or below the threshold pays the recovery fraction $\theta_0$ of contractual debt service. Its creditor's interest income per unit time is therefore
\[
r\,e_{ij}\,\1_{\{X^j>0\}}+\theta_0\,r\,e_{ij}\,\1_{\{X^j\le0\}}
\;=\;
r\,e_{ij}-\theta\,r\,e_{ij}\,\1_{\{X^j\le0\}}.
\]
The creditor therefore loses $\theta r e_{ij}$ per unit time that the counterparty spends in distress. We calculate its buffer from actual receipts, full contractual liabilities, and the outside-book drift $\mu_i$. An institution's own failure to pay does not reduce the obligations charged to its buffer. The state therefore measures its capacity to meet contractual payments, while distress at its counterparties lowers the receipts available to do so.

A right-endpoint approximation to these cash flows over $[t,t+\Delta t]$ is
\begin{equation}
\label{eq:window-update}
X_{t+\Delta t}^{i}
=X_t^{i}
+\Delta t\Bigl[\mu_i+r\sum_{j=1}^N\bigl(e_{ij}-e_{ji}\bigr)\Bigr]
-\theta r\,\Delta t\sum_{j=1}^N e_{ij}\,\1_{\{X_{t+\Delta t}^{j}\le0\}}.
\end{equation}
The buffer changes through outside-book income, net contractual debt service, and losses on distressed counterparties. In continuous time, the cumulative receipt shortfall on the claim $e_{ij}$ is $\theta r e_{ij}\int_0^t\mathbf1_{\{X_s^j\le0\}}\,\dd s$. To obtain an absorbing benchmark, we integrate these flows only until the institution first crosses the threshold:
\begin{equation}
\label{eq:mechanism-ode}
X_t^i=x_i+\int_0^{\tau_i\wedge t}\Bigl[\mu_i+r\sum_{j=1}^N\bigl(e_{ij}-e_{ji}\bigr)-\theta r\sum_{j=1}^N e_{ij}\,\1_{\{X_s^j\le0\}}\Bigr]\dd s,
\qquad
\tau_i:=\inf\{t\ge0:X_t^i\le0\}\wedge T.
\end{equation}
Between successive default times, the set of impaired counterparties is fixed and the solvent buffers evolve affinely. At a new threshold crossing, we add the affected institutions to this set and freeze their buffers. This procedure yields the unique event-driven solution in \cref{prop:mechanism-wellposed}. If we instead let buffers continue to evolve after a crossing, institutions can recover and experience further episodes of distress under the same payment rules.

\smallskip
\emph{The recoverable transmission model.} We use dense scaling \(e_{ij}=e_{ij}^N/N\), with uniformly bounded \(e_{ij}^N\), to keep each institution's aggregate exposure of order one. We describe external buffer growth by a common outside-book drift \(\mu\) and the loss generated by a counterparty by a bounded rate \(\ell:\R\to[0,\ell_\ast]\). The debt-service specialization is \(\ell(x)=\theta r\mathbf 1_{\{x\le0\}}\); we also allow the loss intensity to be specified independently of the interest rate, including the unit hard rule. Each buffer continues to evolve throughout the horizon, so an institution can recover when its net flow improves. Under the hard loss rule, it stops transmitting losses upon recovery and resumes if it becomes distressed again. With these choices, the dynamics take the form \eqref{eq:finite-network}. Under homogeneous weights \(e_{ij}^N\equiv1\), incoming and outgoing exposures balance for every institution, and the equation reduces to \eqref{eq:intro-meanfield}.

\subsection{The finite network}
\label{sec:finite-network}

Consider $N$ institutions whose buffers $(X_t^{i,N})_{1\le i\le N}$ are measured relative to a distress threshold of zero. The dense-normalized weight $e_{ij}^N\ge0$ measures the exposure to sender $j$ borne by receiver $i$, and $\ell:\R\to[0,\ell_\ast]$ specifies the current loss rate as a function of the sender's buffer. Adding external drift, net contractual flows, and accumulated counterparty losses as in \cref{sec:derivation}, we obtain
\begin{equation}
\label{eq:finite-network}
X_t^{i,N} = x_i^N + \mu t + \frac{r t}{N}\sum_{j=1}^N (e_{ij}^N-e_{ji}^N) - \frac{1}{N}\sum_{j=1}^N e_{ij}^N \int_0^t \ell(X_s^{j,N})\,\dd s,
\end{equation}
for \(1\le i\le N\). The common drift \(\mu\in\R\) represents the outside book, while \(r\in\R\) scales the effect of the net exposure position. Under the debt-service interpretation with \(r>0\), net interbank assets contribute positively to the contractual flow and net liabilities contribute negatively. The final term integrates each counterparty's current loss rate and subtracts the exposure-weighted total from the receiver's buffer.

We specify the diagonal according to the financial channel being modeled. Direct bilateral liabilities use \(e_{ii}^N=0\). In common-exposure models, a positive diagonal records an institution's own contribution to the shared channel. These diagonal terms cancel when we compute the exposure imbalance, but remain in the cumulative-loss term. We retain them when the factor representation includes this contribution to the common channel.

For the hard rule \(\ell(x)=\mathbf 1_{\{x\le0\}}\), we use the greatest cumulative-distress path, whose existence follows from nonnegativity of the weights. As the width of the positive-side ramp tends to zero, its solutions converge to this path. The associated state is the smallest among the hard solutions; see \cref{cor:indicator-canonical-sample,thm:indicator-greatest}. The finite-population comparison in \cref{thm:indicator-finiteN} covers every Borel choice of hard solution under its sampling and density assumptions.

Under the hard rule, we distinguish the institutions currently in distress from the time they have already spent in distress. The instantaneous threshold fraction, \(N^{-1}\sum_i\mathbf 1_{\{X_t^{i,N}\le0\}}\), can rise as institutions enter distress and fall as they recover. Integrating each institution's indicator records its accumulated time in distress, which stops increasing during recovery and resumes if distress returns. Counterparties incur losses over this accumulated time; the continuum counterpart is the occupation profile \(H_t\) in \eqref{eq:Ht-definition}.

\subsection{Rank-\texorpdfstring{$K$}{K} exposure matrices}
\label{sec:rankK}

Fix $K\in\mathbb{N}$. We represent the exposure structure through $K$ pairs of sender and receiver loadings, so that its algebraic rank is at most $K$. Assume that
\begin{equation}
\label{eq:rankK-factorization}
e_{ij}^N = \frac{1}{K}\sum_{k=1}^K a_{i,k}^N b_{j,k}^N.
\end{equation}
Here $b_{j,k}^N$ measures the contribution of institution $j$ to channel $k$, and $a_{i,k}^N$ measures institution $i$'s exposure to the losses carried by that channel. To compute the loss transmitted through a channel, we first weight counterparties' loss rates by their sender loadings. We then allocate the aggregate to exposed institutions according to their receiver loadings. An institution's contribution to a channel and its exposure to that channel need not coincide.

\begin{example}[From mean field to two directed channels]
\label{ex:meanfield-to-rank2}
With $K=1$ and $a_{i,1}^N=b_{j,1}^N=1$, every institution receives the same aggregate loss and \eqref{eq:finite-network} reduces to \eqref{eq:intro-meanfield}. With general rank-one loadings, institutions still interact through a single aggregate, but they may contribute to it and bear its losses in different proportions.

For two-way transmission, partition the institutions into a core $C$ and a periphery $P$. Take $K=2$ and nonnegative constants $\gamma_1,\gamma_2,\delta_1,\delta_2$, with
\[
a_{i,1}^N=\gamma_1\1_{\{i\in C\}},\quad b_{j,1}^N=\delta_1\1_{\{j\in P\}},
\qquad
a_{i,2}^N=\gamma_2\1_{\{i\in P\}},\quad b_{j,2}^N=\delta_2\1_{\{j\in C\}}.
\]
Then $e_{ij}^N=\frac12(a_{i,1}^Nb_{j,1}^N+a_{i,2}^Nb_{j,2}^N)$. Through the first channel, distressed peripheral institutions impose losses on the core. Through the second, distressed core institutions impose losses on the periphery. We therefore need only two aggregate loss paths to describe the reciprocal feedback, regardless of how many institutions each group contains.
\end{example}

Nonnegative factors give nonnegative exposures directly; signed factors are also admissible when their sum in \eqref{eq:rankK-factorization} is entrywise nonnegative. The bounded-Lipschitz theory covers uniformly bounded signed factors. For hard losses, nonnegativity of the resulting exposures supplies the canonical solution, and the stated threshold-mass conditions control uniqueness and stability. \Cref{sec:financial-interpretations} develops the financial interpretations of interbank tiers, clearing, and common exposures.

Introduce the type vector
\[
z_i^N=(x_i^N,a_{i,1}^N,\dots,a_{i,K}^N,b_{i,1}^N,\dots,b_{i,K}^N)\in\mathcal Z:=\R^{2K+1},
\]
and the empirical type law
\[
\mu_0^N := \frac{1}{N}\sum_{i=1}^N \delta_{z_i^N}.
\]
We write a generic point of $\mathcal Z$ as $z=(x,a,b)$ with $a=(a_1,\dots,a_K)$ and $b=(b_1,\dots,b_K)$. The norm on $\mathcal Z$ is the $\ell^1$ norm,
\[
|z| := |x| + \sum_{k=1}^K |a_k| + \sum_{k=1}^K |b_k|.
\]
\subsection{The rank-\texorpdfstring{$K$}{K} feedback system}
\label{sec:unified-system}

We describe a population by the joint distribution of its initial buffers and loadings. We compute the exposure imbalance from the mean sender and receiver loadings; at each time, we obtain the channel feedback rates by averaging current loss rates with the sender loadings. The same equations apply to a finite empirical population and its limiting type distribution. Fix a probability measure $\varrho\in\mathcal P(\mathcal Z)$ with integrable factor coordinates and define the means
\[
\bar a_k^{\varrho}:=\int a_k\,\varrho(\dd z), \qquad \bar b_k^{\varrho}:=\int b_k\,\varrho(\dd z),
\]
and the imbalance coefficient
\begin{equation}
\label{eq:Rlimit}
\Lambda^{\varrho}(z):=\frac{r}{K}\sum_{k=1}^K \bigl(a_k\bar b_k^{\varrho}-b_k\bar a_k^{\varrho}\bigr),
\qquad z=(x,a,b)\in\mathcal Z.
\end{equation}
Given a deterministic vector $\beta=(\beta_1,\dots,\beta_K)$ of feedback variables, the associated state map is
\begin{equation}
\label{eq:Xlimit-map}
X_t^{\varrho}(z)=x+\mu t+t\Lambda^{\varrho}(z)-\frac{1}{K}\sum_{k=1}^K a_k\beta_k(t),
\end{equation}
and the system is closed by requiring
\begin{equation}
\label{eq:beta-limit-correct}
\beta_k(t)=\int_0^t m_k(s)\,\dd s,
\qquad
m_k(t)=\int b_k\,\ell(X_t^{\varrho}(z))\,\varrho(\dd z).
\end{equation}
We use $\varrho=\mu_0^N$ for the empirical system, writing $\bar a_k^N,\bar b_k^N,\Lambda^N,X_t^N,\beta_k^N,m_k^N$, and $\varrho=\mu_0$ for the large-population system, where the superscript is omitted. The following lemma identifies the empirical system exactly with the finite network:

\begin{lemma}[Exact reformulation of the finite network]
\label{lem:reformulation}
Let $\ell$ be bounded and measurable and take $\varrho=\mu_0^N$. If $\beta^N$ solves the closed system \eqref{eq:Xlimit-map}--\eqref{eq:beta-limit-correct}, then $X_t^{i,N}:=X_t^N(z_i^N)$, $1\le i\le N$, solves the finite-network equation \eqref{eq:finite-network}. Conversely, if $(X^{i,N})_{1\le i\le N}$ solves \eqref{eq:finite-network}, then the vector $\beta^N$ defined from these paths by \eqref{eq:beta-limit-correct} satisfies the closed system, and $X_t^{i,N}=X_t^N(z_i^N)$.
\end{lemma}

\begin{proof}
Integrals against $\mu_0^N$ are averages over the atoms $z_1^N,\dots,z_N^N$. Substituting the factorization \eqref{eq:rankK-factorization} into the contagion term of \eqref{eq:finite-network} gives
\[
\frac1N\sum_{j=1}^N e_{ij}^N\int_0^t \ell(X_s^{j,N})\,\dd s
=\frac1K\sum_{k=1}^K a_{i,k}^N\,\Bigl(\frac1N\sum_{j=1}^N b_{j,k}^N\int_0^t \ell(X_s^{j,N})\,\dd s\Bigr)
=\frac1K\sum_{k=1}^K a_{i,k}^N\beta_k^N(t),
\]
where the last equality holds once $X_s^{j,N}=X_s^N(z_j^N)$ and $\beta^N$ is given by \eqref{eq:beta-limit-correct}. Similarly, $\frac1N\sum_j e_{ij}^N=\frac1K\sum_k a_{i,k}^N\bar b_k^N$ and $\frac1N\sum_j e_{ji}^N=\frac1K\sum_k b_{i,k}^N\bar a_k^N$, so the imbalance term of \eqref{eq:finite-network} equals $t\Lambda^N(z_i^N)$. Hence \eqref{eq:finite-network} for the family $(X^{i,N})_i$ is identical, atom by atom, to \eqref{eq:Xlimit-map}--\eqref{eq:beta-limit-correct} evaluated at $z=z_i^N$; reading this identity in the two directions proves both claims.
\end{proof}

Once the $K$ feedback paths are known, we recover each institution's state from its initial buffer and loadings using \eqref{eq:Xlimit-map}. The feedback system has dimension $K$ for every population size $N$.

We can also evaluate the network feedback without forming the full exposure matrix. If $q_j(t):=\ell(X_t^{j,N})$ and $E^N=K^{-1}AB^\top$, then the network feedback is
\[
E^Nq(t)=\frac1K A\bigl(B^\top q(t)\bigr).
\]
Evaluating the empirical right-hand side and storing the factors cost $O(NK)$ operations and memory, compared with $O(N^2)$ for a dense matrix. Further grouping or quadrature can reduce the cost of the empirical average.

With $\varrho=\mu_0$, the integrals average over the limiting population rather than the finite empirical distribution. Comparing the dynamics under these two type laws accounts for the population error in \cref{thm:bridge}.

\paragraph{Notation.}
The scalar $\mu$ is a drift; $\mu_0$ and $\mu_0^N$ are probability measures on the type space $\mathcal Z$. We write $\nu_t$ for a state law on $\R$ and $\widehat\nu_t$ for the joint law of the state and its static factors; a superscript $N$ denotes the corresponding empirical law. Thus $\nu_t$ is the state marginal of $\widehat\nu_t$.

For probability laws $\rho,\widetilde\rho$ with finite first moments, the Wasserstein distance is
\[
\Wone(\rho,\widetilde\rho)
=\inf_{\pi\in\Gamma(\rho,\widetilde\rho)}
\int d(z,\widetilde z)\,\pi(\dd z,\dd\widetilde z),
\]
where $\Gamma(\rho,\widetilde\rho)$ is the set of couplings with these marginals. The cost $d$ is the $\ell^1$ type distance for joint laws and absolute distance for state laws. The infimum is attained; see Villani \cite[Chapters~4 and~6]{villani2009} for this formulation and its basic properties. The functions $m_k$ are instantaneous weighted loss rates and $\beta_k$ their time integrals. In the kernel formulation, $u$ is the receiving label, $v$ the sending label, and $H_t(u)$ is the cumulative time spent in distress.

\section{The finite-rank generalized mean-field limit}
\label{sec:finite-rank-limit}

Fix $K\in\mathbb N$. Throughout this section the factor loadings satisfy the following bound; assumptions on the initial state are stated separately.

\begin{assumption}
\label{ass:bounded-types}
There exists $M>0$ such that for every $N$ and every $1\le k\le K$,
\[
|a_k|\le M, \qquad |b_k|\le M
\]
for $\mu_0^N$-almost every $z=(x,a,b)$ and for $\mu_0$-almost every $z=(x,a,b)$.
\end{assumption}

\subsection{Reduction to a \texorpdfstring{$K$}{K}-dimensional feedback system}

Once the cumulative loss in each channel is known, the state map determines every institutional buffer. The same feedback equation applies to an empirical population and to a general type distribution; passing to the population model changes the distribution over which losses are averaged.

\begin{proposition}[Finite-dimensional feedback equation]
\label{prop:beta-ode}
Under \cref{ass:bounded-types}, let $\ell:\R\to[0,\ell_\ast]$ be bounded and measurable. The limit system \eqref{eq:Xlimit-map}--\eqref{eq:beta-limit-correct} is equivalent to
\begin{equation}
\label{eq:feedback-ode}
\beta_k(t)=\int_0^t F_k(s,\beta(s))\,\dd s, \qquad 1\le k\le K,
\end{equation}
where
\begin{equation}
\label{eq:Fk}
F_k(t,\beta):=\int b_k\,\ell\!\left(x+\mu t+t\Lambda(z)-\frac{1}{K}\sum_{j=1}^K a_j\beta_j\right)\mu_0(\dd z).
\end{equation}
Every solution is absolutely continuous and satisfies the nonautonomous ODE $\dot\beta_k(t)=F_k(t,\beta(t))$ for almost every $t$, with $\beta_k(0)=0$. The same representation holds for the finite-$N$ system after replacing $\mu_0$ and $\Lambda$ by $\mu_0^N$ and $\Lambda^N$. If $\ell$ is Lipschitz, then $F$ is globally Lipschitz in $\beta$, uniformly on compact time intervals, and the rank-$K$ system is globally well posed.
\end{proposition}

\begin{proof}
Substitution of the state map into the feedback integral gives \eqref{eq:feedback-ode}; conversely, a solution of this equation defines $X$ by \eqref{eq:Xlimit-map}. If $\ell$ is Lipschitz with constant $L_\ell$, then for any $\beta,\tilde\beta\in\R^K$,
\[
\abs{F_k(t,\beta)-F_k(t,\tilde\beta)}
\le \frac{L_\ell}{K}\int |b_k|\sum_{j=1}^K |a_j|\,|\beta_j-\tilde\beta_j|\,\mu_0(\dd z)
\le M^2L_\ell\norm{\beta-\tilde\beta}_\infty.
\]
The field is measurable in time and satisfies $|F_k(t,\beta)|\le M\ell_\ast$. The Carath\'eodory theorem gives a unique solution, and this uniform bound extends it to every finite time horizon.
\end{proof}

\subsection{Wasserstein stability for bounded Lipschitz losses}

We compare two populations with different distributions of initial buffers and loadings. Following Dobrushin's approach \cite{dobrushin1979}, we couple their initial types and follow the coupled states through the dynamics. For bounded Lipschitz losses, the next theorem bounds the distance between the evolving state--factor laws by the distance between the initial type laws. This is the estimate used for the sampling error in the graphon approximation theorem.

\begin{theorem}
\label{thm:finite-rank-w1}
Suppose \cref{ass:bounded-types} holds, $\mu_0^N,\mu_0\in\mathcal P_1(\mathcal Z)$, and $\ell:\R\to[0,\ell_\ast]$ is Lipschitz with constant $L_\ell$. Then for every $T>0$ there exists $C_T<\infty$, depending only on $T,M,|r|,\ell_\ast$, and $L_\ell$, such that
\begin{equation}
\label{eq:w1-main}
\sup_{0\le t\le T} \Wone(\widehat\nu_t^N,\widehat\nu_t) \le C_T\,\Wone(\mu_0^N,\mu_0),
\end{equation}
where
\[
\widehat\nu_t^N := (\Xi_t^N)_\#\mu_0^N,
\qquad
\widehat\nu_t := (\Xi_t)_\#\mu_0,
\]
with
\[
\Xi_t^N(z):=(X_t^N(z),a,b), \qquad \Xi_t(z):=(X_t(z),a,b).
\]
Consequently, if $\mu_0^N\to\mu_0$ in $\Wone$, the joint state--factor laws converge in $\Wone$ uniformly on compact time intervals. The estimate applies to deterministic and sampled type laws.
\end{theorem}

\begin{proof}[Proof sketch]
Couple the initial type laws optimally and compare their state maps. The factor means and loadings contribute a multiple of the initial Wasserstein distance. The Lipschitz loss bounds the feedback difference by this distance and the coupled state error, so Gronwall's lemma closes the estimate. Adding the unchanged factor coordinates gives the joint-law bound; see \cref{app:proofs-sec3}.
\end{proof}

\begin{remark}[Sampling and rank normalization]

\label{rem:sampled-types}
The $1/K$ normalization makes $C_T$ independent of $K$ under the uniform factor bounds. For i.i.d.\ types with law $\mu_0\in\mathcal P_1(\mathcal Z)$, the empirical laws converge almost surely in $\Wone$, by weak convergence and the strong law for first moments. The state--factor laws therefore converge almost surely as well. Finite first moments suffice here; the transport-density representation in \cref{thm:transport-pde} uses the stronger assumption of a compactly supported $C^1$ joint density.
\end{remark}

\subsection{Discontinuous threshold losses}

Small changes in feedback can move buffers across zero and change individual loss rates discontinuously. To bound the resulting change in aggregate losses, we impose a uniform density bound on the buffer distributions generated by bounded feedback inputs. Averaging the indicator then gives a Lipschitz vector field. In \cref{sec:graphon}, we obtain a different comparison from an Osgood tube along a single reference solution.

\begin{assumption}
\label{ass:density}
For every $T>0$ and $C>0$, there exists $M_{\rho}(T,C)<\infty$ such that for every $t\in[0,T]$ and every $\beta\in[-C,C]^K$, the scalar random variable
\[
\Psi_t(z,\beta):=x+\mu t+t\Lambda(z)-\frac{1}{K}\sum_{k=1}^K a_k\beta_k
\]
under $\mu_0$ admits a density bounded by $M_{\rho}(T,C)$.
\end{assumption}

\begin{remark}[Sufficient conditions for \cref{ass:density}]
\label{rem:density-sufficient}
We can verify \cref{ass:density} by conditioning on the loadings: for fixed loadings, feedback only translates the initial buffer. Let $\rho_{A,B}:=\Law(A,B)$. Suppose that under \(\mu_0\) the initial state \(X_0\) admits a regular conditional density \(f_{X_0\mid A,B}(\cdot\mid a,b)\) given \((A,B)=(a,b)\) such that
\begin{equation}
\label{eq:conditional-density}
\operatorname*{ess\,sup}_{(a,b)\sim\rho_{A,B}}
\norm{f_{X_0\mid A,B}(\cdot\mid a,b)}_{L^\infty(\R)}
\le \bar M<\infty.
\end{equation}
Then \cref{ass:density} holds with \(M_\rho(T,C)=\bar M\) for every \(T,C>0\). For fixed $t$ and $\beta$, write
\[
\Psi_t=X_0+G_t(A,B;\beta),
\qquad
G_t(a,b;\beta):=\mu t+t\Lambda(a,b)-\frac1K\sum_{k=1}^K a_k\beta_k.
\]
Conditional on $(A,B)=(a,b)$, the variable $\Psi_t$ is a translation of $X_0$ and therefore has density $y\mapsto f_{X_0\mid A,B}(y-G_t(a,b;\beta)\mid a,b)$. Averaging over $(A,B)$ preserves the bound $\bar M$. The same argument remains uniform when the population means entering $\Lambda$ are replaced by deterministic vectors $(\alpha,\gamma)\in[-M,M]^{2K}$, because these parameters only translate the conditional law. This uniform form is used in \cref{thm:indicator-finiteN}. Condition \eqref{eq:conditional-density} holds, in particular, in each of the following cases:
\begin{enumerate}[label={\rm(\roman*)},leftmargin=2.2em]
\item \(X_0\) is independent of \((A,B)\) and has bounded density;
\item \(X_0\) is lognormal, independently of \((A,B)\), or conditionally lognormal with uniformly bounded conditional densities;
\item $X_0$ is conditionally Pareto with parameters $x_m(a,b)>0$ and $\alpha(a,b)>0$ satisfying $\operatorname*{ess\,sup}_{a,b}\alpha(a,b)/x_m(a,b)<\infty$; the conditional density is bounded by this ratio.
\end{enumerate}
Wasserstein estimates additionally require $\E|X_0|<\infty$; for a fixed Pareto law this is equivalent to $\alpha>1$. Translating the initial state preserves the density bound.
\end{remark}

\begin{theorem}
\label{thm:indicator}
Assume \cref{ass:bounded-types,ass:density} and let $\ell(x)=\1_{\{x\le 0\}}$. Then the limiting rank-$K$ system \eqref{eq:Xlimit-map}--\eqref{eq:beta-limit-correct} is well posed.
\end{theorem}

\begin{proof}[Proof sketch]
Any solution is confined a priori to the cube $\mathcal C_T=[-MT,MT]^K$. On this cube the indicator feedback map is Lipschitz: if two feedback vectors differ by $\delta$ in the sup norm, the corresponding indicators can disagree only on a threshold tube of width $M\delta$, whose $\mu_0$-mass is at most $2M_\rho(T,MT)M\delta$ by \cref{ass:density}. Carath\'eodory theory for the feedback ODE then yields global existence and uniqueness.
\end{proof}

Each empirical population is atomic, so we compare its threshold averages directly with those of the sampling law. The uniform empirical-process estimate in \cref{lem:indicator-vc} bounds this discrepancy over the feedback cube. Using that estimate, the next theorem compares any Borel choice of finite hard solution with the unique population path.

\begin{theorem}[Fixed-rank finite-$N$ convergence estimate for selected indicator solutions]
\label{thm:indicator-finiteN}
Assume \cref{ass:bounded-types} and the conditional-density condition \eqref{eq:conditional-density} of \cref{rem:density-sufficient}. Let $N\ge2$ and let $Z_1,\dots,Z_N$ be i.i.d.\ with law $\mu_0$, define
\[
\mu_0^N:=\frac1N\sum_{i=1}^N \delta_{Z_i},
\qquad
\bar a_k^N:=\int a_k\,\dd\mu_0^N,
\qquad
\bar b_k^N:=\int b_k\,\dd\mu_0^N,
\]
and, for $\beta\in\mathcal C_T=[-MT,MT]^K$, set
\[
\Lambda^N(z):=\frac{r}{K}\sum_{k=1}^K (a_k\bar b_k^N-b_k\bar a_k^N),
\qquad
\Psi_t^N(z,\beta):=x+\mu t+t\Lambda^N(z)-\frac1K\sum_{k=1}^K a_k\beta_k,
\]
\[
F_k^N(t,\beta):=\int b_k\,\1_{\{\Psi_t^N(z,\beta)\le 0\}}\,\mu_0^N(\dd z),
\qquad 1\le k\le K.
\]
Let $\beta$ be the unique solution of the deterministic limit equation from \cref{thm:indicator}. Let $S_N:\mathcal Z^N\to C([0,T];\mathcal C_T)$ be Borel measurable and set $\beta^N=S_N(Z_1,\ldots,Z_N)$. Assume that, for $\mu_0^{\otimes N}$-almost every sample,
\[
\beta_k^N(t)=\int_0^t F_k^N(s,\beta^N(s))\,\dd s,
\qquad 1\le k\le K.
\]
Then $\beta^N$ is absolutely continuous, and there exists $C_T<\infty$, depending only on $T$, $M$, $|r|$, and $\bar M$, such that
\[
\E\Bigl[\sup_{0\le t\le T} \norm{\beta^N(t)-\beta(t)}_\infty\Bigr]
\le C_T\sqrt{\frac{K\log N}{N}}.
\]
If, in addition, $\E|X_0|<\infty$, then
\[
\E\Bigl[\sup_{0\le t\le T} \Wone(\widehat\nu_t^N,\widehat\nu_t)\Bigr]
\le C_T\sqrt{\frac{K\log N}{N}} + C_T\E\Wone(\mu_0^N,\mu_0),
\]
where $\widehat\nu_t^N:=(\Xi_t^N)_\#\mu_0^N$, $\widehat\nu_t:=(\Xi_t)_\#\mu_0$, and
\begin{align*}
\Xi_t^N(z)&:=\Bigl(x+\mu t+t\Lambda^N(z)-\frac1K\sum_{k=1}^K a_k\beta_k^N(t),a,b\Bigr),\\
\Xi_t(z)&:=\Bigl(x+\mu t+t\Lambda(z)-\frac1K\sum_{k=1}^K a_k\beta_k(t),a,b\Bigr).
\end{align*}
\end{theorem}

\begin{proof}[Proof sketch]
Separate the feedback-field error into a uniform empirical-process term and a perturbation of the factor means. The first is controlled by \cref{lem:indicator-vc}; the second uses the conditional density and concentration of the sample means. Gronwall's lemma gives the feedback-path estimate. Coupling the state maps on the sampled atoms then gives the state-law bound; see \cref{app:proofs-sec3}.
\end{proof}

\begin{corollary}[Canonical Borel selection for nonnegative sampled networks]
\label{cor:indicator-canonical-sample}
In the setting of \cref{thm:indicator-finiteN}, assume that the induced sampled exposure matrix
\[
e_{ij}^N=\frac1K\sum_{k=1}^K a_{i,k}^N b_{j,k}^N
\]
is entrywise nonnegative for every sample outside a null set. Then the sampled hard-threshold system has a greatest cumulative-distress solution $H^{N,+}$ and a smallest state solution $X^{N,+}$ on that event. Define $H^{N,+}$ on the Borel set where the factors satisfy the bound $M$ and the exposure matrix is nonnegative, and extend it by zero elsewhere. The map from $(Z_1,\ldots,Z_N)$ to $H^{N,+}\in C([0,T];\R^N)$ is Borel measurable, as is
\[
\beta_k^{N,+}(t):=\frac1N\sum_{j=1}^N b_{j,k}^N H_j^{N,+}(t),
\qquad 1\le k\le K.
\]
This feedback path satisfies the sampled equation in \cref{thm:indicator-finiteN}. Consequently the estimates of that theorem hold for this canonical solution. Nonnegativity is required of the exposure matrix, allowing signed factors.
\end{corollary}

\begin{proof}[Proof sketch]
On the finite probability space, the indicator feedback map is order preserving. Iteration from the top profile $H_i^{(0)}(t)=t$ decreases to a hard fixed point that dominates every other fixed point. The same path is the full vanishing limit of the positive-side ramp. Finite upper iterates are Borel functions of the sample, and their monotone convergence is uniform in time, which makes the limit a Borel $C([0,T];\R^N)$-valued map. Exact low-rank reformulation then gives the displayed feedback path.
\end{proof}

\begin{remark}[Canonical selections via smoothing]
\label{rem:indicator-selection}
For nonnegative exposures, the canonical selected path in \cref{cor:indicator-canonical-sample} is obtained by replacing $\1_{\{x\le 0\}}$ with the positive-side Lipschitz regularizations
\begin{equation}
\label{eq:elleps}
\ell_\varepsilon(x):=
\begin{cases}
1, & x\le 0,\\
1-x/\varepsilon, & 0<x<\varepsilon,\\
0, & x\ge \varepsilon,
\end{cases}
\qquad \varepsilon>0,
\end{equation}
solving the regularized sampled ODE, and letting $\varepsilon\downarrow0$. By \cref{thm:indicator-positive-ramp}, the cumulative-distress profiles decrease to the greatest hard solution, while the state profiles increase to its associated state. For each fixed width, sampled stability follows from \cref{cor:regularized-sampled}; a uniform density bound gives the $O(\varepsilon)$ smoothing estimate in \cref{prop:graphon-indicator-smoothing}. The distinction between positive-side and negative-side ramps is discussed in \cref{rem:indicator-two-ramps}.
\end{remark}

\begin{remark}[Rate separation and the low-rank advantage]
\label{rem:rate-separation}
For fixed $K$, \cref{thm:indicator-finiteN} bounds the expected uniform feedback error by a constant multiple of $\sqrt{K\log N/N}$. Comparing the full state laws also requires $\E\Wone(\mu_0^N,\mu_0)$, which measures how well the sampled population represents the type distribution on $\R^{2K+1}$. Its rate depends on dimension and tail behavior; under a moment condition $\E|Z|^q<\infty$ for some $q>1$, Fournier and Guillin \cite{fournier2015} give quantitative bounds. Thus estimating the $K$ aggregate feedback coordinates and estimating the full type distribution contribute separate errors. Under the conditional-density assumption, every Borel sampled selection has the same population limit; for nonnegative sampled exposures, we use the canonical selection in \cref{cor:indicator-canonical-sample}.
\end{remark}

\section{The directed-kernel (labeled-graphon) formulation}
\label{sec:graphon}

A general directed kernel describes exposure as a function of the receiving and sending types. Finite-factor models are special cases, and we compare their trajectories with those of a prescribed kernel as the number of factors grows. Throughout this section, $I=[0,1]$ carries Lebesgue measure $\lambda$, and a graphon means a bounded directed kernel on $I^2$. All $L^p$ comparisons keep the receiving and sending labels fixed, so an exposure error compares the same institutional types in the two systems. \Cref{tab:result-map} lists the assumptions used for Lipschitz-loss approximation and for hard-threshold existence, selection, and stability.

\begin{table}[t]
\centering
\caption{The two approximation regimes and the role of threshold regularity.}
\label{tab:result-map}
\small
\renewcommand{\arraystretch}{1.15}
\begin{tabularx}{\textwidth}{>{\raggedright\arraybackslash}p{0.24\textwidth}>{\raggedright\arraybackslash}X}
\toprule
Question & Conditions and conclusion\\
\midrule
Population and exposure error & Bounded Lipschitz loss; uniform factor bounds for the sampled low-rank bridge; Wasserstein and aligned $L^1$ estimates (\cref{thm:bridge}).\\
Hard-solution selection & Bounded nonnegative kernel; upper iteration and positive-side ramps select the greatest cumulative-distress path (\cref{thm:indicator-greatest,thm:indicator-positive-ramp}).\\
Uniqueness and sensitivity & An existing reference hard solution with an Osgood tube; uniqueness and aligned $L^p$ comparison, also for signed kernels (\cref{prop:indicator-osgood,thm:indicator-osgood-stability}).\\
Sampling & Bounded nonnegative kernel and an Osgood-regular target; convergence of sampled latent-label networks (\cref{thm:indicator-sampled-label}).\\
\bottomrule
\end{tabularx}
\end{table}

\subsection{Directed graphon contagion equation}

A kernel \(W\in L^\infty(I^2)\) represents the exposure of type \(u\) to counterparty type \(v\), with the same orientation as $e_{ij}^N$ in \eqref{eq:finite-network}. Financial exposures are nonnegative; bounded signed kernels also enter the analysis through spectral approximations.

The data are a deterministic pair $(W,x_0)$ with $x_0\in L^\infty(I)$. We use the same letter for the kernel and its integral operator,
\[
(Wf)(u):=\int_I W(u,v)f(v)\,\dd v,
\qquad
\norm{Wf}_{L^\infty}\le\norm W_{L^\infty}\norm f_{L^1}.
\]
Superscripts $W^{(K)}$ denote finite-rank approximants, and $\Pi_KW$ denotes block averaging. The Wasserstein distance remains $\Wone$. Define
\begin{equation}
\label{eq:graphon-R}
R_W(u):=\int_I \bigl(W(u,v)-W(v,u)\bigr)\,\dd v,
\end{equation}
and the nonlinear operator
\begin{equation}
\label{eq:gammaW}
(\Gamma_W \varphi)(u):=\int_I W(u,v)\,\ell(\varphi(v))\,\dd v.
\end{equation}
The graphon contagion equation reads
\begin{equation}
\label{eq:graphon-equation}
X_t(u)=x_0(u)+\mu t + r t R_W(u)-\int_0^t (\Gamma_W X_s)(u)\,\dd s.
\end{equation}
The row $W(u,\cdot)$ determines losses received by type $u$; subtracting the column integral in $R_W$ gives its net contractual position. Equation~\eqref{eq:graphon-equation} is the continuum version of \eqref{eq:finite-network}. As in the finite system, $\lambda\{u:X_t(u)\le0\}$ is the fraction currently below the threshold. It differs from the cumulative occupation time $H_t(u)$ defined below.

\begin{uremark}[Deterministic data and sampling]
Equation~\eqref{eq:graphon-equation} is deterministic: $W$ and $x_0$ are fixed measurable functions, and $\nu_t=\Law(X_t(U))$ for $U\sim\mathrm{Unif}(I)$ is the pushforward of Lebesgue measure. Random initial capital can be incorporated by augmenting the latent type, and sampled models are introduced in \cref{rem:sampled-types,thm:indicator-finiteN,thm:indicator-sampled-label,cor:indicator-sampled-bridge}.
\end{uremark}

\subsection{Population and exposure approximation}

Changing the kernel alters both the contractual imbalance flow and the losses transmitted during distress. For a bounded Lipschitz loss, the next theorem bounds the resulting state-path discrepancy in aligned $L^1$ by the kernel and initial-buffer discrepancies.

\begin{theorem}
\label{thm:graphon-wellposed}
Let $\ell:\R\to[0,\ell_\ast]$ be bounded and Lipschitz with constant $L_\ell$, and let $W\in L^\infty(I^2)$ and $x_0\in L^\infty(I)$. Then for every $T>0$ the graphon equation \eqref{eq:graphon-equation} has a unique solution
\[
X\in C([0,T];L^\infty(I)).
\]
Moreover, if $(W_1,x_0^1)$ and $(W_2,x_0^2)$ generate solutions $X^1$ and $X^2$, and if
\[
\norm{W_1}_{L^\infty} \le M_W,
\qquad
\norm{W_2}_{L^\infty} \le M_W,
\]
then
\begin{equation}
\label{eq:graphon-stability}
\sup_{0\le t\le T}\norm{X_t^1-X_t^2}_{L^1}
\le C_T\Bigl(\norm{x_0^1-x_0^2}_{L^1}+\norm{W_1-W_2}_{L^1(I^2)}\Bigr),
\end{equation}
for a constant $C_T$ depending only on $T,|r|,\ell_\ast,L_\ell$, and $M_W$.
\end{theorem}

\begin{proof}[Proof sketch]
Well-posedness follows from a Picard iteration in $C([0,T];L^\infty(I))$, because $\Gamma_W$ is Lipschitz from $L^1$, hence from $L^\infty$, into $L^\infty$ with constant $\norm{W}_{L^\infty}L_\ell$. For the stability estimate, subtract the two equations, use $\norm{R_{W_1}-R_{W_2}}_{L^1}\le 2\norm{W_1-W_2}_{L^1(I^2)}$ and $\norm{\Gamma_{W_1}\varphi-\Gamma_{W_2}\psi}_{L^1}\le \ell_\ast\norm{W_1-W_2}_{L^1(I^2)}+M_WL_\ell\norm{\varphi-\psi}_{L^1}$, and apply Gronwall's lemma in $L^1$.
\end{proof}

To compare a finite network with the target graphon, we insert the continuum model with the same factors as the finite network. The type-law estimate of \cref{thm:finite-rank-w1} bounds the first comparison, while the preceding kernel estimate bounds the second. The resulting error bound therefore separates population sampling from approximation of the continuum exposures and buffers. The uniform factor bounds under the $1/K$ normalization ensure that the constants remain uniform as the number of factors grows.

\begin{theorem}[Low-rank reduction principle for admissible bounded-factor families]
\label{thm:bridge}
Let \(\ell:\R\to[0,\ell_\ast]\) be bounded and Lipschitz. Let \(W\in L^\infty(I^2)\) and \(x_0\in L^\infty(I)\). For each \(K\in\mathbb N\), suppose that the following assumptions hold.

\begin{enumerate}[label={\rm(A\arabic*)},leftmargin=2.5em]
\item \emph{Bounded-factor approximation.} There are coefficient maps \(a_k^{(K)},b_k^{(K)}\in L^\infty(I)\) such that
\[
W^{(K)}(u,v)=\frac{1}{K}\sum_{k=1}^K a_k^{(K)}(u)b_k^{(K)}(v).
\]

\item \emph{Uniform boundedness.} For some \(M<\infty\),
\[
\norm{W}_{L^\infty}
+
\sup_{K\ge1}
\left(
\norm{x_0^{(K)}}_{L^\infty}
+\max_{1\le k\le K}\norm{a_k^{(K)}}_{L^\infty}
+\max_{1\le k\le K}\norm{b_k^{(K)}}_{L^\infty}
\right)
\le M.
\]
The factor bounds imply \(\norm{W^{(K)}}_{L^\infty}\le M^2\) after enlarging the constant if necessary, so the approximating kernels are uniformly bounded.

\item \emph{Continuum approximation.}
\[
W^{(K)}\to W\quad\text{in }L^1(I^2),
\qquad
x_0^{(K)}\to x_0\quad\text{in }L^1(I).
\]
\end{enumerate}

Let \(U\sim{\rm Unif}(I)\) and define the rank-\(K\) type law
\[
\mu_0^{(K)}
:=
\Law\bigl(x_0^{(K)}(U),a^{(K)}(U),b^{(K)}(U)\bigr),
\]
where \(a^{(K)}=(a_1^{(K)},\dots,a_K^{(K)})\) and \(b^{(K)}=(b_1^{(K)},\dots,b_K^{(K)})\). Let \(\mu_0^{N,K}\) be empirical measures on \(\R^{2K+1}\) satisfying
\[
\supp(\mu_0^{N,K})\subset[-M,M]^{2K+1}.
\]
Denote by \(\nu_t^{N,K}\) the law of the state component in the finite-\(N\) rank-\(K\) system associated with \(\mu_0^{N,K}\), by
\[
\nu_t^K:=\Law(X_t^K(U))
\]
the state law of the rank-\(K\) graphon model with kernel \(W^{(K)}\) and profile \(x_0^{(K)}\), and by
\[
\nu_t:=\Law(X_t(U))
\]
the state law of the target graphon model with data \((W,x_0)\). Then for every \(T>0\) there exists \(C_T<\infty\), depending only on \(T,r,\ell_\ast,L_\ell\), and \(M\), such that
\begin{equation}
\label{eq:bridge-bound}
\sup_{0\le t\le T}\Wone(\nu_t^{N,K},\nu_t)
\le C_T\Wone(\mu_0^{N,K},\mu_0^{(K)})
+C_T\left(
\norm{W^{(K)}-W}_{L^1(I^2)}
+\norm{x_0^{(K)}-x_0}_{L^1(I)}
\right).
\end{equation}
\end{theorem}
\begin{proof}[Proof sketch]
Apply \cref{thm:finite-rank-w1} to the pair $(\mu_0^{N,K},\mu_0^{(K)})$, which controls the finite-population sampling error; couple the rank-$K$ and target graphon state laws through the common uniform variable $U$ and invoke the stability estimate of \cref{thm:graphon-wellposed}, which controls the truncation error; conclude by the triangle inequality.
\end{proof}

Two constructions are useful here. Block averages approximate any bounded kernel by bounded finite-rank kernels and preserve nonnegativity. Singular-value truncation gives an explicit $L^2$ tail error and can introduce negative entries. Applying \cref{thm:bridge} also requires uniform factor bounds, which do not follow from the tail estimate. We give these constructions and their assumptions in \cref{prop:block-L1-approx,rem:svd,rem:bridge-factor-gap}. The feedback representation also extends to a continuous collection of factors, as shown in \cref{prop:factorized-kernel}.

\subsection{Hard thresholds: existence and selection}
\label{sec:graphon-indicator-beyond-factorization}

We construct a hard-threshold solution without a density assumption. For a nonnegative kernel, increasing cumulative distress lowers buffers and can only increase the subsequent distress indicators. The integral feedback map is therefore order preserving. Iterating it from the largest admissible cumulative profile gives the greatest cumulative-distress solution.

For $t\in[0,T]$ and $u\in I$, define the cumulative distress profile
\begin{equation}
\label{eq:Ht-definition}
H_t(u):=\int_0^t \1_{\{X_s(u)\le 0\}}\,\dd s.
\end{equation}
Then the indicator-loss graphon equation \eqref{eq:graphon-equation} is equivalent to
\begin{equation}
\label{eq:graphon-indicator-H}
X_t(u)=x_0(u)+\mu t+r tR_W(u)-\int_I W(u,v)H_t(v)\,\dd v.
\end{equation}
Introducing the threshold functional
\begin{equation}
\label{eq:Psi-graphon-general}
\Psi_t(u,h):=x_0(u)+\mu t+r tR_W(u)-\int_I W(u,v)h(v)\,\dd v,
\end{equation}
gives the integral feedback equation
\begin{equation}
\label{eq:H-feedback-general}
H_t(u)=\int_0^t \1_{\{\Psi_s(u,H_s)\le 0\}}\,\dd s.
\end{equation}
This extends the finite-rank feedback equation \eqref{eq:beta-limit-correct} to general kernels.

A \emph{hard solution} means a jointly measurable pair satisfying \eqref{eq:graphon-indicator-H}--\eqref{eq:H-feedback-general} for every $t$, outside a common null set of labels. Its paths are absolutely continuous in time; the differential equation is required only almost everywhere in time. An admissible cumulative profile has $H_0=0$, $0\le H_t\le t$, and pointwise $1$-Lipschitz paths. For such a profile, write
\[
(\mathcal T_WH)_t(u):=\int_0^t
\1_{\{\Psi_s(u,H_s)\le0\}}\,\dd s.
\]
The iteration below starts from the largest admissible cumulative profile and converges to a cumulative-distress solution that dominates every other hard cumulative-distress profile. It also applies when the threshold law has atoms.

\begin{theorem}[Greatest hard-indicator solution]
\label{thm:indicator-greatest}
Let $W\in L^\infty(I^2)$ satisfy $W\ge0$ a.e. and let $x_0\in L^\infty(I)$. For every $T>0$, \eqref{eq:H-feedback-general} has a greatest solution $H^+$: every other hard solution $H$ satisfies
\[
H_t(u)\le H_t^+(u)
\quad\text{for a.e. }u\text{ and every }t\in[0,T].
\]
For each comparison, jointly measurable representatives satisfy this order for every $t$ on one full-measure subset of $I$.
It is obtained from the upper iteration
\begin{equation}
\label{eq:indicator-upper-iteration}
H_t^{(0)}(u):=t,
\qquad H^{(n+1)}:=\mathcal T_WH^{(n)}.
\end{equation}
The iterates decrease a.e. to $H^+$ and converge in $C([0,T];L^p(I))$ for every $1\le p<\infty$. The associated state
\[
X_t^+=x_0+\mu t+rtR_W-WH_t^+
\]
belongs to $C([0,T];L^\infty(I))$ and is the smallest state profile among all hard solutions.
\end{theorem}

The iteration uses the non-strict threshold convention: $a_n\uparrow a$ implies $\1_{\{a_n\le0\}}\downarrow\1_{\{a\le0\}}$. The positive-side ramp preserves this selection.

\begin{theorem}[The positive-side ramp selects the greatest solution]
\label{thm:indicator-positive-ramp}
Under the hypotheses of \cref{thm:indicator-greatest}, let $H^{\varepsilon,+}$ be the unique cumulative profile for the ramp $\ell_\varepsilon$ in \eqref{eq:elleps}. If $0<\varepsilon'<\varepsilon$, then
\[
H^{\varepsilon',+}\le H^{\varepsilon,+}.
\]
Moreover, as $\varepsilon\downarrow0$,
\[
H^{\varepsilon,+}\downarrow H^+
\quad\text{a.e.},
\qquad
H^{\varepsilon,+}\longrightarrow H^+
\quad\text{in }C([0,T];L^p(I)),\quad 1\le p<\infty,
\]
and the states converge in $C([0,T];L^\infty(I))$.
\end{theorem}

\begin{remark}[Ramp conventions]
\label{rem:indicator-two-ramps}
At the threshold, the ramp value distinguishes two selection mechanisms. The positive-side ramp has $\ell_\varepsilon(0)=1$ and selects $H^+$. The negative-side ramp
\[
\ell^-_\varepsilon(x)=\max\{0,\min\{1,-x/\varepsilon\}\},
\]
used in \cref{sec:real-data}, is zero at the threshold and converges pointwise to $\1_{\{x<0\}}$. The selection mechanism is visible with atomless state laws: for $W=1$, $x_0(u)=u^2$, and $\mu=r=0$, the negative-side dynamics stay at $H=0$, while the positive-side ramp selects the nonzero greatest solution described below. At atoms, a negative-side limit need not solve the closed-threshold equation: with $W=0$, $x_0=0$, and $\mu=r=0$, the hard solution has $H_t=t$, while the negative-side ramp gives $H_t=0$. Under the Osgood tube condition of \cref{thm:indicator-osgood-stability}, both ramps converge to the unique hard solution: the comparison in \cref{cor:indicator-osgood-bridge} applies to either ramp because its discrepancy from the indicator is confined to $|x|\le\varepsilon$.
\end{remark}

The next examples separate existence, uniqueness, and sensitivity to the initial buffers. They also show why the sign of the kernel matters for constructing a hard solution.

\begin{proposition}[Threshold counterexamples]
\label{prop:indicator-obstructions}
\begin{enumerate}[label={\rm(\roman*)},leftmargin=2.2em]
\item A nonnegative rank-one kernel can have several hard solutions.
\item A unique hard solution can depend discontinuously on $x_0$ in $L^\infty$ when the target has threshold mass.
\item An atomless initial law can give rise to multiple hard solutions.
\item A bounded signed rank-one kernel may have no hard-indicator solution.
\end{enumerate}
\end{proposition}

\begin{proof}[Proof sketch]
For (i), take $W\equiv w>0$, $x_0=0$, $r=0$, and $0<\mu<w$: both $H=0$ and $H_t=t$ solve the equation. For (ii), take $W\equiv w>0$, $x_0=0$, $\mu=r=0$ and perturb $x_0$ to $1/n$. For (iii), take $W=1$, $x_0(u)=u^2$, and $\mu=r=0$, which reduces to $\dot\beta=\sqrt\beta$ and admits delayed solutions. For (iv), take $W=-1$, $x_0=0$, $r=0$, and $\mu=-1/2$; the scalar equation switches with incompatible velocities on the two sides of zero.
\end{proof}

\subsection{Threshold mass and uniqueness}

We compare an arbitrary hard solution with a fixed reference solution. A discrepancy between their distress indicators can occur only where the reference buffer is within the state displacement of zero. Bounding the mass of this set gives a scalar inequality for the cumulative-distress error. Under the Osgood condition, an error that starts at zero must remain zero. This comparison requires an existing reference solution and also applies to signed kernels.

\begin{proposition}[Uniqueness under a threshold-tube bound]
\label{prop:indicator-osgood}
Let $(H,X)$ be one hard-indicator solution for a bounded kernel $W$ and put $M_W:=\norm{W}_{L^\infty}$. Suppose that, for a continuous nondecreasing function $\omega:[0,\infty)\to[0,1]$ with $\omega(0)=0$ and $\omega(a)>0$ for $a>0$,
\[
\sup_{0\le t\le T}\lambda\{u:\abs{X_t(u)}\le a\}\le\omega(a),
\qquad a\ge0,
\]
and
\[
\int_{0+}\frac{\dd s}{\omega(M_Ws)}=\infty,
\]
with the usual interpretation if $M_W=0$. Then $(H,X)$ is the unique hard solution for these data. In particular, the linear threshold-tube bound
\begin{equation}
\label{eq:indicator-reference-tube}
\sup_{0\le t\le T}\lambda\{u:\abs{X_t(u)}\le a\}\le La,
\qquad a\ge0,
\end{equation}
implies uniqueness.
\end{proposition}

\begin{proposition}[Sharpness of the Osgood condition]
\label{prop:indicator-osgood-sharp}
Fix $c>0$ and let $F$ be the continuous distribution function of a bounded nonnegative profile $x_0$, with $F(0)=0$ and $F(a)>0$ for $a>0$ near zero. For $W\equiv c$ and $\mu=r=0$, the hard system is unique if and only if
\begin{equation}
\label{eq:indicator-rank-one-osgood-iff}
\int_{0+}\frac{\dd y}{F(cy)}=\infty.
\end{equation}
If the integral is finite, it has a continuum of delayed solutions. Every continuous nondecreasing tube modulus that fails the Osgood integral is realized near zero by a bounded nonnegative rank-one example with multiple solutions. Thus the integral condition is sharp for uniqueness guarantees based on threshold-tube mass.
\end{proposition}

For example, near zero,
\[
\omega(a)=Ca\bigl(\log(e/a)\bigr)^\gamma
\]
satisfies the Osgood condition exactly when $\gamma\le1$, whereas $\omega(a)=Ca^\alpha$ with $0<\alpha<1$ does not.

\begin{remark}[Threshold concentration and atoms]
\label{rem:regulatory-atoms}

Regulatory or accounting bunching can create an atom or a narrow concentration near the distress threshold. An atom at the distress threshold violates the tube condition with $\omega(0)=0$, and a narrow concentration enlarges its modulus. Multiple solutions can then occur, but failure of the condition does not imply nonuniqueness. For $W\equiv0$, $x_0\equiv0$, and $\mu=r=0$, the hard solution is uniquely $X_t\equiv0$, $H_t\equiv t$, although its state law is $\delta_0$ at every time. Positive-side regularization still selects the greatest solution for nonnegative kernels; the negative-side convention has the different behavior described in \cref{rem:indicator-two-ramps}.
\end{remark}

\subsection{Perturbation and deterministic approximation}

We compare systems with different exposures and initial buffers. For finite $p$, a small $L^p$ error need not be small for every institution. We bound separately the contributions to the loss discrepancy from institutions with large data errors and from institutions near the reference threshold. Optimizing this separation gives the effective modulus below.

For the perturbation results below, let $(H^1,X^1)$ and $(H^2,X^2)$ be hard solutions generated by $(W_1,x_0^1)$ and $(W_2,x_0^2)$ with the same $\mu,r$, put $M_1:=\norm{W_1}_{L^\infty}$, and define
\begin{equation}
\label{eq:indicator-delta-p}
\delta_p:=\norm{x_0^1-x_0^2}_{L^p(I)}
+T(1+2\abs r)\norm{W_1-W_2}_{L^p(I^2)}.
\end{equation}
The comparison uses
\begin{align}
\label{eq:indicator-perturbation-split}
d(t)&:=\norm{H_t^1-H_t^2}_{L^1},\qquad
X_t^1-X_t^2=-W_1(H_t^1-H_t^2)+A_t,\\
\label{eq:indicator-data-error}
A_t&:=(x_0^1-x_0^2)+rt(R_{W_1}-R_{W_2})-(W_1-W_2)H_t^2.
\end{align}
Here $\sup_{t\le T}\norm{A_t}_{L^p}\le\delta_p$ by Jensen's inequality, while $\norm{W_1(H_t^1-H_t^2)}_{L^\infty}\le M_1d(t)$. Thus the data error is controlled in $L^p$, and the feedback error is controlled uniformly in the label.

\begin{theorem}[Stability under an Osgood tube]
\label{thm:indicator-osgood-stability}
Assume that the reference solution satisfies
\begin{equation}
\label{eq:indicator-reference-osgood-tube}
\sup_{t\le T}\lambda\{\abs{X_t^1}\le a\}\le\omega(a),
\qquad a\ge0,
\end{equation}
for a modulus $\omega$ with the continuity, monotonicity, and positivity properties in \cref{prop:indicator-osgood}, satisfying the Osgood divergence. For $1\le p<\infty$ define
\begin{align}
\label{eq:indicator-effective-modulus}
\Omega_{p,\delta}(z)
&:=1\wedge\inf_{\eta>0}
\left\{\omega(z+\eta)+\left(\frac{\delta}{\eta}\right)^p\right\},\\
\label{eq:indicator-effective-primitive}
\Phi_{p,\delta}(a)
&:=\int_0^a\frac{\dd s}{\Omega_{p,\delta}(s)}.
\end{align}
For $p=\infty$, set $\Omega_{\infty,\delta}(z):=\omega(z+\delta)$ and $\Phi_{\infty,\delta}(a):=\int_0^a\omega(s+\delta)^{-1}\,\dd s$, and in both cases use
\[
\Phi_{p,\delta}^{-1}(y):=\inf\{a\ge0:\Phi_{p,\delta}(a)\ge y\}.
\]
If $M_1>0$, then, for every solution of the second hard system,
\begin{align}
\label{eq:indicator-osgood-H}
M_1\sup_{t\le T}\norm{H_t^1-H_t^2}_{L^1}
&\le\Phi_{p,\delta_p}^{-1}(M_1T),\\
\label{eq:indicator-osgood-X}
\sup_{t\le T}\norm{X_t^1-X_t^2}_{L^p}
&\le\delta_p+\Phi_{p,\delta_p}^{-1}(M_1T),\\
\label{eq:indicator-osgood-q}
\sup_{t\le T}\norm{\1_{\{X_t^1\le0\}}-\1_{\{X_t^2\le0\}}}_{L^1}
&\le\Omega_{p,\delta_p}\!\left(\Phi_{p,\delta_p}^{-1}(M_1T)\right).
\end{align}
When $M_1=0$, the corresponding estimates are
\[
\sup_{t\le T}\norm{H_t^1-H_t^2}_{L^1}\le T\Omega_{p,\delta_p}(0),
\qquad
\sup_{t\le T}\norm{X_t^1-X_t^2}_{L^p}\le\delta_p,
\]
and the indicator mismatch is at most $\Omega_{p,\delta_p}(0)$. All right-hand sides vanish as $\delta_p\to0$.
\end{theorem}

The bound accounts for two sources of buffer displacement: the change in the data and the change in cumulative distress generated by feedback. An institution's indicator can differ between the two systems only if its reference buffer is within this displacement of zero. Applying the tube bound to those institutions gives the scalar differential inequality used in the estimate.

For a linear tube, optimizing the split between large data errors and buffers near zero yields a power law. With finite $p$, data errors can be concentrated on a small set, giving the exponent $p/(p+1)$; uniformly small data errors give a linear estimate.

\begin{theorem}[Aligned $L^p$ stability under a linear tube]
\label{thm:indicator-one-sided-stability}
Assume that the reference solution $(H^1,X^1)$ satisfies \eqref{eq:indicator-reference-tube}, and let $M_1$ and $\delta_p$ be as above.
For $1\le p<\infty$, with $\theta_p=p/(p+1)$, there is $C_{T,p}<\infty$, depending only on $T,p,L,M_1$, such that
\begin{align}
\label{eq:indicator-one-sided-H}
\sup_{t\le T}\norm{H_t^1-H_t^2}_{L^1(I)}
&\le C_{T,p}\delta_p^{\theta_p},\\
\label{eq:indicator-one-sided-X}
\sup_{t\le T}\norm{X_t^1-X_t^2}_{L^p(I)}
&\le \delta_p+C_{T,p}\delta_p^{\theta_p},\\
\label{eq:indicator-one-sided-q}
\sup_{t\le T}\norm{\1_{\{X_t^1\le0\}}-\1_{\{X_t^2\le0\}}}_{L^1(I)}
&\le C_{T,p}\delta_p^{\theta_p}.
\end{align}
For $p=\infty$, the state error is measured in $L^\infty$ and all three bounds are linear in $\delta_\infty$. The exponent $p/(p+1)$ is optimal for each of the three errors at every finite $p$.
\end{theorem}

We apply the perturbation estimate in two ways. For nonnegative approximants, \cref{thm:indicator-greatest} supplies hard solutions to compare with the target. For signed approximants, we solve the Lipschitz-ramp equation and let its width vanish together with the data error. Both comparisons use threshold regularity only along the target hard solution.

\begin{corollary}[Convergence under an Osgood tube]
\label{cor:indicator-osgood-bridge}
Suppose the target solution for $(W,x_0)$ satisfies the tube and Osgood assumptions of \cref{thm:indicator-osgood-stability}. First let $W^{(n)}\ge0$ be bounded kernels and let $x_0^{(n)}$ be bounded profiles. If, for some $1\le p\le\infty$, $(W^{(n)},x_0^{(n)})\to(W,x_0)$ in aligned $L^p(I^2)\times L^p(I)$, then every hard solution of the approximating systems converges to the unique target solution with the modulus in \cref{thm:indicator-osgood-stability}; in particular this holds for their canonical greatest solutions and for equal-block representatives of deterministic finite nonnegative networks.

Alternatively, let $W^{(n)}$ be bounded signed kernels and let $x_0^{(n)}$ be bounded profiles. For $\varepsilon_n>0$, denote their unique positive-side-ramp dynamics by
\[
(H^{n,\varepsilon_n},X^{n,\varepsilon_n}).
\]
Let $\varepsilon_n\downarrow0$. For some fixed $1\le p\le\infty$, define
\[
\delta_{n,p}:=\norm{x_0^{(n)}-x_0}_{L^p(I)}
+T(1+2\abs r)\norm{W^{(n)}-W}_{L^p(I^2)},
\]
and, for finite $p$,
\[
\Omega_{p,\delta,\varepsilon}(z)
:=1\wedge\inf_{\eta>0}
\left\{\omega(z+\eta+\varepsilon)+\left(\frac{\delta}{\eta}\right)^p\right\}.
\]
For $p=\infty$, put $\Omega_{\infty,\delta,\varepsilon}(z):=\omega(z+\delta+\varepsilon)$. Using this modulus in \cref{thm:indicator-osgood-stability} gives convergence whenever $\delta_{n,p}\to0$ and $\varepsilon_n\to0$.
\end{corollary}

\begin{proof}
Nonnegativity supplies the approximating hard solutions through \cref{thm:indicator-greatest}, and \cref{thm:indicator-osgood-stability} gives the comparison. For the ramp, use \eqref{eq:indicator-perturbation-split}--\eqref{eq:indicator-data-error} with the target as system~1 and the regularized approximant as system~2. The loss discrepancy is supported where
\[
\abs{X_t}\le \norm W_{L^\infty}d(t)+\abs{A_t}+\varepsilon_n,
\]
which gives $\Omega_{p,\delta,\varepsilon}$. Fatou's lemma and the divergent Osgood integral make the inverse bound tend to zero. Equal-block embedding preserves the finite equations exactly.
\end{proof}

\Cref{cor:indicator-explicit-log-osgood,cor:indicator-strong-bridge} give explicit log--Osgood moduli and the corresponding rates for linear tubes, including signed spectral approximations with vanishing ramp width.

\subsection{Sampling a nonnegative kernel}

For sampled networks, the institution labels are random, and we must also estimate the errors in row averages and threshold counts. We use bounded nonnegative kernels to obtain finite greatest solutions, then compare their cumulative distress with the target path. Concentration estimates bound the sampling terms, while the target's tube modulus controls the feedback error through a scalar comparison equation.

\begin{theorem}[Sampled latent-label convergence under an Osgood tube]
\label{thm:indicator-sampled-label}
Fix bounded Borel representatives $W:I^2\to[0,M_W]$ and $x_0:I\to\R$, and jointly Borel representatives of a target hard solution $(H,X)$ whose paths are Lipschitz in time for almost every label. Assume that its tube is bounded as in \eqref{eq:indicator-reference-osgood-tube}, with a modulus satisfying the hypotheses of \cref{thm:indicator-osgood-stability}. Let $U_1,\ldots,U_N$ be independent uniform labels and set
\[
e_{ij}^N:=W(U_i,U_j)\quad(i\ne j),
\qquad e_{ii}^N:=0,
\qquad x_i^N:=x_0(U_i).
\]
Let $(H^{N,+},X^{N,+})$ be the greatest finite solution and define
\[
D_N(t):=\frac1N\sum_{i=1}^N\abs{H_i^{N,+}(t)-H_t(U_i)}.
\]
For $\kappa>0$, let $\mathcal R_\kappa$ be the unique solution of
\begin{equation}
\label{eq:indicator-sampled-osgood-flow}
\dot{\mathcal R}_\kappa(t)
=\omega\!\left(M_W\mathcal R_\kappa(t)+\kappa\right)+\kappa,
\qquad \mathcal R_\kappa(0)=0,
\end{equation}
and put
\[
\Gamma_\kappa(t):=M_W\mathcal R_\kappa(t)+\kappa.
\]
There are constants $C_0,C,c>0$, depending only on $T,M_W,\abs\mu,\abs r$, such that, with
\begin{equation}
\label{eq:indicator-sampled-kappa}
\kappa_N(z):=C_0\left(\sqrt{\frac{\log N+z}{N}}+\frac1N\right),
\end{equation}
for every $z\ge1$ and $N\ge2$ the following event has probability at least $1-Ce^{-cz}$:
\begin{align}
\label{eq:indicator-sampled-osgood-H}
\sup_{t\le T}D_N(t)
&\le\mathcal R_{\kappa_N(z)}(T),\\
\label{eq:indicator-sampled-osgood-X}
\max_{i\le N}\sup_{t\le T}\abs{X_i^{N,+}(t)-X_t(U_i)}
&\le\Gamma_{\kappa_N(z)}(T).
\end{align}
The same event gives these bounds for every finite hard solution, together with
\begin{align}
\label{eq:indicator-sampled-osgood-W1}
\sup_{t\le T}\Wone\left(\frac1N\sum_{i=1}^N\delta_{X_i^{N,+}(t)},\Law(X_t(U))\right)
&\le\Gamma_{\kappa_N(z)}(T)+C\kappa_N(z),\\
\label{eq:indicator-sampled-osgood-fraction}
\sup_{t\le T}\left|\frac1N\sum_{i=1}^N\1_{\{X_i^{N,+}(t)\le0\}}
-\lambda\{X_t\le0\}\right|
&\le C\left[\omega\!\left(\Gamma_{\kappa_N(z)}(T)\right)+\kappa_N(z)\right].
\end{align}
The constant in the state-law estimate may also depend on $\norm{x_0}_{L^\infty}$. When $M_W>0$, the Osgood condition gives $\mathcal R_\kappa(T)\to0$ as $\kappa\downarrow0$; when $M_W=0$, convergence follows from continuity of $\omega$ at zero. For $\omega(a)=1\wedge La$, all four errors are $O_{\Prob}(\sqrt{\log N/N})$.
\end{theorem}

\begin{proof}[Proof sketch]
Conditional Hoeffding bounds control the sampled row errors, while the Dvoretzky--Kiefer--Wolfowitz inequality controls the empirical threshold mass. A time grid and the Lipschitz bound on the target paths make both estimates uniform on $[0,T]$. An indicator mismatch then yields
\[
D_N(t)\le\int_0^t
\left[\omega\!\left(M_WD_N(s)+\kappa_N(z)\right)+\kappa_N(z)\right]\,\dd s,
\]
and comparison with \eqref{eq:indicator-sampled-osgood-flow} closes the estimate.
\end{proof}

\begin{remark}[Sharpness of the sampled Osgood boundary]
Take $W\equiv1$, $\mu=r=0$, and $x_0(u)=u^2$. The continuum greatest solution has aggregate feedback $\beta(t)=t^2/4$ for small times, whereas an independent finite sample has $\min_iU_i^2>0$ almost surely and its hard system remains at $H_i^N=0$. A positive gap therefore persists between the sampled and continuum greatest solutions. For sufficiently small $a$, the intrinsic uniform tube is $m_X(a)=\sqrt{2a}$, so its Osgood integral is finite. More generally, the construction in \cref{prop:indicator-osgood-sharp} gives the same sampling obstruction for any prescribed non-Osgood tube modulus of the initial profile.
\end{remark}

Explicit sampled rates for log--Osgood tubes are given in \cref{cor:indicator-sampled-log-rates}.

\subsection{Verifying threshold regularity from the inputs}
\label{sec:threshold-criteria}

The preceding results assume a bound on threshold mass along a solution. The following conditions can be checked from the initial profile and kernel, before solving the equation. On each branch of the type space, we compare the separation of initial buffers with the variation that the kernel can generate over the time horizon. The next proposition uses this comparison to bound threshold mass for every admissible cumulative profile.

\begin{proposition}[A branchwise criterion from the initial profile and kernel]
\label{prop:indicator-primitive-osgood}
Let $0=a_0<a_1<\cdots<a_J=1$ and $I_j=(a_{j-1},a_j)$.  For every $j$,
suppose that there are $\sigma_j\in\{-1,1\}$, $\theta_j\in[0,1)$, and a
continuous strictly increasing $\chi_j:[0,|I_j|]\to[0,\infty)$ with
$\chi_j(0)=0$ such that, for $\lambda^2$-a.e. ordered pair $u<v$ in $I_j$,
\begin{align}
\label{eq:indicator-primitive-initial}
\sigma_j\bigl(x_0(v)-x_0(u)\bigr)
&\ge\chi_j(v-u),\\
\label{eq:indicator-primitive-network}
T\norm{W(v,\cdot)-W(u,\cdot)}_{L^1(I)}
+|r|T\abs{R_W(v)-R_W(u)}
&\le\theta_j\chi_j(v-u).
\end{align}
Define
\begin{align}
\label{eq:indicator-primitive-inverse}
\chi_j^{\leftarrow}(y)
&:=\sup\{s\in[0,|I_j|]:\chi_j(s)\le y\},\\
\label{eq:indicator-primitive-tube}
\omega_{\rm pr}(a)
&:=1\wedge\sum_{j=1}^J
\chi_j^{\leftarrow}\!\left(\frac{2a}{1-\theta_j}\right).
\end{align}
For every jointly measurable $h$ with $0\le h_t(u)\le t$, set
\[
X_t^h(u):=x_0(u)+\mu t+rtR_W(u)-\int_IW(u,z)h_t(z)\,\dd z.
\]
Then, for every such $h$,
\begin{equation}
\label{eq:indicator-primitive-uniform-tube}
\sup_{t\le T}\lambda\{u:\abs{X_t^h(u)-q}\le a\}
\le\omega_{\rm pr}(a),
\qquad q\in\R,\quad a\ge0.
\end{equation}
Consequently, if
\begin{equation}
\label{eq:indicator-primitive-osgood}
\int_{0+}\frac{\dd a}{\omega_{\rm pr}(a)}=\infty,
\end{equation}
the hard equation has at most one solution; if $W\ge0$, it has exactly one
solution by \cref{thm:indicator-greatest}.
\end{proposition}

For example, separations of order $s/[\log(A/s)]^\gamma$, $0\le\gamma\le1$, give Osgood tubes of order $a[\log(A/a)]^\gamma$. \Cref{cor:indicator-primitive-log-osgood,ex:indicator-rank-one-log-osgood} give the precise bounds and an example with an unbounded buffer density. A uniform bound on threshold-profile densities therefore suffices for uniqueness, but is not necessary.

\paragraph{Signed kernels and uniform density.}
For signed kernels, we can construct a hard solution by imposing a density bound for every admissible feedback profile. This is stronger than the tube condition along one solution: it makes the full feedback map Lipschitz and permits a fixed-point construction without an order argument.

\begin{assumption}
\label{ass:graphon-indicator-density}
Let $T>0$ and $C>0$. There exists $M_\rho^W(T,C)<\infty$ such that for every $t\in[0,T]$ and every $h\in L^\infty(I)$ with $\norm{h}_{L^\infty}\le C$, the scalar random variable
\begin{equation}
\label{eq:graphon-density-projection}
\Psi_t(U,h)=x_0(U)+\mu t+r tR_W(U)-\int_I W(U,v)h(v)\,\dd v,
\qquad U\sim\lambda,
\end{equation}
admits a density bounded by $M_\rho^W(T,C)$.
\end{assumption}

\begin{theorem}
\label{thm:graphon-indicator-general}
Let $W\in L^\infty(I^2)$ and $x_0\in L^\infty(I)$, and let $\ell(x)=\1_{\{x\le 0\}}$. Assume \cref{ass:graphon-indicator-density} holds with $C=T$. Then there exists a unique
\[
H\in C([0,T];L^1(I))\cap L^\infty([0,T]\times I)
\]
such that
\[
0\le H_t(u)\le t\qquad\text{for every }t\in[0,T]\text{ and for a.e. }u\in I,
\]
\[
\norm{H_t-H_s}_{L^1(I)}\le |t-s|\qquad\text{for all }s,t\in[0,T],
\]
and \eqref{eq:H-feedback-general} holds for every $t\in[0,T]$ and almost every $u\in I$. Consequently, the profile
\begin{equation}
\label{eq:X-from-H-general}
X_t(u)=x_0(u)+\mu t+r tR_W(u)-\int_I W(u,v)H_t(v)\,\dd v
\end{equation}
belongs to $C([0,T];L^\infty(I))$ and is the unique solution of the graphon contagion equation \eqref{eq:graphon-equation} with indicator loss.
\end{theorem}

We verify this density assumption by requiring the slope of the initial buffer profile on each branch to dominate the accumulated variation of kernel rows and columns. The derivative criterion is stated in \cref{prop:graphon-indicator-transverse}. A common positive margin gives existence and bounds path errors linearly in the $L^\infty$ data error throughout a family, including signed kernels. We use this family estimate to compare kernel approximations in \cref{sec:graphon-indicator-bridge}.

\section{Financial network interpretations}\label{sec:financial-interpretations}

\Cref{tab:examples} shows how several financial-network architectures can be represented by sender and receiver factors. Each represented channel contributes a factor pair, so the number of channels bounds the algebraic rank. In \cref{sec:real-data}, we construct common-sovereign channels from disclosed EBA holdings and study the resulting dynamics under prescribed stress scenarios.

\begin{table}[tbp]
\centering
\small
\renewcommand{\arraystretch}{1.22}
\setlength{\tabcolsep}{6pt}
\caption{Financial architectures and their low-rank representations. Normalizing constants such as $1/K$ are absorbed into the displayed loadings.}
\label{tab:examples}
\begin{tabularx}{\textwidth}{>{\raggedright\arraybackslash}p{0.22\textwidth}>{\raggedright\arraybackslash}p{0.34\textwidth}>{\centering\arraybackslash}p{0.10\textwidth}>{\raggedright\arraybackslash}X}
\toprule
Architecture & Kernel or factorization & Rank bound & Interpretation \\
\midrule
Homogeneous mean field & $W(u,v)\equiv 1$ or $W(u,v)=a(u)b(v)$ with $a\equiv b\equiv 1$ & $1$ & Equal exposure to aggregate distress. \\
Core-periphery or tiered payment network & $W(u,v)=a_1(u)b_1(v)+a_2(u)b_2(v)$ & $\le2$ & Two directions of transmission between core and periphery \cite{boss2004,craig2014}. \\
Multi-CCP market & $W(u,v)=\sum_{m=1}^M c_m(u)d_m(v)$ & $\le M$ & One loss-allocation channel per clearing venue \cite{veraart2025}. \\
Multiplex bank--NBFI system & $W=\sum_{\ell=1}^L \omega_\ell W_\ell$ & layer dependent & Separate lending, funding, clearing, and common-asset layers \cite{aldasoro2017,bcbs2025nbfi}. \\
\bottomrule
\end{tabularx}
\end{table}

\subsection{Mean field and interbank tiers}

The two-channel construction in \cref{ex:meanfield-to-rank2} describes reciprocal exposure between a core and a periphery, as motivated by tiered interbank markets \cite{boss2004,craig2014}. A distressed institution in either group reduces the buffers of its creditors in the other group. If those creditors become distressed, they can transmit losses back through the opposite channel. The two feedback paths record these directions separately; the size of the losses depends on group sizes, initial buffers, and loadings. In \cref{sec:numerics-exact-low-rank}, we compare the distress paths of the two groups.

On the continuum type space, homogeneous mean field becomes the constant kernel $W\equiv1$, a heterogeneous single channel becomes $W(u,v)=a(u)b(v)$, and the two-tier construction becomes a $2\times2$ block kernel. For a general bounded directed kernel, we approximate exposures by retaining finitely many such factors. Under the stated loss and threshold-regularity assumptions, the estimates in \cref{sec:graphon} bound the change in the state paths caused by this approximation.

\subsection{Multiple CCPs}

Consider a market with multiple central counterparties (CCPs)\@. Suppose that $c_m(u)$ measures the sensitivity of institution $u$ to losses allocated through CCP $m$, while $d_m(v)$ measures the contribution of institution $v$ to stress transmitted through that CCP\@. Then
\[
W(u,v)=\sum_{m=1}^M c_m(u)d_m(v)
\]
has rank at most $M$, where $M$ is the number of clearing venues. Each CCP contributes a factor, with shared or dependent loadings reducing the effective rank. We draw on the multiple-CCP framework of Veraart and Aldasoro \cite{veraart2025} for this interpretation. The same construction can describe overlapping clearing channels through which banks serve NBFIs \cite{bcbs2025nbfi}.

\subsection{Multiplex exposures and bank--NBFI networks}

A bank may be linked to the same counterparty through lending, repo, derivatives, collateral, payments, and common holdings. Bank--NBFI linkages add funding, clearing, and risk-transfer channels \cite{bcbs2025nbfi}. We represent each type of link by a separate kernel and sum the weighted contributions of the layers:

\[
W=\sum_{\ell=1}^L \omega_\ell W_\ell.
\]
If layer $\ell$ has rank $r_\ell$, the aggregate rank is at most $\sum_\ell r_\ell$, with possible reductions when factors are shared across layers. The empirical multiplex literature motivates keeping these layers distinct \cite{aldasoro2017}. In this representation, the rank parameter bounds the number of linearly independent transmission channels in the model.

The same factor construction can represent losses allocated through insurance, supply chains, and energy clearing; examples are given in \cref{sec:other-applications}.

\section{Numerical experiments}\label{sec:numerics}

We use the experiments to examine how distress spreads through different exposure channels and how accurately reduced models reproduce the resulting buffer paths. We begin with four institutions so that each episode of distress, recovery, and renewed loss transmission can be followed directly. We then study larger low-rank networks, comparing the trajectories of groups with different sender and receiver roles, and test the continuum approximations under sampling, smoothing, and kernel truncation. Finally, we construct channels from disclosed sovereign portfolios and compare reductions under prescribed shocks. This application shows how the banks that enter distress and the channels through which they transmit losses jointly determine approximation accuracy.

\subsection{Recovery followed by renewed distress}\label{sec:recovery-example}
An institution can recover and later return to distress as different counterparties cross the threshold. To exhibit this mechanism, consider the four-institution network in \cref{fig:recovery-reentry}, with $\ell(x)=\1_{\{x\le0\}}$, $N=4$, $r=0$, $\mu=0.1$, $T=5$, and
\[
x_0=(-0.1,-1,0.3,0.04),\qquad
\frac{e}{N}=\begin{pmatrix}
0&0&0&0\\0&0&0&0\\0&0.2&0&0\\0.2&0&0.3&0
\end{pmatrix}.
\]
The first institution recovers at time $1$. The second remains distressed throughout the horizon and pushes the third below zero at time $3$. Institution $4$ incurs losses from institution $1$ until time $1$ and from institution $3$ after time $3$, so its buffer is
\[
X_4(t)=0.04+0.1t-0.2\min(t,1)-0.3(t-3)_+.
\]
It first enters distress at $t=0.4$, recovers at $t=1.6$, and enters again at $t=3.7$. Its cumulative occupation time is
\[
H_4(t)=\min\{(t-0.4)_+,1.2\}+(t-3.7)_+,\qquad 0\le t\le5.
\]
Thus $H_4$ stops increasing during recovery and resumes during the later episode, reaching $2.5$ at $T=5$. The target's buffer changes direction when its counterparties enter or leave distress, even though all exposures and outside drifts remain constant.

\begin{figure}[H]\centering
\includegraphics[width=\textwidth]{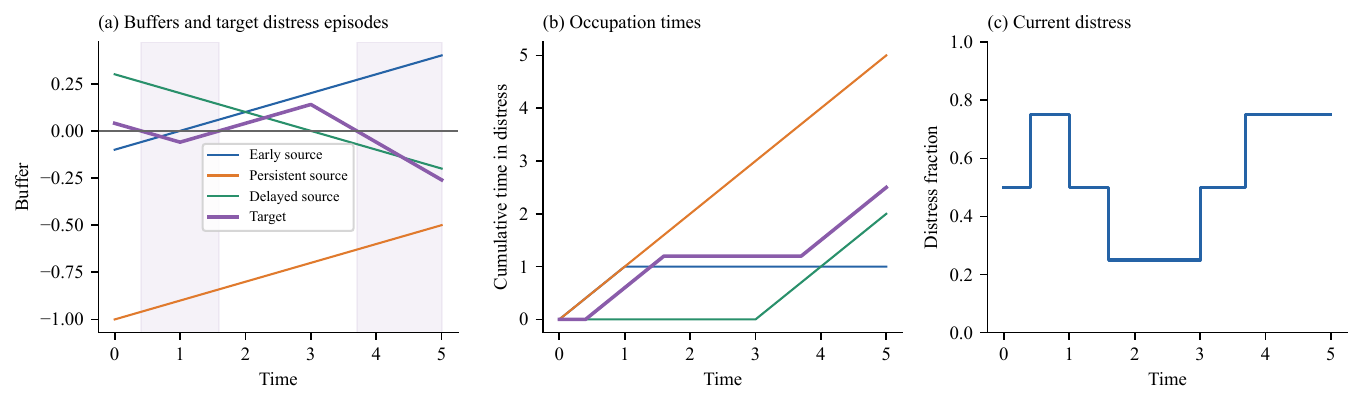}
\caption{Recovery and renewed distress under constant drift and a fixed nonnegative network. (a) Buffer paths; shading marks the target's two distress episodes. (b) Cumulative occupation times. (c) The instantaneous distress fraction, shown by its values immediately after each crossing. Every crossing is isolated and transverse, so the affine segments give the unique trajectory.}
\label{fig:recovery-reentry}\end{figure}

\subsection{Population feedback and observables}
The preceding example follows losses from individual counterparties. In larger populations, the factor coordinates aggregate those losses while retaining each group's exposure to them. We compare finite populations with their limiting group trajectories, then vary the direction and strength of transmission.

The four population comparisons use one or two exposure factors, hard loss $\ell(x)=\1_{\{x\le0\}}$, $r=0$, and $\mu<0$. An institution belongs to group $g$ with probability $p_g$; conditional on that group, its initial buffer has law $N(m_g,\sigma_g^2)$ and its factor loadings are $a_g,b_g$. \Cref{tab:lowrank-parameters} specifies each architecture. The limiting feedback solves
\begin{equation}\label{eq:numerical-beta}
\dot\beta_k(t)=\sum_g p_g b_{g,k}\Phi_{\mathrm N}\!\left(
\frac{-m_g-\mu t+K^{-1}a_g\cdot\beta(t)}{\sigma_g}\right),\qquad \beta_k(0)=0.
\end{equation}
We match the mixture using deterministic group counts and initial-buffer quantiles, and compare the resulting trajectories with those of independent populations sampled from the same complete-type law.

Throughout the experiments, $h_t^N=N^{-1}\sum_i\1_{\{X_t^{i,N}\le0\}}$ and $h_t=\mu_0\{z:X_t(z)\le0\}$ denote instantaneous threshold fractions. In graphon experiments, $h_t=\lambda\{u:X_t(u)\le0\}$. For a regularized trajectory, the threshold fraction still counts states below zero; the mean smoothed distress is the average loss rate. Population construction, integration, and error evaluation are specified in \cref{app:lowrank-parameters,sec:appendix-robustness}.

\subsection{Finite-population approximation under exact low rank}\label{sec:numerics-exact-low-rank}
The finite populations closely follow their limits while the groups retain different distress paths. The rank-one limiting fraction rises from $9.12\%$ to $40.62\%$. In the core--periphery model, the periphery reaches $25.93\%$ and the core $11.24\%$ at $T=2.5$ (\cref{fig:numerics-1}). One or two feedback coordinates reproduce these group differences without increasing the dimension as the population grows. The overlapping-CCP and bank--NBFI group paths appear in \cref{fig:numerics-2}.

\begin{table}[t]
\centering\small
\caption{Uniform threshold-fraction errors $\sup_{0\le t\le T}|h_t^N-h_t|$ for deterministic quantile-matched populations in the four exact low-rank architectures.}
\label{tab:numerics}
\begin{tabular}{lrrrrr}\toprule
Example & $K$ & $T$ & $h(T)$ & $N=400$ error & $N=1600$ error\\\midrule
Rank-one mean field & 1 & 2.0 & 0.4062 & 0.001276 & 0.000315 \\
Core-periphery & 2 & 2.5 & 0.2226 & 0.002543 & 0.000627 \\
Overlapping CCP channels & 2 & 2.3 & 0.2417 & 0.003410 & 0.000886 \\
Bank-NBFI multiplex & 2 & 2.2 & 0.1304 & 0.002546 & 0.000622 \\
\bottomrule\end{tabular}
\end{table}
\begin{figure}[t]\centering
\includegraphics[width=\textwidth]{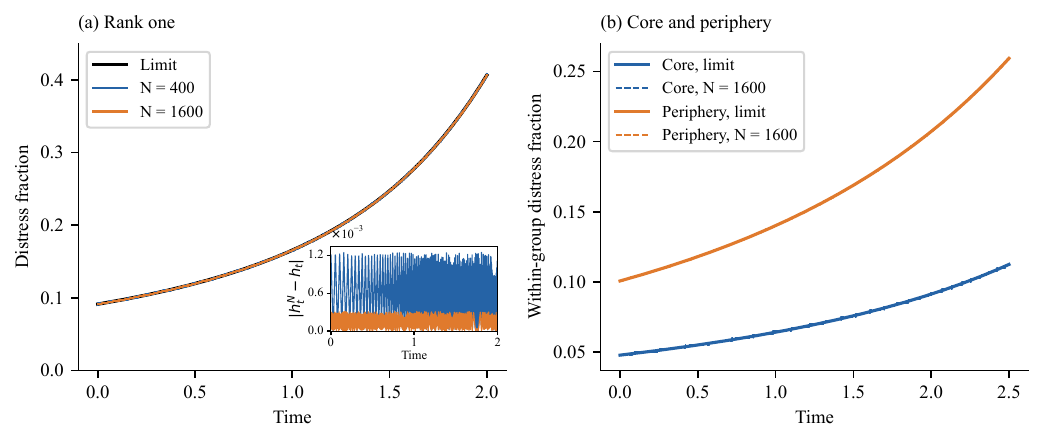}
\caption{Quantile-matched finite populations and their Gaussian limits. (a) The rank-one threshold fraction; the inset shows $|h_t^N-h_t|$ at the common observation times for $N=400$ (blue) and $N=1600$ (orange). (b) Within-group fractions for core and periphery; dashed lines show $N=1600$ and solid lines the limit.}
\label{fig:numerics-1}\end{figure}

Across the four architectures, the finite and limiting threshold fractions differ uniformly by between $3.15\times10^{-4}$ and $8.86\times10^{-4}$ at $N=1600$ (\cref{tab:numerics}). Increasing $N$ from $400$ to $1600$ reduces these errors by factors between $3.85$ and $4.09$. Quantile matching fixes the group composition and the initial-buffer quantiles; independent sampling gives the wider ranges in \cref{tab:mc-summary}.

To measure how this sampling variation decreases with population size, we repeat the rank-one benchmark for independent populations. The mean feedback-path error has a fitted log--log slope of $-0.465$ against $N$, close to the $-1/2$ reference slope (\cref{fig:indicator-poc-rate}). \Cref{sec:appendix-robustness-mc,sec:appendix-poc-rate} report the trajectory ranges and the convergence experiment.

\subsection{Directed imbalance and multilayer amplification}
The direction of exposures determines which group benefits from the imbalance term. In the two-group directed example, setting $r=0.05$ gives net lenders a terminal fraction of $0.091\%$, compared with $0.236\%$ at $r=0$, while net borrowers rise from $0.391\%$ to $13.47\%$ (\cref{fig:numerics-4}a). Net lenders also finish below their initial fraction $0.274\%$. The same change therefore reduces lender distress and increases borrower distress.

The two groups have equal population weights and opposite imbalance terms, $R_1=1.82$ and $R_2=-1.82$. Their baseline drifts are $\mu+rR_g$, so the imbalance raises the lenders' drift and lowers the borrowers' drift. The feedback follows \eqref{eq:numerical-beta} with the corresponding group drift in place of $\mu$. We hold the exposure loadings and initial-buffer distributions fixed as $r$ changes; \cref{app:lowrank-parameters} gives the complete specification.

\begin{figure}[t]\centering
\includegraphics[width=\textwidth]{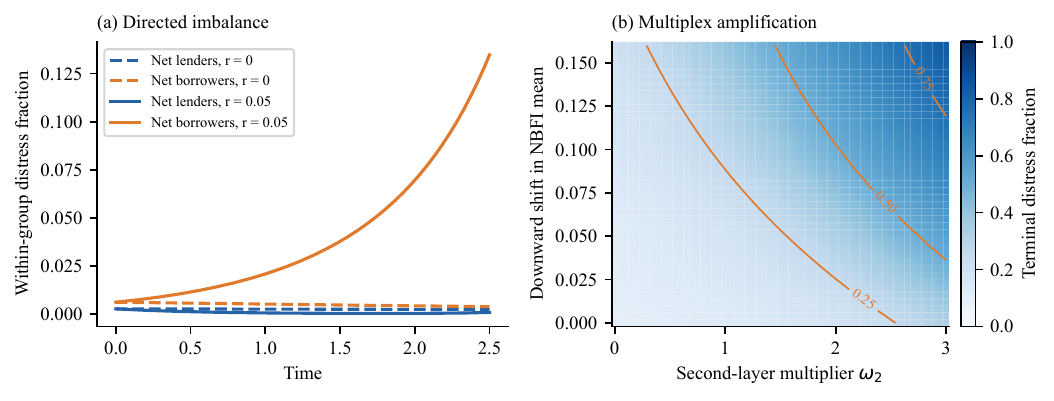}
\caption{Directed imbalance and multiplex amplification. (a) Within-group threshold fractions for the two-group directed model. (b) Terminal fraction in the bank--NBFI model as the second-layer receiver loading is multiplied by $\omega_2$ and the NBFI mean buffer is reduced by $\Delta_{\mathrm{NBFI}}$. Contours mark fractions $0.25$, $0.50$, and $0.75$. The normalization remains $K=2$ throughout this comparison.}
\label{fig:numerics-4}\end{figure}

A stronger second layer amplifies the effect of weaker NBFI buffers (\cref{fig:numerics-4}b). At the baseline buffer distribution, adding the second layer raises terminal distress from $10.06\%$ at $\omega_2=0$ to $13.04\%$ at $\omega_2=1$. With an NBFI mean reduction of $0.08$, the corresponding fractions are $14.27\%$ and $23.40\%$, increasing to $43.23\%$ at $\omega_2=2$.

To obtain this comparison, we use the multiplex model in \cref{tab:lowrank-parameters} and replace $a_{g,2}$ by $\omega_2a_{g,2}$ and the NBFI mean $0.24$ by $0.24-\Delta_{\mathrm{NBFI}}$. The normalization and other model inputs remain fixed, and the response is continuous in these parameters.

\subsection{Threshold concentration and continuum approximation}
\label{sec:numerics-piecewise-indicator-bridge}

When many institutions lie close to distress, small buffer changes can activate a substantial loss flow. We illustrate nonuniqueness in a rank-one model and measure smoothing and truncation errors in a network with a thin threshold layer.

\subsubsection{The Osgood boundary and ramp selection}
\label{sec:numerics-osgood-boundary}

Take $W\equiv1$, $\mu=r=0$, and $x_0=F^{\leftarrow}$ for a continuous distribution function $F$ on $[0,1]$. The aggregate cumulative distress $\beta(t)=\int_I H_t(u)\,\dd u$ satisfies
\begin{equation}
\label{eq:numerics-osgood-scalar}
\dot\beta=F(\beta),\qquad
\dot\beta_\varepsilon=\frac1\varepsilon\int_0^\varepsilon F(\beta_\varepsilon+s)\,\dd s,
\qquad\beta(0)=\beta_\varepsilon(0)=0,
\end{equation}
where the second equation uses the positive-side ramp. We compare $F(a)=a$, $F(a)=a\log(e/a)$, and $F(a)=\sqrt a$ on $[0,1]$. In the first two cases the reciprocal integral diverges and the unique hard path is $\beta=0$. In the square-root case, $x_0(u)=u^2$ and every delayed path
\[
\beta_\tau(t)=\tfrac14(t-\tau)_+^2,\qquad \tau\ge0,
\]
is a solution on $[0,1]$. Thus the same exposure structure and initial population admit both immediate transmission and arbitrarily delayed activation.

The two one-sided ramps select different paths in the square-root example (\cref{fig:osgood-boundary}). The positive-side ramp converges to the greatest path $\beta^+(t)=t^2/4$; the negative-side ramp stays on the zero path. For the linear case, $\beta_\varepsilon(t)=\frac\varepsilon2(e^t-1)$, so its error is exactly linear in $\varepsilon$. The logarithmic case has a unique hard path, but its ramp approximation converges much more slowly. Fits over the four smallest widths give exponents $1.000$, $0.365$, and $0.500$ for the linear, logarithmic, and square-root profiles, respectively; the logarithmic exponent is close to the comparison scale $e^{-1}$. A uniquely determined loss flow can therefore remain sensitive to the smoothing width.

\begin{figure}[tbp]
\centering
\includegraphics[width=\textwidth]{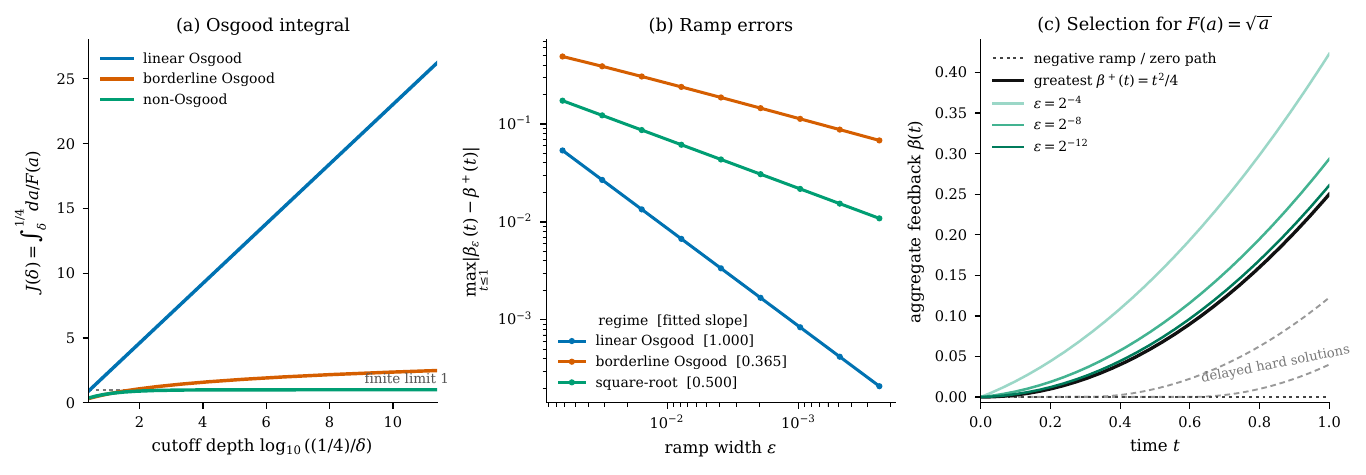}
\caption{Threshold concentration and path selection for a constant rank-one kernel. (a) The reciprocal integral diverges for the linear and logarithmic profiles and has a finite limit for the square-root profile. (b) Positive-ramp errors for widths $2^{-4},\ldots,2^{-12}$; brackets give fitted exponents. (c) For $x_0(u)=u^2$, positive ramps approach the greatest hard path, negative ramps remain at zero, and delayed hard paths lie between them.}
\label{fig:osgood-boundary}
\end{figure}

\subsubsection{An infinite-rank graphon with a thin threshold layer}
\label{sec:numerics-piecewise-indicator}

To examine smoothing in a directed infinite-rank network, let $I_1=[0,\frac12)$, $I_2=[\frac12,1]$, and
\[
W(u,v)=B_{ij}+0.16e^{-10\sin^2(\pi(u-v))},\qquad
B=\begin{pmatrix}0.38&0.24\\0.31&0.36\end{pmatrix},
\quad (u,v)\in I_i\times I_j.
\]
The block component distinguishes the two directions of transmission, while the periodic component has infinitely many Fourier modes. We set $r=0$, $\mu=-0.06$, $T=2$, and
\[
x_0(u)=\begin{cases}-0.18+1.95u,&u<\frac12,\\
0.82+1.75(u-\frac12),&u\ge\frac12.
\end{cases}
\]
The minimum branch slope is $m_0=1.75$. The derivative integrals in \cref{prop:graphon-indicator-transverse} are $A_1=A_2=0.32(1-e^{-10})$, giving the positive margin $m_0-TA_1\approx1.11$. This supplies a common density bound for the feedback profiles used in the hard-threshold construction.

\begin{figure}[tbp]
\centering
\includegraphics[width=\textwidth]{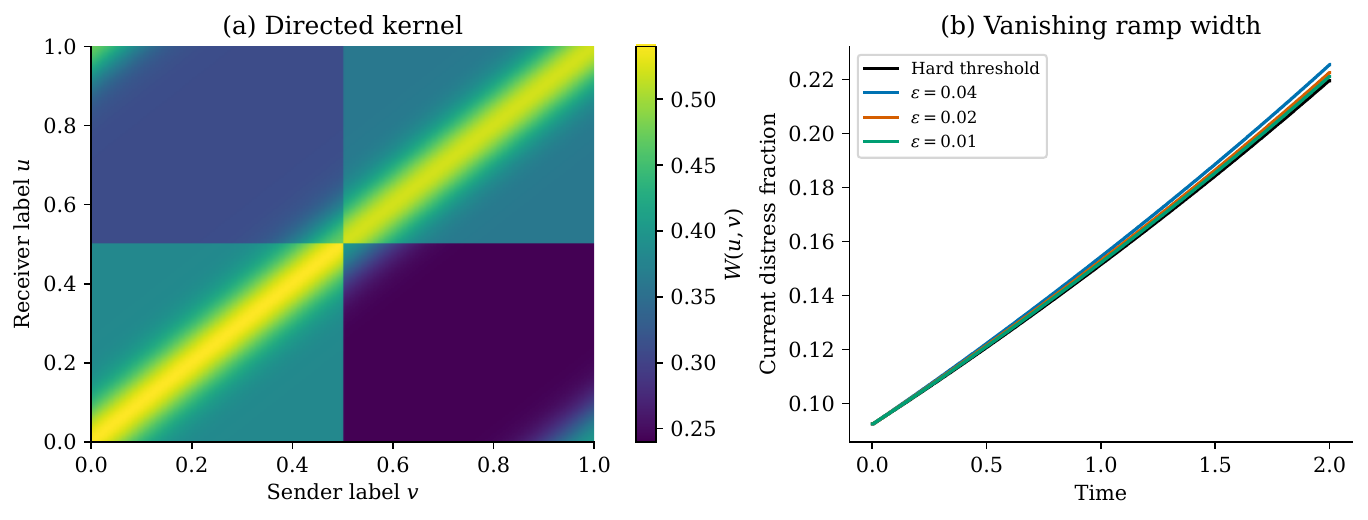}
\caption{Smoothing a transverse threshold layer. (a) The directed block-plus-periodic kernel. (b) Current distress fractions for the hard rule and three positive-side ramps on 2000 midpoint labels. The hard path is integrated between consecutive threshold crossings.}
\label{fig:piecewise-indicator}
\end{figure}

Under the hard rule, the distressed share grows to $21.95\%$ by $T=2$. Positive-side smoothing activates institutions slightly before their buffers reach zero, increasing both the transmitted loss and the terminal distress fraction. Halving the ramp width approximately halves the state discrepancy (\cref{fig:piecewise-indicator,tab:piecewise-indicator-summary}). We then flatten the initial-buffer profile in \cref{sec:numerics-piecewise-sensitivity}; placing more institutions near the threshold increases the response to smoothing.

\begin{table}[tbp]
\centering\small
\caption{Smoothing the piecewise graphon. State errors are $\max_t\|X_t^\varepsilon-X_t^{\rm ind}\|_{L^1}$ at the common output times; threshold errors compare the corresponding current distress fractions.}
\label{tab:piecewise-indicator-summary}
\begin{tabular}{lrrrr}\toprule
Rule & $\varepsilon$ & $h(T)$ & State error & Threshold error\\\midrule
Hard threshold & -- & 0.2195 & 0 & 0\\
Positive ramp & 0.04 & 0.2255 & $9.40\times10^{-3}$ & $6.00\times10^{-3}$\\
Positive ramp & 0.02 & 0.2225 & $4.70\times10^{-3}$ & $3.00\times10^{-3}$\\
Positive ramp & 0.01 & 0.2210 & $2.35\times10^{-3}$ & $1.50\times10^{-3}$\\\bottomrule
\end{tabular}
\end{table}

\subsubsection{Preserving structure under kernel approximation}
\label{sec:numerics-piecewise-indicator-trig-bridge}

We retain the block component exactly and approximate the periodic component by
\[
W^{{\rm sm},(M)}(u,v)=0.16e^{-5}\left[I_0(5)+2\sum_{m=1}^M I_m(5)\cos(2\pi m(u-v))\right],
\]
where $I_m$ denotes the modified Bessel function. Each positive frequency contributes a sine--cosine pair. All truncations share the derivative bound
\[
A_1^*,A_2^*\le\left(0.256\pi^2e^{-10}I_1(10)\right)^{1/2}<0.554,
\]
and hence $TA_1^*<1.108<m_0$. The family meets the hypotheses of \cref{cor:trigonometric-family}, so adding factors improves the kernel while preserving the common bound on the population near the threshold.

\begin{figure}[tbp]
\centering
\includegraphics[width=\textwidth]{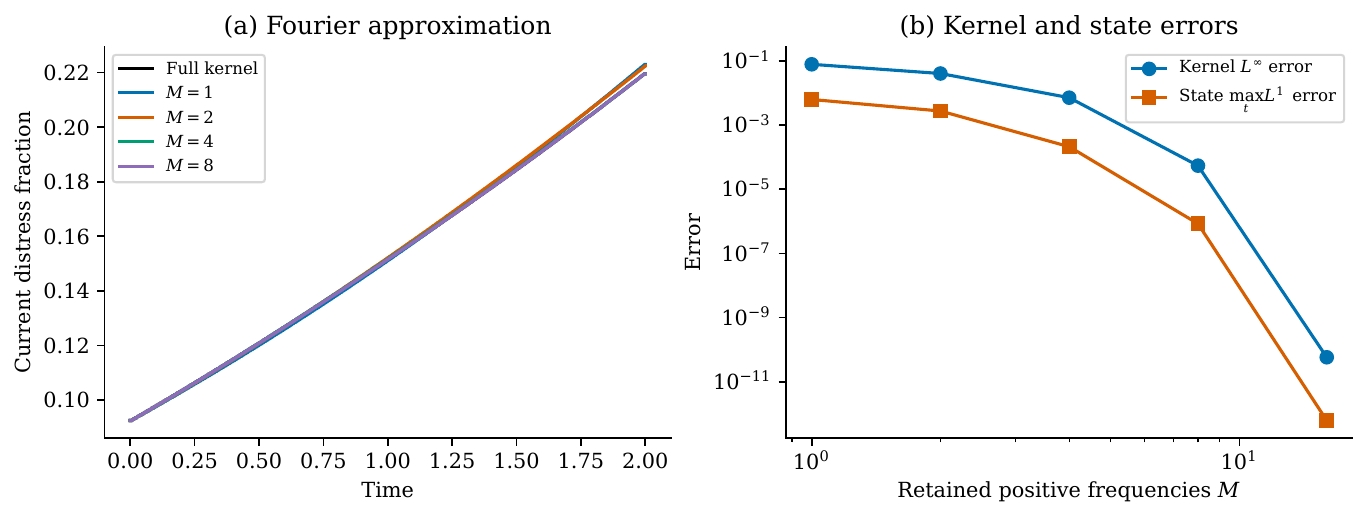}
\caption{Finite-factor approximation of the hard-threshold graphon. (a) Current distress fractions under Fourier truncation. (b) Kernel and state discrepancies relative to the full periodic kernel on the same spatial grid. The common transversality margin allows the hard-threshold stability theorem to apply throughout the family.}
\label{fig:piecewise-indicator-bridge}
\end{figure}

The state discrepancy falls from $6.26\times10^{-3}$ at $M=1$ to $2.11\times10^{-4}$ at $M=4$ and $8.38\times10^{-7}$ at $M=8$ (\cref{fig:piecewise-indicator-bridge,tab:piecewise-indicator-bridge}). The common transversality margin bounds the uniform-in-time $L^1$ errors in cumulative distress and buffers by a constant multiple of the $L^\infty$ kernel error, uniformly over this Fourier family (\cref{thm:uniformly-transverse-stability}).

With a bounded Lipschitz loss, we also truncate a directed positive series in \cref{app:graphon-truncation}. The state error decreases from $3.27\times10^{-2}$ with four retained factors to $2.85\times10^{-3}$ with 24 factors. Here the aligned $L^1$ estimate controls the effect of discarded factors. The hard-threshold estimate uses, in addition, a bound on the population near the threshold.

\subsection{Sovereign portfolios, threshold sensitivity, and factor selection}
\label{sec:real-data}

The sovereign-portfolio experiments ask which institutions a given stress pushes below zero and how closely different channel selections reproduce their buffer paths. We use the portfolio structure disclosed in the EBA 2025 EU-wide transparency exercise \cite{eba2025}: direct on-balance-sheet gross sovereign holdings at 30 June 2025 (item 2520810), aggregated over residual maturities. The 119 named institutions disclose EUR~3.914 trillion. Removing three institutions with zero holdings leaves 116 banks; the 38 country or regional buckets accounting individually for at least $0.1\%$ of aggregate holdings cover $99.704\%$ of the total. \Cref{app:eba-construction} describes the data construction and the six-bank illustration.

Let $s_{i,k}$ be bank $i$'s disclosed holding in retained bucket $k$. The receiver and sender loadings are
\begin{equation}
\label{eq:a-b}
a_{i,k}^{\mathrm{emp}}=\frac{s_{i,k}}{\sum_{l=1}^K s_{i,l}},
\qquad
b_{i,k}^{\mathrm{emp}}=\frac{s_{i,k}}{N^{-1}\sum_{j=1}^N s_{j,k}}.
\end{equation}
The receiver loading is a portfolio share within the retained book; the sender loading measures an institution's holding relative to the cross-sectional mean for that channel. These holdings define the common-exposure kernel
\begin{equation}
\label{eq:empirical-e}
e_{ij}^{N,\mathrm{emp}}=\lambda_E\sum_{k=1}^K a_{i,k}^{\mathrm{emp}}b_{j,k}^{\mathrm{emp}},
\qquad \lambda_E=0.075,
\qquad \frac1N\sum_j e_{ij}^{N,\mathrm{emp}}=\lambda_E.
\end{equation}
For each complete retained-book construction, every bank has row-mean interaction strength $\lambda_E$. In the normalization of the factor model, $c_E=K\lambda_E$: it is $0.30$ for four channels and $2.85$ for 38 channels. Diagonal entries represent each bank's own participation in a shared sovereign channel.

\subsubsection{A stylized stress and its threshold layer}

We keep the disclosed portfolios as the exposure loadings and impose a stylized stress to set the initial buffers. The loss intensity is also specified as part of the scenario. We take $T=3$, $\mu=r=0$, and
\begin{equation}
\label{eq:initial-buffer}
x_i=x_\ast-\sum_k\eta_k a_{i,k}^{\mathrm{emp}},
\quad x_\ast=0.40,
\quad (\eta_{\rm FR},\eta_{\rm DE},\eta_{\rm IT},\eta_{\rm NL})=(0.35,0.15,0.45,0.05),
\end{equation}
with all other shock entries zero. Losses increase linearly with the depth of distress and saturate once the buffer reaches $-\varepsilon$:
\begin{equation}
\label{eq:empirical-dynamics}
X_t^{i,N}=x_i-\frac1N\sum_j e_{ij}^{N,\mathrm{emp}}
\int_0^t\ell_\varepsilon^-(X_s^{j,N})\,\dd s,
\qquad
\ell_\varepsilon^-(x)=\max\{0,\min\{1,-x/\varepsilon\}\},
\quad\varepsilon=0.05.
\end{equation}
This negative-side ramp starts the loss flow when a buffer becomes negative. We report both the instantaneous distress rate $N^{-1}\sum_i\ell_\varepsilon^-(X_t^{i,N})$ and the number of buffers below zero.

In the six-bank illustration, Intesa Sanpaolo starts below zero and Banco Santander starts at $0.00068$. Loss transmission subsequently carries Santander through the threshold. The terminal mean distress rate is $0.1696$, with two of six banks below zero. Doubling $\lambda_E$ raises these quantities to $0.5691$ and five of six banks. The other initial buffers and the intensity comparison appear in \cref{tab:six-bank-intensity,tab:holdings}.

We vary the common buffer level $x_\ast$ while holding portfolios, shock weights, and interaction strength fixed (\cref{fig:eba-threshold}). In the six-bank example, $x_\ast=0.38$ produces two initial and five terminal distressed banks; $x_\ast=0.40$ produces one initial and two terminal distressed banks; at $x_\ast=0.41$ no bank starts distressed and the loss flow remains zero. The initial count within $0.02$ of zero is two in all three cases. Banks near zero can enter distress after a small transmitted loss, but the initially distressed banks must generate that loss flow in this zero-drift experiment.

For the 116-bank population, the baseline has four banks below zero initially and at $T=3$, although losses reduce institutional buffers by as much as $0.0256$. Lowering $x_\ast$ to $0.38$ leaves the initial count at four and raises the terminal count to six; at $x_\ast=0.36$ the count rises from five to seven. Thus an unchanged distressed count can coexist with declining buffers; lowering the common buffer level then brings additional banks below zero on the same portfolio network. The full sensitivity grid in \cref{app:eba-robustness} also varies $\lambda_E$ and $\varepsilon$.

\begin{figure}[htbp]
\centering
\includegraphics[width=.98\textwidth]{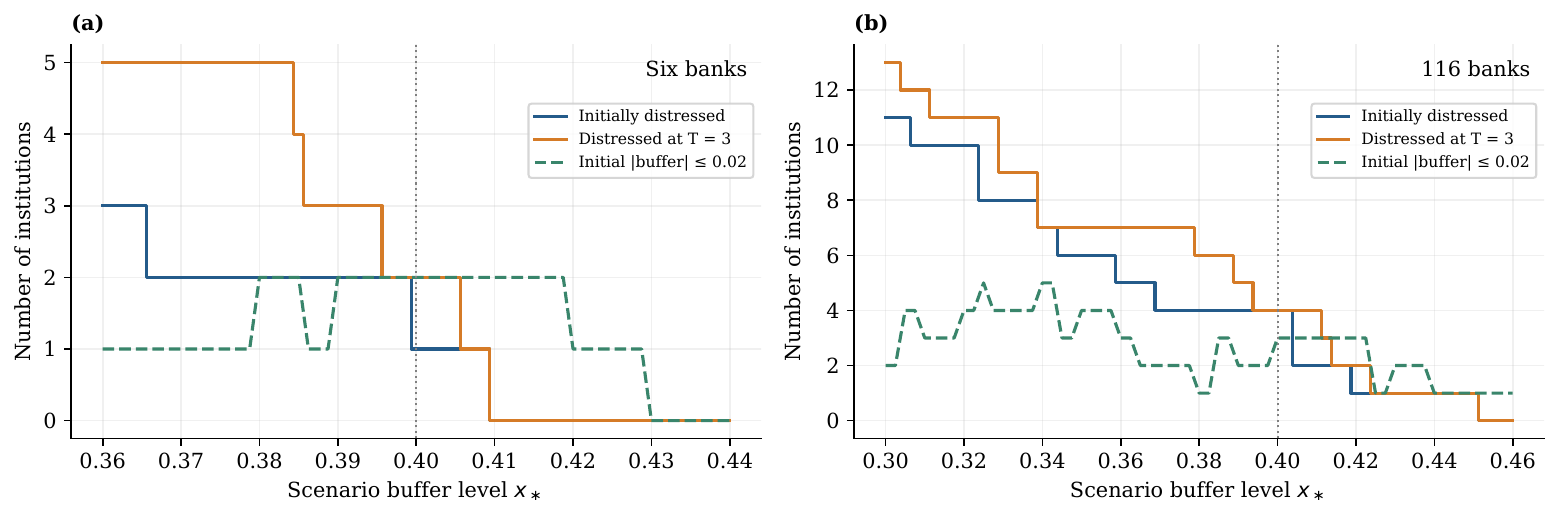}
\caption{Initial buffers and additional threshold crossings under the fixed sovereign-portfolio stress. The common buffer level $x_\ast$ varies; all portfolio loadings, shock weights, and $\lambda_E=0.075$ remain fixed. The dashed curve counts institutions with initial buffer within $0.02$ of zero. (a) Six banks and four sovereign channels. (b) 116 banks and 38 channels. The dotted vertical line marks $x_\ast=0.40$.}
\label{fig:eba-threshold}
\end{figure}

\subsubsection{Which channels should a reduced model retain?}
\label{sec:full-sample}

For the full population, we compare two ways to reduce the same 38-factor kernel. SVD retains the leading singular modes. The scenario ordering ranks factors by the prescribed shock weight $\eta_k$, breaking ties by aggregate holdings; its first two channels are Italy and France. Both approximations start from the original initial buffers. Channel truncation also preserves the original weights of the retained factors, without renormalization. We measure each approximation's maximum buffer-path error against the full system.

Eleven singular modes account for $90\%$ of the kernel's Frobenius energy and 19 account for $99\%$. Yet retaining only Italy and France gives path error $7.20\times10^{-4}$, compared with $2.34\times10^{-2}$ for rank-$2$ SVD and $1.12\times10^{-2}$ for rank-$12$ SVD; see \cref{fig:fullsample-rank}. In this baseline, the two scenario factors reduce the latter error by a factor of about $15.6$. A rank-$32$ SVD eventually gives the smaller error, $4.95\times10^{-5}$ versus $1.90\times10^{-4}$ for 32 scenario-ordered factors. The modes that best preserve matrix energy therefore need not preserve the buffer response to a specified stress equally well.

In the baseline, only the four initially distressed banks generate losses. Decomposing the cumulative buffer losses at $T=3$ by country channel and summing over receiving banks assigns $91.42\%$ to Italy and $1.39\%$ to France. These two channels therefore account for most of the cumulative losses over the horizon. The rank-$12$ SVD also changes the direction of transmission at some institutions. At Jyske Bank, it assigns negative exposure weights to all four distressed senders: the bank's buffer rises from $0.34563$ to $0.35671$, whereas the full kernel lowers it to $0.34549$. This reversal produces the largest path error. Retaining the original country factors preserves both the dominant loss channel and nonnegative loss transmission in this scenario.

We repeat the comparison for three assignments of the four nonzero shock weights and three buffer levels (\cref{tab:eba-scenarios}). In all nine cases the two channels with the largest prescribed shocks outperform the rank-$12$ SVD for this path-error criterion. Their error ranges from $2.60\times10^{-4}$ to $1.04\times10^{-2}$, with the largest errors in scenarios that generate further threshold crossings.

\begin{figure}[H]
\centering
\includegraphics[width=.98\textwidth]{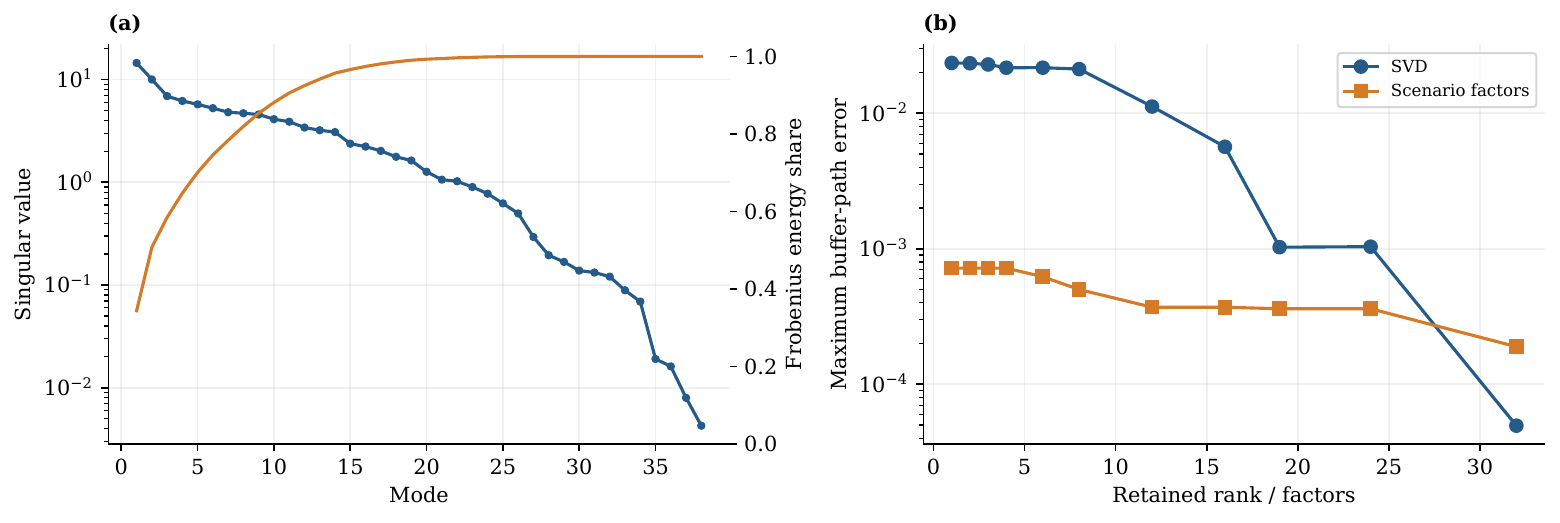}
\caption{Spectral compression and scenario-ordered channel retention for the 116-bank baseline. (a) Singular values of the common-exposure kernel and cumulative Frobenius energy. (b) Maximum buffer-path errors with identical initial buffers and unchanged retained channel weights.}
\label{fig:fullsample-rank}
\end{figure}

For the sampling accuracy of this fixed empirical population, \cref{app:eba-resampling} reports the smoothed and hard-threshold experiments and derives the smoothed $O(N^{-1/2})$ expected-error bound.

\section{Concluding remarks}
\label{sec:conclusion}

We have studied a dynamic contagion model in which an institution's current buffer determines its loss transmission rate. Under the hard rule, an institution transmits losses throughout each distress episode. Recovery stops new loss transmission, but creditors continue to bear the losses accumulated during earlier episodes. We showed that sender--receiver factors preserve these dynamics in a small system of feedback equations. When the exposure matrix has an exact factor representation, this system recovers every institution's buffer path. Our stability estimates quantify state-law errors caused by population approximation and path errors caused by changes in the exposure kernel.

For a hard threshold, the same exposure network can respond very differently depending on how many institutions are close to distress. We proved that an Osgood bound on this mass along one reference path ensures uniqueness and bounds the amplification of perturbations. Our rank-one examples explain why the threshold distribution matters: sufficiently concentrated profiles permit delayed cascades, while an atom at the threshold can produce an unstable response even when the path is unique. For nonnegative exposures, we also showed that positive-side regularization selects the greatest cumulative-distress solution whenever several hard paths coexist.

The sovereign-portfolio experiments illustrate how exposure structure and initial buffers jointly shape contagion. Country loadings determine which banks transmit and receive losses along each stressed channel; initial buffers determine which banks are distressed at the outset and influence which banks subsequently cross the threshold. In the baseline calculation, two channels chosen from the shock weights reproduce the full network's buffer paths more closely than twelve leading singular modes. Lowering the initial buffers activates additional banks even though their exposures remain unchanged. These results show why we assess a reduced network by the contagion paths it reproduces under a specified shock, taking account of both the channels that carry losses and the institutions close enough to distress to transmit them.

\clearpage\appendix

\section{Proofs of the main results}
\label{app:proofs}

The proofs are grouped into finite-rank estimates, kernel estimates, elementary model facts, and hard-threshold stability arguments.

\subsection{Proofs for the finite-rank results}
\label{app:proofs-sec3}

\begin{proof}[Proof of \Cref{thm:finite-rank-w1}]\label{proof:thm:finite-rank-w1}
Fix $T>0$, let $\pi^N\in\Gamma(\mu_0^N,\mu_0)$ be an optimal coupling, and write
\[
q_N:=\Wone(\mu_0^N,\mu_0),\qquad
A_N:=\int |a-\tilde a|_1\,\dd\pi^N,\qquad
B_N:=\int |b-\tilde b|_1\,\dd\pi^N.
\]
Then $\int|x-\tilde x|\,\dd\pi^N+A_N+B_N=q_N$. Since differences of factor means are bounded by the corresponding coupled coordinate differences,
\[
\sum_{k=1}^K|\bar a_k^N-\bar a_k|\le A_N,
\qquad
\sum_{k=1}^K|\bar b_k^N-\bar b_k|\le B_N.
\]
Expanding the imbalance term and using the bound $M$ on the factors and their means gives
\[
\int|\Lambda^N(z)-\Lambda(\tilde z)|\,\dd\pi^N
\le \frac{2|r|M}{K}(A_N+B_N).
\]
Define the coupled state error
\[
D_t^N:=\int|X_t^N(z)-X_t(\tilde z)|\,\dd\pi^N.
\]
The state map and $|\beta_k(t)|\le Mt\ell_\ast$ imply
\begin{equation}
\label{eq:D-before-m}
\begin{aligned}
D_t^N&\le \int|x-\tilde x|\,\dd\pi^N
+\frac{2|r|Mt}{K}(A_N+B_N)
+\frac{Mt\ell_\ast}{K}A_N \\
&\quad+\frac{M}{K}\sum_{k=1}^K\int_0^t|m_k^N(s)-m_k(s)|\,\dd s.
\end{aligned}
\end{equation}
For the feedback terms, boundedness and Lipschitz continuity of $\ell$ give
\[
|m_k^N(s)-m_k(s)|
\le \ell_\ast\int|b_k-\tilde b_k|\,\dd\pi^N+ML_\ell D_s^N.
\]
Substitution into \eqref{eq:D-before-m}, followed by $K\ge1$, yields
\[
D_t^N\le \bigl[1+TM(2|r|+\ell_\ast)\bigr]q_N
+M^2L_\ell\int_0^tD_s^N\,\dd s,\qquad 0\le t\le T.
\]
Gronwall's lemma bounds $\sup_{t\le T}D_t^N$ by
$[1+TM(2|r|+\ell_\ast)]e^{M^2L_\ell T}q_N$.
Finally, $(\Xi_t^N,\Xi_t)_\#\pi^N$ couples the joint laws, so
\[
\Wone(\widehat\nu_t^N,\widehat\nu_t)\le D_t^N+A_N+B_N\le D_t^N+q_N.
\]
Thus one may take
$C_T=1+[1+TM(2|r|+\ell_\ast)]e^{M^2L_\ell T}$, which is independent of $K$.
\end{proof}

\begin{proof}[Proof of \Cref{thm:transport-pde}]\label{proof:thm:transport-pde}
Since \(\beta_k(t)=\int_0^t m_k(s)\,\dd s\) and \(|m_k(s)|\le M\ell_\ast\), each \(\beta_k\) is absolutely continuous and \(\dot\beta_k(t)=m_k(t)\) for a.e. \(t\). By construction,
\[
X_t(z)=x+\Theta_t(a,b),
\]
so \(\Xi_t(x,a,b)=(x+\Theta_t(a,b),a,b)\). This map is a translation in the \(x\)-variable, hence the pushforward density is exactly \eqref{eq:density-explicit}. The formula \eqref{eq:mk-density} follows from the same change of variables.

It remains to justify the transport equation. Let \(\varphi\in C_c^1([0,T)\times\mathcal Z)\). For fixed \((x,a,b)\), the map
\[
t\mapsto \varphi(t,x+\Theta_t(a,b),a,b)
\]
is absolutely continuous, and for a.e. \(t\),
\[
\frac{\dd}{\dd t}\varphi(t,x+\Theta_t(a,b),a,b)
=
\partial_t\varphi(t,x+\Theta_t(a,b),a,b)
+\dot\Theta_t(a,b)\partial_x\varphi(t,x+\Theta_t(a,b),a,b).
\]
Since \(\dot\Theta_t(a,b)=v_f(t,a,b)\) for a.e. \(t\), integrating in time and then integrating against \(\mu_0\) gives \eqref{eq:transport-weak-form}; the terminal term vanishes because \(\varphi\) is compactly supported in \([0,T)\). Thus \eqref{eq:transport-pde} holds in the distributional sense.

For $t_n\to t$, continuity of $\Theta$ and compact support of $\Phi$ imply
\[
\norm{f(t_n,\cdot)-f(t,\cdot)}_{L^1(\mathcal Z)}\longrightarrow0
\]
by dominated convergence. The static factor support is unchanged by the translation, so
\[
|m_k(t_n)-m_k(t)|
\le M\ell_\ast\norm{f(t_n,\cdot)-f(t,\cdot)}_{L^1(\mathcal Z)}\longrightarrow0.
\]
Thus $\Theta$ is $C^1$ in time, and \eqref{eq:density-explicit} gives
\[
\partial_t f(t,x,a,b)=-\dot\Theta_t(a,b)\partial_x f(t,x,a,b).
\]
Because \(v_f=\dot\Theta_t\) and \(v_f\) is independent of \(x\), this is equivalent to the classical equation \eqref{eq:transport-pde}.
\end{proof}

\begin{proof}[Proof of \Cref{thm:indicator}]\label{proof:thm:indicator}
Fix $T>0$ and set $\rho_\ast=M_\rho(T,MT)$. The feedback satisfies $|F_k|\le M$, so every solution remains in $\mathcal C_T=[-MT,MT]^K$. For $\beta,\widetilde\beta\in\mathcal C_T$,
\[
|\Psi_t(z,\beta)-\Psi_t(z,\widetilde\beta)|
\le M\norm{\beta-\widetilde\beta}_\infty=:\delta.
\]
An indicator mismatch is contained in $\{|\Psi_t(z,\beta)|\le\delta\}$. Hence
\[
|F_k(t,\beta)-F_k(t,\widetilde\beta)|
\le 2M\rho_\ast\delta
=2M^2\rho_\ast\norm{\beta-\widetilde\beta}_\infty.
\]
The field is measurable in time and Lipschitz on the cube. Compose it with coordinatewise projection onto $\mathcal C_T$ to obtain a globally Lipschitz bounded field. Its unique solution from zero satisfies $|\beta_k(t)|\le Mt$ and therefore solves the original equation on $[0,T]$. Since $T$ is arbitrary, the solution is global.
\end{proof}

\begin{proof}[Proof of \Cref{lem:indicator-vc}]\label{proof:lem:indicator-vc}
The thresholds are affine halfspaces in $\R^{2K+1}$, with VC dimension at most $d=2K+2$. For fixed $k$, membership of a point $(z,y)$ in the subgraph of $b_k\mathbf1_H$ is determined by the three statements $z\in H$, $y<b_k$, and $y<0$. Only the first varies with $H$. Thus the number of subgraph traces on $n$ points is at most $K(n+1)^d$, after taking all $K$ coordinates. This polynomial trace bound proves that the class is VC-subgraph; it remains valid when the factors are signed.

For measurability, let $\mathcal H_0$ be the countable class of affine halfspaces with rational coefficients, and set
\[
\mathcal G_0:=\{b_k\mathbf1_H:H\in\mathcal H_0,\ 1\le k\le K\},
\qquad
Z_{N,K}:=\sup_{g\in\mathcal G_0}|(\mathbb P_N-\mathbb P)g|.
\]
Every closed affine halfspace is the pointwise limit of halfspaces in $\mathcal H_0$: approximate its coefficients with error $o(n^{-1})$ and move its constant term outward by $n^{-1}$. This also includes boundary points. Bounded convergence then gives the claimed domination by $Z_{N,K}$ on the full-measure factor-bounded set.

On $N$ points, the number of distinct vectors $(g(Z_i))_{i\le N}$ is at most $K(N+1)^d$, by the VC growth bound. Symmetrization and the Rademacher maximal inequality therefore give
\[
\E Z_{N,K}
\le 2M\sqrt{\frac{2\log(2K(N+1)^d)}{N}}
\le C(M)\sqrt{\frac{K\log N}{N}},\qquad N\ge2.
\]
These inequalities apply to the countable class $\mathcal G_0$; see \cite[Secs.~2.2, 2.3, and~2.6]{vdvw1996}.
\end{proof}

\begin{proof}[Proof of \Cref{thm:indicator-finiteN}]\label{proof:thm:indicator-finiteN}
Let $\beta^N$ be a measurable selection satisfying the hypotheses of the theorem. The estimates below are pathwise and uniform over all such selections.

Let $\alpha=(\bar a_1,\dots,\bar a_K)$ and $\gamma=(\bar b_1,\dots,\bar b_K)$ denote the population means, and similarly $\alpha^N=(\bar a_1^N,\dots,\bar a_K^N)$ and $\gamma^N=(\bar b_1^N,\dots,\bar b_K^N)$ for the sample. For any $\alpha',\gamma'\in[-M,M]^K$ define
\[
F_k^{N,\alpha',\gamma'}(t,\beta):=\int g_{t,\beta,\alpha',\gamma',k}(z)\,\mathbb P_N(\dd z),
\qquad
F_k^{\alpha',\gamma'}(t,\beta):=\int g_{t,\beta,\alpha',\gamma',k}(z)\,\mathbb P(\dd z).
\]
Then $F_k^N=F_k^{N,\alpha^N,\gamma^N}$ and the deterministic feedback is $F_k=F_k^{\alpha,\gamma}$. Therefore
\begin{align}
\label{eq:indicator-F-split}
\sup_{t,\beta,k}|F_k^N(t,\beta)-F_k(t,\beta)|
&\le \sup_{t,\beta,k,\alpha',\gamma'} |F_k^{N,\alpha',\gamma'}(t,\beta)-F_k^{\alpha',\gamma'}(t,\beta)| \\
&\quad + \sup_{t,\beta,k}|F_k^{\alpha^N,\gamma^N}(t,\beta)-F_k^{\alpha,\gamma}(t,\beta)|.
\end{align}
By \cref{lem:indicator-vc}, the first term is bounded by the measurable variable $Z_{N,K}$, with
\begin{equation}
\label{eq:indicator-empirical-process}
\E Z_{N,K}\le C\sqrt{\frac{K\log N}{N}}.
\end{equation}
It remains to control the second term in \eqref{eq:indicator-F-split}, a population-measure comparison to which the bounded-density hypothesis applies. Set
\[
\eta_N:=\max_{1\le j\le K}\max\bigl\{|\bar a_j^N-\bar a_j|,\,|\bar b_j^N-\bar b_j|\bigr\}.
\]
For every $z=(x,a,b)$,
\begin{align*}
\abs{\Psi_t^{\alpha^N,\gamma^N}(z,\beta)-\Psi_t^{\alpha,\gamma}(z,\beta)}
&= \frac{|r|t}{K}\left|\sum_{j=1}^K \bigl(a_j(\bar b_j^N-\bar b_j)-b_j(\bar a_j^N-\bar a_j)\bigr)\right| \\
&\le 2|r|TM\,\eta_N=:\delta_N.
\end{align*}
By \cref{rem:density-sufficient}, the scalar threshold projection associated with any fixed $\alpha',\gamma'\in[-M,M]^K$ and $\beta\in\mathcal C_T$ has density bounded by $\bar M$. An indicator mismatch lies in the $\delta_N$-tube around the population threshold. Therefore
\[
\sup_{t,\beta,k}|F_k^{\alpha^N,\gamma^N}(t,\beta)-F_k^{\alpha,\gamma}(t,\beta)|
\le 2M\bar M\,\delta_N
\le 4|r|TM^2\bar M\,\eta_N.
\]
Since $|a_j|,|b_j|\le M$, Hoeffding's inequality and a union bound yield
\begin{equation}
\label{eq:indicator-mean-concentration}
\E\eta_N \le C\sqrt{\frac{\log(2K)}{N}} \le C\sqrt{\frac{K\log N}{N}},
\end{equation}
for a constant $C=C(M)$ and all $N\ge2$. Set $A_N:=Z_{N,K}+4|r|TM^2\bar M\eta_N$. Then
\begin{equation}
\label{eq:indicator-uniform-F-bound}
\sup_{t,\beta,k}|F_k^N(t,\beta)-F_k(t,\beta)|\le A_N,
\qquad
\E A_N\le C_T\sqrt{\frac{K\log N}{N}}.
\end{equation}

For $0\le t\le T$, set $\Delta_t^N:=\sup_{0\le s\le t}\norm{\beta^N(s)-\beta(s)}_\infty$. The two integral equations give
\begin{align*}
\norm{\beta^N(t)-\beta(t)}_\infty
&\le \int_0^t \norm{F^N(s,\beta^N(s))-F(s,\beta(s))}_\infty\,\dd s \\
&\le \int_0^t \sup_{\beta\in\mathcal C_T}\norm{F^N(s,\beta)-F(s,\beta)}_\infty\,\dd s
 + \int_0^t \norm{F(s,\beta^N(s))-F(s,\beta(s))}_\infty\,\dd s.
\end{align*}
By the proof of \cref{thm:indicator}, the deterministic indicator feedback map is Lipschitz on $\mathcal C_T$ with constant $L_T:=2M^2\bar M$. Hence
\[
\Delta_t^N \le t A_N+ L_T\int_0^t \Delta_s^N\,\dd s,\qquad 0\le t\le T.
\]
Pathwise Gronwall gives $\Delta_T^N\le T e^{L_TT}A_N$. Taking expectations and using \eqref{eq:indicator-uniform-F-bound} yields
\[
\E\Delta_T^N \le C_T\sqrt{\frac{K\log N}{N}}.
\]
This proves the feedback-path estimate.

For the state-law bound, first observe that the map $\Xi_t$ is Lipschitz in $z$ with respect to the $\ell^1$ metric on $\mathcal Z$, with a constant depending only on $T$, $M$, and $|r|$; indeed,
\[
|\Xi_t(z)-\Xi_t(\tilde z)| \le C_T |z-\tilde z|_1.
\]
Therefore
\begin{equation}
\label{eq:indicator-pushforward-w1}
\sup_{0\le t\le T} \Wone\bigl((\Xi_t)_\#\mu_0^N,(\Xi_t)_\#\mu_0\bigr)
\le C_T\Wone(\mu_0^N,\mu_0).
\end{equation}
Next, since $\Xi_t^N$ and $\Xi_t$ are evaluated on the same sampled atoms,
\begin{align*}
\Wone\bigl((\Xi_t^N)_\#\mu_0^N,(\Xi_t)_\#\mu_0^N\bigr)
&\le \frac1N\sum_{i=1}^N \abs{X_t^N(Z_i)-X_t(Z_i)} \\
&\le t\sup_{1\le i\le N}\abs{\Lambda^N(Z_i)-\Lambda(Z_i)}
+M\norm{\beta^N(t)-\beta(t)}_\infty \\
&\le 2|r|TM\eta_N+M\norm{\beta^N(t)-\beta(t)}_\infty.
\end{align*}
The expectation of the additional imbalance term is controlled by \eqref{eq:indicator-mean-concentration} and is absorbed by the same $\sqrt{K\log N/N}$ rate. Taking the supremum over $t\in[0,T]$, then expectations, and combining with \eqref{eq:indicator-pushforward-w1} proves the final bound.
\end{proof}

\subsection{Proofs for the kernel results}
\label{app:proofs-sec4}

\begin{proof}[Proof of \Cref{thm:graphon-wellposed}]\label{proof:thm:graphon-wellposed}
Fix $T>0$. Set
\[
X_t^{(0)}(u):=x_0(u)+\mu t+r tR_W(u)
\]
and define inductively
\[
X_t^{(n+1)}(u):=x_0(u)+\mu t+r tR_W(u)-\int_0^t (\Gamma_W X_s^{(n)})(u)\,\dd s.
\]
Each Picard iterate is continuous as an $L^\infty(I)$-valued path and hence strongly measurable; the time integral below is therefore a Bochner integral in $L^\infty(I)$. Since
\[
\abs{R_W(u)}\le 2\norm{W}_{L^\infty},
\qquad
\abs{(\Gamma_W\varphi)(u)}
\le \int_I \abs{W(u,v)}\,\abs{\ell(\varphi(v))}\,\dd v
\le \norm{W}_{L^\infty}\ell_\ast,
\]
each iterate belongs to \(C([0,T];L^\infty(I))\) and is uniformly bounded there. Moreover,
\begin{align*}
\norm{\Gamma_W\varphi-\Gamma_W\psi}_{L^\infty}
&\le \sup_{u\in I}\int_I \abs{W(u,v)}\,\abs{\ell(\varphi(v))-\ell(\psi(v))}\,\dd v \\
&\le \norm{W}_{L^\infty}L_\ell\norm{\varphi-\psi}_{L^1}
\le \norm{W}_{L^\infty}L_\ell\norm{\varphi-\psi}_{L^\infty}.
\end{align*}
Hence, with
\[
\Delta_t^{(n)}:=\sup_{0\le s\le t}\norm{X_s^{(n+1)}-X_s^{(n)}}_{L^\infty},
\]
we obtain
\[
\Delta_t^{(n)}\le \norm{W}_{L^\infty}L_\ell\int_0^t \Delta_s^{(n-1)}\,\dd s.
\]
A standard factorial estimate shows that $(X^{(n)})_n$ is Cauchy in $C([0,T];L^\infty(I))$. The limit therefore belongs to $C([0,T];L^\infty(I))$ and solves \eqref{eq:graphon-equation}. Uniqueness in $C([0,T];L^\infty(I))$ follows from the same estimate applied to the difference of two solutions.

For stability, write
\begin{align*}
X_t^1(u)-X_t^2(u)
&=x_0^1(u)-x_0^2(u)+rt\bigl(R_{W_1}(u)-R_{W_2}(u)\bigr) \\
&\qquad -\int_0^t \bigl((\Gamma_{W_1}X_s^1)(u)-(\Gamma_{W_2}X_s^2)(u)\bigr)\,\dd s.
\end{align*}
Using
\[
\norm{R_{W_1}-R_{W_2}}_{L^1}\le 2\norm{W_1-W_2}_{L^1(I^2)}
\]
and
\[
\norm{\Gamma_{W_1}\varphi-\Gamma_{W_2}\psi}_{L^1}
\le \ell_\ast\norm{W_1-W_2}_{L^1(I^2)} + M_WL_\ell\norm{\varphi-\psi}_{L^1},
\]
we find
\[
\norm{X_t^1-X_t^2}_{L^1}
\le \norm{x_0^1-x_0^2}_{L^1} + (2|r|T+T\ell_\ast)\norm{W_1-W_2}_{L^1(I^2)}
+ M_WL_\ell\int_0^t \norm{X_s^1-X_s^2}_{L^1}\,\dd s.
\]
Gronwall's lemma yields \eqref{eq:graphon-stability}.
\end{proof}

\begin{proof}[Proof of \Cref{thm:bridge}]\label{proof:thm:bridge}
Because the rank-$K$ graphon equation with kernel $W^{(K)}$ is exactly the continuum version of the rank-$K$ factor model, \cref{thm:finite-rank-w1} applied to the pair $(\mu_0^{N,K},\mu_0^{(K)})$ gives
\[
\sup_{0\le t\le T}\Wone(\nu_t^{N,K},\nu_t^K)\le C_T\,\Wone(\mu_0^{N,K},\mu_0^{(K)}),
\]
after projecting the joint law of $(X_t,a,b)$ onto the state coordinate. Next, couple $\nu_t^K$ and $\nu_t$ by the common random variable $U$. Then
\[
\Wone(\nu_t^K,\nu_t)\le \E\abs{X_t^K(U)-X_t(U)}=\norm{X_t^K-X_t}_{L^1(I)},
\]
and \cref{thm:graphon-wellposed} yields
\[
\sup_{0\le t\le T}\Wone(\nu_t^K,\nu_t)
\le C_T\Bigl(\norm{W^{(K)}-W}_{L^1(I^2)}+\norm{x_0^{(K)}-x_0}_{L^1(I)}\Bigr).
\]
The claim follows from the triangle inequality.
\end{proof}

\begin{proof}[Proof of \Cref{prop:graphon-indicator-smoothing}]\label{proof:prop:graphon-indicator-smoothing}
For each $K$, the graphon model with kernel $W^{(K)}$ is exactly the continuum version of the rank-$K$ factor model generated by $\mu_0^{(K)}$. Let $\beta^{K,\varepsilon}$ and $\beta^{K,\mathrm{ind}}$ denote the corresponding feedback coordinates for the smoothed and indicator losses. Writing $F_k^{K,\varepsilon}$ and $F_k^{K,\mathrm{ind}}$ for the associated feedback right-hand sides, observe that for every \(y\in\R\),
\[
\abs{\ell_\varepsilon(y)-\mathbf 1_{\{y\le0\}}}
\le \mathbf 1_{\{0<y<\varepsilon\}}.
\]
Together with the uniform density hypothesis, this implies
\[
\sup_{\beta\in[-MT,MT]^K}\abs{F_k^{K,\varepsilon}(t,\beta)-F_k^{K,\mathrm{ind}}(t,\beta)}
\le M M_\rho(T,MT)\,\varepsilon,
\qquad 1\le k\le K.
\]
By \cref{thm:indicator}, the indicator field is Lipschitz on $[-MT,MT]^K$ with constant $2M^2M_\rho(T,MT)$, uniformly in $K$. Therefore
\[
\norm{\beta^{K,\varepsilon}(t)-\beta^{K,\mathrm{ind}}(t)}_\infty
\le M M_\rho(T,MT)\,t\,\varepsilon + 2M^2M_\rho(T,MT)\int_0^t \norm{\beta^{K,\varepsilon}(s)-\beta^{K,\mathrm{ind}}(s)}_\infty\,\dd s,
\]
so Gronwall's lemma yields
\[
\sup_{0\le t\le T}\norm{\beta^{K,\varepsilon}(t)-\beta^{K,\mathrm{ind}}(t)}_\infty\le C_T\varepsilon
\]
with $C_T$ independent of $K$. Returning to the state map and using $\max_k\norm{a_k^{(K)}}_{L^\infty}\le M$, we obtain
\begin{align*}
\sup_{0\le t\le T}\norm{X_t^{K,\varepsilon}-X_t^{K,\mathrm{ind}}}_{L^1(I)}
&\le \sup_{0\le t\le T}\norm{X_t^{K,\varepsilon}-X_t^{K,\mathrm{ind}}}_{L^\infty(I)} \\
&\le M\sup_{0\le t\le T}\norm{\beta^{K,\varepsilon}(t)-\beta^{K,\mathrm{ind}}(t)}_\infty \\
&\le C_T\varepsilon.
\end{align*}
The Wasserstein bound follows by coupling both laws through the same $U\sim\mathrm{Unif}(I)$.
\end{proof}

\begin{proof}[Proof of \Cref{prop:block-L1-approx}]\label{proof:prop:block-L1-approx}
The bound $\norm{\Pi_KW}_{L^\infty}\le \norm{W}_{L^\infty}$ is immediate from Jensen's inequality on each block average. Since $\Pi_KW$ is constant on each rectangle $I_i^{(K)}\times I_j^{(K)}$, the associated integral operator maps $L^2(I)$ into the $K$-dimensional space of functions that are constant on the partition $\{I_j^{(K)}\}_{j=1}^K$; hence its rank is at most $K$. 

To prove $L^1$ convergence, fix $\varepsilon>0$ and choose $g\in C(I^2)$ such that $\norm{W-g}_{L^1(I^2)}<\varepsilon$. The block-averaging operator is an $L^1$ contraction, so
\[
\norm{\Pi_K(W-g)}_{L^1(I^2)}\le \norm{W-g}_{L^1(I^2)}<\varepsilon.
\]
Uniform continuity of $g$ on the compact square implies $\norm{\Pi_Kg-g}_{L^1(I^2)}\to0$. Consequently,
\[
\norm{\Pi_KW-W}_{L^1(I^2)}
\le 2\varepsilon+\norm{\Pi_Kg-g}_{L^1(I^2)},
\]
and the claim follows by first taking $K\to\infty$ and then $\varepsilon\downarrow0$.

If $W$ is piecewise $C^1$ on finitely many rectangles with bounded first derivatives, then the cells entirely contained in one smooth rectangle contribute $O(K^{-1})$ by the mean-value theorem, while the cells intersecting the finitely many rectangle boundaries occupy total area $O(K^{-1})$ and contribute at most $2\norm{W}_{L^\infty}$ there. Summing these two contributions yields $\norm{\Pi_KW-W}_{L^1(I^2)}=O(K^{-1})$.
\end{proof}

\begin{proof}[Proof of \Cref{thm:factorized-indicator}]\label{proof:thm:factorized-indicator}
If $B_\ast=0$, then $F_t\equiv0$ and $C\equiv0$. Otherwise set
\[
A_\ast:=\norm{a}_{L^\infty},\qquad \nu_\ast:=\nu(\Theta),\qquad
\rho_\ast:=M_\rho^\Theta(T,TB_\ast),\qquad
D_\ast:=|\mu|+2|r|\norm{W}_{L^\infty}.
\]
For $c,\widetilde c$ bounded by $TB_\ast$ and $s,t\in[0,T]$,
\[
|\Psi_t(u,c)-\Psi_s(u,\widetilde c)|
\le A_\ast\norm{c-\widetilde c}_{L^1(\Theta,\nu)}+D_\ast|t-s|=:\delta
\quad\text{for a.e. }u.
\]
An indicator mismatch lies in $\{|\Psi_t(u,c)|\le\delta\}$, whose measure is at most $2\rho_\ast\delta$. Consequently
\begin{equation}
\label{eq:factorized-feedback-joint-continuity}
\norm{F_t(c)-F_s(\widetilde c)}_{L^1(\Theta,\nu)}
\le 2\nu_\ast B_\ast\rho_\ast
\bigl(A_\ast\norm{c-\widetilde c}_{L^1(\Theta,\nu)}+D_\ast|t-s|\bigr).
\end{equation}
In particular, $F_t(c)$ is jointly norm-continuous on this ball.

Work on the closed subset
\[
\mathcal B_T^\Theta:=
\{C\in C([0,T];L^1(\Theta,\nu)):|C_t|\le tB_\ast\ \nu\text{-a.e. for every }t\}.
\]
Closedness follows because each order interval $\{|c|\le tB_\ast\}$ is closed in $L^1$. For $C\in\mathcal B_T^\Theta$, \eqref{eq:factorized-feedback-joint-continuity} makes $s\mapsto F_s(C_s)$ continuous in $L^1$. Thus
\[
(\mathcal TC)_t:=\int_0^t F_s(C_s)\,\dd s
\]
is a Bochner integral, and $|F_s(C_s)|\le B_\ast$ shows that $\mathcal T$ preserves $\mathcal B_T^\Theta$. With $L=2\nu_\ast B_\ast\rho_\ast A_\ast$, successive Picard differences satisfy
\[
\sup_{s\le t}\norm{C_s^{n+1}-C_s^n}_{L^1}
\le L\int_0^t\sup_{v\le s}\norm{C_v^n-C_v^{n-1}}_{L^1}\,\dd s.
\]
The resulting factorial bound gives a fixed point, and the same estimate with Gronwall's lemma proves uniqueness.

The map $c\mapsto\int_\Theta a(\cdot,\theta)c(\theta)\nu(\dd\theta)$ has norm at most $A_\ast$ from $L^1(\Theta,\nu)$ to $L^\infty(I)$. Therefore \eqref{eq:factorized-graphon} defines $X\in C([0,T];L^\infty(I))$, and Fubini's theorem gives the indicator graphon equation. Conversely, any graphon solution defines $C$ by \eqref{eq:Ctheta}; this path belongs to $\mathcal B_T^\Theta$ and solves the same feedback equation. Uniqueness of $C$ then gives uniqueness of $X$.
\end{proof}

\begin{proof}[Proof of \Cref{thm:graphon-indicator-general}]\label{proof:thm:graphon-indicator-general}
Set $\rho_\ast=M_\rho^W(T,T)$, $M_W=\norm W_{L^\infty}$, and $D_\ast=|\mu|+2|r|M_W$. Define
\[
G_t(h)(u):=\mathbf1_{\{\Psi_t(u,h)\le0\}}.
\]
For $h,\widetilde h$ bounded by $T$,
\[
|\Psi_t(u,h)-\Psi_s(u,\widetilde h)|
\le M_W\norm{h-\widetilde h}_{L^1}+D_\ast|t-s|=:\delta
\quad\text{for a.e. }u.
\]
An indicator mismatch lies in $\{|\Psi_t(u,h)|\le\delta\}$. The density assumption gives
\[
\norm{G_t(h)-G_s(\widetilde h)}_{L^1}
\le2\rho_\ast\bigl(M_W\norm{h-\widetilde h}_{L^1}+D_\ast|t-s|\bigr).
\]
Thus $G$ is jointly norm-continuous and Lipschitz in $h$ on the admissible ball.

Let
\[
\mathcal B_T:=\{H\in C([0,T];L^1(I)):0\le H_t\le t\text{ a.e. for every }t\}.
\]
This is closed in $C([0,T];L^1(I))$. For $H\in\mathcal B_T$, the path $s\mapsto G_s(H_s)$ is continuous in $L^1$, so
\[
(\mathcal TH)_t:=\int_0^t G_s(H_s)\,\dd s
\]
is a Bochner integral. Since $0\le G_s(H_s)\le1$, the map preserves $\mathcal B_T$ and its images are $1$-Lipschitz in time. Picard iteration from zero satisfies
\[
\Delta_n(t)\le2\rho_\ast M_W\int_0^t\Delta_{n-1}(s)\,\dd s,
\qquad
\Delta_n(t):=\sup_{s\le t}\norm{H_s^{n+1}-H_s^n}_{L^1}.
\]
The factorial estimate gives a fixed point in $\mathcal B_T$, and Gronwall's lemma gives uniqueness.

Define $X$ by \eqref{eq:X-from-H-general}. Boundedness of $W$ makes $h\mapsto Wh$ continuous from $L^1(I)$ to $L^\infty(I)$, hence $X\in C([0,T];L^\infty(I))$. Choose a jointly measurable representative of $X$ and set
\[
\widehat H_t(u):=\int_0^t\mathbf1_{\{X_s(u)\le0\}}\,\dd s.
\]
The fixed-point identity shows that $\widehat H_t=H_t$ in $L^1$ for every $t$. This representative satisfies $0\le\widehat H_t(u)\le t$ on a common full-measure set. Fubini's theorem now gives \eqref{eq:graphon-equation}. Conversely, every indicator graphon solution produces a fixed point through \eqref{eq:Ht-definition}, proving uniqueness of the state profile.
\end{proof}

\begin{proof}[Proof of \Cref{cor:graphon-indicator-initial-stability}]\label{proof:cor:graphon-indicator-initial-stability}
Write $\eta=\norm{x_0^1-x_0^2}_{L^\infty}$, $M_W=\norm W_{L^\infty}$, and $\rho_\ast=M_\rho^W(T,T)$. Along the two solutions,
\[
|X_t^1(u)-X_t^2(u)|\le\eta+M_W\norm{H_t^1-H_t^2}_{L^1}
\quad\text{for a.e. }u.
\]
A mismatch of the indicators therefore lies in the tube around $X_t^1$ of this width. Its measure is at most $2\rho_\ast(\eta+M_W\norm{H_t^1-H_t^2}_{L^1})$, giving
\[
\norm{H_t^1-H_t^2}_{L^1}
\le2\rho_\ast\eta t+2\rho_\ast M_W\int_0^t\norm{H_s^1-H_s^2}_{L^1}\,\dd s.
\]
Gronwall's lemma yields the cumulative-profile bound. The displayed pointwise estimate then gives the state bound.
\end{proof}

\begin{proof}[Proof of \Cref{prop:graphon-indicator-transverse}]\label{proof:prop:graphon-indicator-transverse}
Fix $t\in[0,T]$ and $h\in L^\infty(I)$ with $\norm{h}_{L^\infty}\le C$. On a branch $I_j$, the $L^1$--$C^1$ assumptions imply
\[
\frac{\dd}{\dd u}\int_I W(u,v)h(v)\,\dd v
=\int_I D_1W(u,v)h(v)\,\dd v
\]
and
\[
\frac{\dd}{\dd u}R_W(u)
=\int_I D_1W(u,v)\,\dd v-
\int_I D_2W(v,u)\,\dd v.
\]
Thus $u\mapsto\Psi_t(u,h)$ is $C^1$ on $I_j$ and
\[
\partial_u\Psi_t(u,h)
=x_0'(u)+rt\,\partial_uR_W(u)-
\int_I D_1W(u,v)h(v)\,\dd v.
\]
The $L^1$ derivative bounds give
\begin{align*}
\abs{\partial_u\Psi_t(u,h)}
&\ge \abs{x_0'(u)}-|r|t(A_1+A_2)-C A_1\\
&\ge m_0-|r|T(A_1+A_2)-C A_1=m>0.
\end{align*}
Continuity and the strict lower bound imply that the derivative has constant sign on each branch, so $\Psi_t(\cdot,h)$ has at most $J$ monotone branches. The branchwise change-of-variables argument from \cref{prop:factorized-density-transverse} then yields the density bound $J/m$ for $\Psi_t(U,h)$, $U\sim\lambda$.
\end{proof}

\begin{proof}[Proof of \Cref{thm:uniformly-transverse-stability}]\label{proof:thm:uniformly-transverse-stability}
Branchwise change of variables gives the common density bound $J/m$. Hence \cref{thm:graphon-indicator-general} supplies existence and uniqueness for every member of the family. Along each resulting path, \cref{prop:level-set-control} gives
\[
\sup_{t\le T}\lambda\{|X_t^\alpha|\le a\}\le\frac{2J}{m}a.
\]
For two members, the $p=\infty$ case of \cref{thm:indicator-one-sided-stability} applies with $L=2J/m$, $M_1\le M_W$, and
\[
\delta_\infty=
\norm{x_0^{\alpha_1}-x_0^{\alpha_2}}_{L^\infty}
+T(1+2|r|)\norm{W^{\alpha_1}-W^{\alpha_2}}_{L^\infty}.
\]
It gives \eqref{eq:two-dataset-H-stability} and an $L^\infty$ state bound. Since $\lambda(I)=1$, the latter implies \eqref{eq:two-dataset-X-stability}.
\end{proof}

\subsection{Elementary facts for the model derivation}
\label{app:derivation-proofs}

Here $e_{ij}\ge0$ are the bilateral exposures, $\theta\in(0,1]$ is the loss rate, $\mathbb D_0=\{i:x_i\le0\}$, and $X^{(k)},\mathbb D_k$ are defined by \eqref{eq:cascade}.

\begin{proposition}\label{prop:cascade-maxsol}
The default sets are nondecreasing, $\mathbb D_k\subseteq\mathbb D_{k+1}$ for all $k\ge0$, and the buffers satisfy $X^{(k+1)}\le X^{(k)}$ for $k\ge1$. The default sets are constant from round $N$ onward and the buffers from round $N+1$ onward. The terminal vector $X^{(N+1)}$ is the greatest solution of \eqref{eq:cascade-fixedpoint} in the coordinatewise order on $\R^N$.
\end{proposition}

\begin{proof}
The coordinates in $\mathbb D_0$ remain equal to $x_i\le0$, so $\mathbb D_0\subseteq\mathbb D_1$. If $\mathbb D_{k-1}\subseteq\mathbb D_k$, then for $i\notin\mathbb D_0$,
\[
X_i^{(k+1)}-X_i^{(k)}
=-\theta\sum_{j=1}^N e_{ij}
\bigl(\1_{\mathbb D_k}(j)-\1_{\mathbb D_{k-1}}(j)\bigr)\le0.
\]
The same inequality holds with equality on $\mathbb D_0$. Hence $\mathbb D_k\subseteq\mathbb D_{k+1}$, proving both monotonicity claims by induction. Once two consecutive default sets agree, every subsequent update agrees as well. There can be at most $N$ strict inclusions, so $\mathbb D_N=\mathbb D_{N+1}$. Consequently the update for $X^{(N+1)}$ uses its own default set and satisfies \eqref{eq:cascade-fixedpoint}; all later buffer vectors coincide with it.

To prove maximality, let $Y$ be any solution of \eqref{eq:cascade-fixedpoint}. Then $Y_i=x_i$ on $\mathbb D_0$, so $\{j:Y_j\le0\}$ contains $\mathbb D_0$. Nonnegative exposures give $Y\le X^{(1)}$. If $Y\le X^{(k)}$, then
$\1_{\{Y_j\le0\}}\ge\1_{\mathbb D_k}(j)$ for every $j$, and the same comparison in the fixed-point equation gives $Y\le X^{(k+1)}$. Induction yields $Y\le X^{(N+1)}$, as required.
\end{proof}

\begin{proposition}\label{prop:mechanism-wellposed}
The stopped system \eqref{eq:mechanism-ode} admits a unique solution.
\end{proposition}

\begin{proof}
For a set $D$ of institutions already in default, define the constant solvent drift
\[
c_i(D):=\mu_i+r\sum_{j=1}^N(e_{ij}-e_{ji})
-\theta r\sum_{j\in D}e_{ij}.
\]
Start at $t_0=0$ with $D=\mathbb D_0$ and states $x_i$. Institutions in $\mathbb D_0$ remain at $x_i$ for the whole horizon. Suppose a path has been constructed through time $t_m<T$, all coordinates in $D$ are frozen, and the remaining states $y_i=X_{t_m}^i$ are strictly positive. Their next possible threshold time is
\[
t_{m+1}:=\min\left\{T,\ t_m+
\min_{\substack{i\notin D\\c_i(D)<0}}\frac{y_i}{-c_i(D)}\right\},
\]
where the inner minimum over an empty set is $+\infty$. On $[t_m,t_{m+1}]$ set
$X_t^i=y_i+c_i(D)(t-t_m)$ for $i\notin D$, leaving the other coordinates frozen. At $t_{m+1}<T$, add all newly zero coordinates to $D$, freeze them at zero, and repeat. If $D$ contains every institution, extend all coordinates constantly to $T$. Every event before $T$ adds at least one institution to $D$, so there are at most $N$ such events.

On each open interval between events, the distress indicators equal $\1_D$. Their values at event times do not affect the time integrals. The constructed path therefore satisfies \eqref{eq:mechanism-ode}, with each $\tau_i$ equal to its first threshold time, capped at $T$.

For uniqueness, any solution has the same frozen initial coordinates. Before the first new threshold crossing its remaining coordinates have the prescribed drifts $c_i(\mathbb D_0)$, so both the first event time and the institutions reaching zero are fixed by the displayed formula. Repeating this argument after each event identifies the solution on the whole horizon.
\end{proof}

\subsection{Indicator solutions and stability}
\label{app:indicator-completion}

\begin{proof}[Proof of \Cref{cor:indicator-canonical-sample}]
\label{proof:cor:indicator-canonical-sample}
Let
\[
G_N:=\left\{(Z_1,\ldots,Z_N):
\begin{array}{l}
|a_{i,k}^N|,|b_{i,k}^N|\le M\quad(1\le i\le N,\ 1\le k\le K),\\
e_{ij}^N\ge0\quad(1\le i,j\le N)
\end{array}\right\}.
\]
This is a Borel set of full probability. Fix a sample in $G_N$ and write
\[
q_i(t):=x_i^N+\mu t+\frac{rt}{N}\sum_{j=1}^N(e_{ij}^N-e_{ji}^N),
\qquad
X_i(t)=q_i(t)-\frac1N\sum_{j=1}^Ne_{ij}^NH_j(t).
\]
On the finite probability space define
\[
(\mathcal T_NH)_i(t):=\int_0^t
\1_{\{q_i(s)-N^{-1}\sum_je_{ij}^NH_j(s)\le0\}}\,\dd s.
\]
Entrywise nonnegativity makes $\mathcal T_N$ order preserving. Starting from $H_i^{(0)}(t)=t$ therefore gives a decreasing sequence $H^{(n+1)}=\mathcal T_NH^{(n)}$. Its limit is a fixed point by the monotone argument in the proof of \cref{thm:indicator-greatest} below and dominates every other fixed point.

Every finite iterate is jointly Borel in the sample and time and has continuous, $1$-Lipschitz paths. Evaluation at rational times therefore makes it a Borel $C([0,T];\R^N)$-valued map. The common Lipschitz bound upgrades the pointwise monotone limit to uniform convergence in time. Extending the iterates and limit by zero on $G_N^c$ gives a Borel path on the whole sample space. Its feedback $\beta^{N,+}$ is Borel and lies in $[-MT,MT]^K$ for every sample. On $G_N$, \cref{lem:reformulation} identifies it with a solution of the sampled rank-$K$ equation, so \cref{thm:indicator-finiteN} applies almost surely.
\end{proof}

\begin{proof}[Proof of \Cref{thm:indicator-greatest}]
\label{proof:thm:indicator-greatest}
If $H\le K$, nonnegativity of $W$ gives $WH\le WK$, and therefore
\[
\Psi_t(\cdot,H)\ge\Psi_t(\cdot,K),
\qquad
\mathcal T_WH\le\mathcal T_WK.
\]
Every image of $\mathcal T_W$ lies below the top profile $H_t^{(0)}=t$, so the upper iterates decrease. Define
\[
H_t^+(u):=\inf_{n\ge0}H_t^{(n)}(u).
\]
All iterates are jointly measurable and pointwise $1$-Lipschitz in time; these properties pass to $H^+$. Taking the countable intersection of the full-measure order sets at rational times and then using the common Lipschitz bound gives one full-measure set on which all time-indexed inequalities hold.

Since $H^{(n)}\downarrow H^+$ and $W\ge0$, monotone convergence gives
\[
WH^{(n)}\downarrow WH^+,
\qquad
\Psi(\cdot,H^{(n)})\uparrow\Psi(\cdot,H^+).
\]
For every scalar sequence $a_n\uparrow a$,
\[
\1_{\{a_n\le0\}}\downarrow\1_{\{a\le0\}}.
\]
This includes $a=0$, since then $a_n\le0$ for every $n$. Dominated convergence in time yields
\[
H^+=\lim_nH^{(n+1)}
=\lim_n\mathcal T_WH^{(n)}
=\mathcal T_WH^+.
\]
Thus $H^+$ satisfies the hard-indicator integral equation. If $H$ is any other fixed point, then $H\le H^{(0)}$ and order preservation gives $H\le H^{(n)}$ for every $n$, hence $H\le H^+$. Applying the nonnegative operator $W$ shows that the associated state for $H^+$ is the smallest.

Pointwise convergence and $0\le H^{(n)}-H^+\le T$ give $L^p$ convergence at every time. The functions $t\mapsto\norm{H_t^{(n)}-H_t^+}_{L^p}$ are uniformly equicontinuous, because both profiles are pointwise $1$-Lipschitz. A finite time net upgrades the convergence to $C([0,T];L^p)$. Finally,
\[
\sup_{t\le T}\norm{W(H_t^{(n)}-H_t^+)}_{L^\infty}
\le\norm W_{L^\infty}
\sup_{t\le T}\norm{H_t^{(n)}-H_t^+}_{L^1},
\]
which gives the asserted state regularity and convergence.
\end{proof}

\begin{proof}[Proof of \Cref{thm:indicator-positive-ramp}]
\label{proof:thm:indicator-positive-ramp}
For fixed $\varepsilon$, the ramp feedback is globally Lipschitz in $L^1$ and order preserving, so it has a unique fixed point $H^{\varepsilon,+}$. Starting the corresponding order iteration at the top profile gives a decreasing sequence converging to that fixed point. Since
\[
0<\varepsilon'<\varepsilon
\quad\Longrightarrow\quad
\ell_{\varepsilon'}(x)\le\ell_\varepsilon(x)
\quad\text{for all }x,
\]
comparison of the upper iterations yields $H^{\varepsilon',+}\le H^{\varepsilon,+}$.

Every hard solution $H$ is a subsolution of the regularized map because $\ell_\varepsilon\ge\1_{\{x\le0\}}$. Iterating the regularized map from $H$ and using uniqueness gives $H\le H^{\varepsilon,+}$; in particular, $H^+\le H^{\varepsilon,+}$. Fix a deterministic sequence $\varepsilon_m\downarrow0$ and let
\[
\bar H:=\lim_{m\to\infty}H^{\varepsilon_m,+},
\qquad
\bar X:=x_0+\mu t+rtR_W-W\bar H.
\]
Then $\bar H\ge H^+$ and the regularized states increase to $\bar X$. If $\bar X(t,u)\le0$, every regularized state at $(t,u)$ is nonpositive and the ramp equals one. If $\bar X(t,u)>0$, then $X^{\varepsilon_m,+}(t,u)>\varepsilon_m$ for all sufficiently large $m$, so the ramp equals zero. Therefore
\[
\ell_{\varepsilon_m}(X^{\varepsilon_m,+})
\longrightarrow\1_{\{\bar X\le0\}}
\quad\text{pointwise a.e.}
\]
Dominated convergence makes $\bar H$ a hard solution. Maximality gives $\bar H\le H^+$, hence equality. The norm and state convergences along $\varepsilon_m$ follow exactly as in the preceding proof. Finally, if $0<\varepsilon<\varepsilon_m$, monotonicity gives
\[
0\le H^{\varepsilon,+}-H^+\le H^{\varepsilon_m,+}-H^+.
\]
This proves convergence of the full family as $\varepsilon\downarrow0$.
\end{proof}

\begin{proof}[Proof of \Cref{prop:indicator-osgood}]
\label{proof:prop:indicator-osgood}
If $M_W=0$, the state is independent of the cumulative profile, so its indicator determines $H$ uniquely. Assume $M_W>0$. Let $(K,Y)$ be another solution and set $d_t:=\norm{H_t-K_t}_{L^1}$. The common data give
\[
\abs{X_t(u)-Y_t(u)}\le M_Wd_t.
\]
Whenever the two indicators differ, $\abs{X_t(u)}\le\abs{X_t(u)-Y_t(u)}$. Hence
\[
\norm{\1_{\{X_t\le0\}}-\1_{\{Y_t\le0\}}}_{L^1}
\le\omega(M_Wd_t),
\qquad
d_t\le\int_0^t\omega(M_Wd_s)\,\dd s.
\]
The Bihari--Osgood lemma \cite{bihari1956} and the stated integral condition give $d\equiv0$, and then $X=Y$.
\end{proof}

\begin{proof}[Proof of \Cref{prop:indicator-osgood-sharp}]
\label{proof:prop:indicator-osgood-sharp}
Let $F(a)=\lambda\{u:x_0(u)\le a\}$ and put $\beta(t):=\int_IH_t(u)\,\dd u$. For $W\equiv c$ and $\mu=r=0$, the state is
\[
X_t(u)=x_0(u)-c\beta(t).
\]
Fubini's theorem therefore gives
\[
\dot\beta(t)=F(c\beta(t)),
\qquad \beta(0)=0.
\]
Conversely, every absolutely continuous solution of this scalar equation reconstructs a hard solution through
\[
H_t(u)=\int_0^t\1_{\{x_0(u)\le c\beta(s)\}}\,\dd s.
\]
The scalar Osgood criterion proves \eqref{eq:indicator-rank-one-osgood-iff}. If the integral is finite, define
\[
\Phi(y):=\int_0^y\frac{\dd s}{F(cs)}.
\]
Since $F(cy)>0$ for every $y>0$ and $F\le1$, the function $\Phi$ is finite, strictly increasing, and unbounded. For every $\tau\in[0,T)$,
\[
\beta_\tau(t)=
\begin{cases}
0,&t\le\tau,\\
\Phi^{-1}(t-\tau),&t>\tau
\end{cases}
\]
is a solution on $[0,T]$. To realize a prescribed non-Osgood modulus $\omega$, extend its restriction near zero to a continuous distribution function on a bounded interval and take the generalized inverse as $x_0$. The zero-feedback solution then has $X_t=x_0$ and tube mass $\omega(a)$ near zero, while the delayed solutions give nonuniqueness.
\end{proof}

\begin{proof}[Proof of \Cref{prop:indicator-primitive-osgood}]
\label{proof:prop:indicator-primitive-osgood}
Fix an admissible profile $h$, a time $t$, a branch $I_j$, and an ordered
pair $u<v$ outside the null sets in
\eqref{eq:indicator-primitive-initial}--\eqref{eq:indicator-primitive-network}.
Since $0\le h_t\le t\le T$,
\begin{align*}
\sigma_j\bigl(X_t^h(v)-X_t^h(u)\bigr)
={}&\sigma_j\bigl(x_0(v)-x_0(u)\bigr)
+\sigma_jrt\bigl(R_W(v)-R_W(u)\bigr)\\
&-\sigma_j\int_I
\bigl(W(v,z)-W(u,z)\bigr)h_t(z)\,\dd z\\
\ge{}&\chi_j(v-u)-|r|T\abs{R_W(v)-R_W(u)}\\
&-T\norm{W(v,\cdot)-W(u,\cdot)}_{L^1(I)}\\
\ge{}&(1-\theta_j)\chi_j(v-u).
\end{align*}
The estimate and its null set are uniform over admissible $h$ and $t$.

Fix $q\in\R$ and $a\ge0$, and set
$E_j=\{u\in I_j:\abs{X_t^h(u)-q}\le a\}$.  For almost every ordered
pair $u<v$ in $E_j$,
\[
(1-\theta_j)\chi_j(v-u)
\le\abs{X_t^h(v)-X_t^h(u)}\le2a,
\]
so
\[
v-u\le\chi_j^{\leftarrow}\!\left(\frac{2a}{1-\theta_j}\right).
\]
If $E_j$ has positive measure, let $\alpha_j$ and $\beta_j$ be its
essential infimum and supremum.  For every $\varepsilon>0$, the portions of
$E_j$ within $\varepsilon$ of these two endpoints have positive measure.
The preceding almost-everywhere pairwise bound yields
\[
\beta_j-\alpha_j-2\varepsilon
\le\chi_j^{\leftarrow}\!\left(\frac{2a}{1-\theta_j}\right).
\]
Letting $\varepsilon\downarrow0$ and using
$\lambda(E_j)\le\beta_j-\alpha_j$ proves the branch estimate.  Summing over
$j$ gives \eqref{eq:indicator-primitive-uniform-tube}.

The capped inverse of a continuous strictly increasing function is
continuous, nondecreasing, zero at zero, and positive away from zero.  Thus
$\omega_{\rm pr}$ is an admissible tube modulus.  If
$M_W=\norm W_{L^\infty}>0$, then
\[
\int_{0+}\frac{\dd s}{\omega_{\rm pr}(M_Ws)}
=\frac1{M_W}\int_{0+}\frac{\dd a}{\omega_{\rm pr}(a)}=\infty.
\]
Apply \cref{prop:indicator-osgood}.  If $M_W=0$, the feedback is absent;
if $W\ge0$, existence follows from \cref{thm:indicator-greatest}.
\end{proof}

\begin{proof}[Proof of \Cref{cor:indicator-primitive-log-osgood}]
\label{proof:cor:indicator-primitive-log-osgood}
If
$\chi(s)\ge cs[\log(A/s)]^{-\gamma}$ for small $s$, choose $C_0>1/c$
and set $q(y)=C_0y[\log(B/y)]^\gamma$, with $B$ fixed sufficiently large.
For all small $y$, the point $q(y)$ lies in the valid range and
$\log(A/q(y))\le\log(B/y)$.  Hence
\[
\chi(q(y))
\ge cC_0y\frac{[\log(B/y)]^\gamma}
{[\log(A/q(y))]^\gamma}>y.
\]
Strict monotonicity gives
$\chi^{\leftarrow}(y)\le C_0y[\log(B/y)]^\gamma$.  Apply this estimate to
$2a/(1-\theta_j)$ and sum over the finitely many branches to obtain
\eqref{eq:indicator-primitive-log-tube}.  Since $\bar\gamma\le1$,
\[
\int_{0+}\frac{\dd a}{a[\log(A/a)]^{\bar\gamma}}=\infty.
\]
This proves the Osgood assertion and the stated consequences.
\end{proof}

\begin{proof}[Proof of \Cref{thm:indicator-osgood-stability}]
\label{proof:thm:indicator-osgood-stability}
Write $\Delta W:=W_1-W_2$ and
\[
A_t:=(x_0^1-x_0^2)+rt(R_{W_1}-R_{W_2})-(\Delta W)H_t^2.
\]
Jensen's inequality and $0\le H_t^2\le T$ give
\[
\norm{R_{W_1}-R_{W_2}}_{L^p}
\le2\norm{\Delta W}_{L^p(I^2)},
\qquad
\norm{(\Delta W)H_t^2}_{L^p}
\le T\norm{\Delta W}_{L^p(I^2)}.
\]
Consequently,
\begin{equation}
\label{eq:indicator-state-split-proof}
\sup_{t\le T}\norm{A_t}_{L^p}\le\delta_p,
\qquad
X_t^1-X_t^2=-W_1(H_t^1-H_t^2)+A_t.
\end{equation}
Set $d(t):=\norm{H_t^1-H_t^2}_{L^1}$. This function is absolutely continuous, and its derivative is bounded a.e. by the instantaneous indicator mismatch. Moreover,
\[
\abs{W_1(H_t^1-H_t^2)(u)}\le M_1d(t).
\]
For finite $p$, fix $\eta>0$. On $\{\abs{A_t}\le\eta\}$, an indicator mismatch implies $\abs{X_t^1}\le M_1d(t)+\eta$. The reference tube estimate and Markov's inequality therefore yield
\[
d'(t)
\le\omega(M_1d(t)+\eta)+\left(\frac{\delta_p}{\eta}\right)^p
\quad\text{a.e.}
\]
Taking the infimum over $\eta$ and using the trivial upper bound $1$ gives
\[
d'(t)\le\Omega_{p,\delta_p}(M_1d(t)).
\]
If $\delta_p=0$, \cref{prop:indicator-osgood} gives equality of the solutions; the generalized inverse in the statement is then zero. For $\delta_p>0$, the effective modulus is continuous, nondecreasing, and strictly positive, so its integral $\Phi_{p,\delta_p}$ is strictly increasing and unbounded. For $M_1>0$, scalar separation gives
\[
\int_0^{M_1d(t)}\frac{\dd s}{\Omega_{p,\delta_p}(s)}\le M_1t,
\]
which is \eqref{eq:indicator-osgood-H}. The state split gives \eqref{eq:indicator-osgood-X}, and the optimized mismatch estimate gives \eqref{eq:indicator-osgood-q}. For $p=\infty$, $\abs{A_t}\le\delta_\infty$ a.e., so $d'\le\omega(M_1d+\delta_\infty)$ and the same argument uses $\Phi_{\infty,\delta}$.

For every fixed $s>0$, $\Omega_{p,\delta}(s)\to\omega(s)$ as $\delta\downarrow0$. Hence Fatou's lemma and Osgood divergence give, for every $a>0$,
\[
\liminf_{\delta\downarrow0}
\int_0^a\frac{\dd s}{\Omega_{p,\delta}(s)}=\infty.
\]
It follows that $\Phi_{p,\delta}^{-1}(M_1T)\to0$. Choosing $\eta=\sqrt\delta$ in \eqref{eq:indicator-effective-modulus} also makes the indicator bound vanish. If $M_1=0$, the state split has no feedback term and directly gives the three elementary bounds stated in the theorem.
\end{proof}

\begin{proof}[Proof of \Cref{cor:indicator-explicit-log-osgood}]
\label{proof:cor:indicator-explicit-log-osgood}
Let $d(t)=\norm{H_t^1-H_t^2}_{L^1}$ and
$y(t)=M_1d(t)+q_{p,\delta}$.  For finite $p$, use the mismatch estimate
from the preceding proof with the time-dependent choice $\eta=y(t)$.
Because $M_1d(t)=y(t)-q_{p,\delta}\le y(t)$,
\[
M_1d(t)+\eta\le2y(t),
\qquad
\left(\frac{\delta}{y(t)}\right)^p
=\frac{q_{p,\delta}^{p+1}}{y(t)^p}\le y(t).
\]
As long as $2y(t)\le a_0$, \eqref{eq:indicator-log-osgood-upper} and
$\Lambda(y)\ge1$ imply
\[
d'(t)\le(2C_\omega+1)y(t)\Lambda(y(t))^\gamma,
\]
and therefore
\begin{equation}
\label{eq:indicator-log-y-proof}
y'(t)\le Ky(t)\Lambda(y(t))^\gamma,
\qquad y(0)=q_{p,\delta}.
\end{equation}
For $p=\infty$, the same estimate follows directly from
$d'(t)\le\omega(M_1d(t)+\delta)$.

For equality in \eqref{eq:indicator-log-y-proof}, put
$S(t)=\Lambda(y(t))$.  Then $S'=-KS^\gamma$, and solving this scalar
equation gives \eqref{eq:indicator-explicit-log-flow}.  The assumed bound
at $T$ closes the local bootstrap.  The state split yields
\eqref{eq:indicator-explicit-log-X}.  Substituting $\eta=y(t)$ in the
mismatch estimate and using that $a\mapsto a\Lambda(a)^\gamma$ is
increasing on the local interval gives
\eqref{eq:indicator-explicit-log-q}.  Finally, for $0\le\gamma<1$,
\[
\mathcal Q_{\gamma,K,T}(q)
=q\exp\!\left(KT[\log(A/q)]^\gamma(1+o(1))\right)
=q^{1-o(1)},
\]
while the $\gamma=1$ formula is exact.  These give the stated rates.
\end{proof}

\begin{proof}[Proof of \Cref{thm:indicator-one-sided-stability}]
\label{proof:thm:indicator-one-sided-stability}
Use \eqref{eq:indicator-state-split-proof} and put $d_t:=\norm{H_t^1-H_t^2}_{L^1}$. Then
\[
\sup_t\norm{A_t}_{L^p}\le\delta_p,
\qquad |W_1(H_t^1-H_t^2)|\le M_1d_t.
\]

Fix $\eta>0$. On $\{|A_t|\le\eta\}$, an indicator mismatch implies
\[
\abs{X_t^1}\le M_1d_t+\eta.
\]
The reference tube estimate and Markov's inequality therefore yield
\begin{equation}
\label{eq:indicator-mismatch-proof}
\norm{\1_{\{X_t^1\le0\}}-\1_{\{X_t^2\le0\}}}_{L^1}
\le L(M_1d_t+\eta)+\frac{\delta_p^p}{\eta^p}.
\end{equation}
For $\delta_p>0$, choose $\eta=\delta_p^{p/(p+1)}$. Since
\[
d_t\le\int_0^t
\norm{\1_{\{X_s^1\le0\}}-\1_{\{X_s^2\le0\}}}_{L^1}\,\dd s,
\]
Gronwall gives
\[
\sup_{t\le T}d_t
\le(L+1)T e^{LM_1T}\delta_p^{p/(p+1)}.
\]
Substitution into \eqref{eq:indicator-mismatch-proof} gives the indicator bound. The state split gives
\[
\sup_{t\le T}\norm{X_t^1-X_t^2}_{L^p}
\le\delta_p+M_1\sup_{t\le T}d_t,
\]
which proves the state bound in $L^p$. If $\delta_p=0$, let $\eta\downarrow0$. For $p=\infty$, $\abs{A_t}\le\delta_\infty$ a.e., so the same argument gives linear bounds.

To prove optimality for all three estimates, take $W_1=W_2\equiv c>0$, $\mu=r=0$, $x_0^1(u)=u$, and fix $T<1/c$. Then $H^1=0$, $X^1(u)=u$, and the tube estimate holds with $L=1$. For $0<b<1$ set
\[
x_0^{2,b}(u)=
\begin{cases}
-u,&0<u<b,\\
u,&b\le u\le1
\end{cases}.
\]
Then $H_t^{2,b}=t\1_{(0,b)}$ on $[0,T]$: the perturbed state is negative on $(0,b)$, while for $u\ge b$ it is at least $b(1-cT)>0$. Consequently,
\[
\norm{H_t^{2,b}-H_t^1}_{L^1}=bt,
\qquad
\norm{x_0^{2,b}-x_0^1}_{L^p}
=\frac{2}{(p+1)^{1/p}}b^{1+1/p}.
\]
The indicator mismatch is $b$, and
\[
\norm{X_t^{2,b}-X_t^1}_{L^p}
\ge\norm{X_t^{2,b}-X_t^1}_{L^1}\ge cbt(1-b).
\]
The preceding upper bounds show that, at any fixed $t>0$, all three errors are of order $b\asymp\delta_p^{p/(p+1)}$.
\end{proof}

\begin{proof}[Proof of \Cref{thm:indicator-sampled-label}]
\label{proof:thm:indicator-sampled-label}
Choose a common full-measure set of labels on which the hard-solution identities hold for every $t\in[0,T]$. Such a set exists by first taking rational times and then using time continuity. On this set, $H$ is $1$-Lipschitz and $X$ is $L_T$-Lipschitz, where
\[
L_T:=|\mu|+2|r|M_W+M_W.
\]
Almost surely, every sampled label lies in this set. The upper finite iterates are Borel functions of the sample, and their monotone convergence is uniform in time. Thus the canonical path is a Borel map into $C([0,T];\mathbb R^N)$.

Write $R_i^N:=N^{-1}\sum_j(e_{ij}^N-e_{ji}^N)$ and define
\[
\rho_i^N(t):=rt\bigl(R_i^N-R_W(U_i)\bigr)
-\left[\frac1N\sum_je_{ij}^NH_t(U_j)
-\int_IW(U_i,v)H_t(v)\,\dd v\right].
\]
For $j\ne i$, put
\[
\zeta_{ij}(t):=rt\bigl(W(U_i,U_j)-W(U_j,U_i)\bigr)
-W(U_i,U_j)H_t(U_j).
\]
Conditional on $U_i$, these variables are independent, uniformly bounded, and uniformly Lipschitz in time. Their common conditional mean is $m_i(t):=\E[\zeta_{ij}(t)\mid U_i]$, and
\[
\rho_i^N(t)=\frac1N\sum_{j\ne i}
\left(\zeta_{ij}(t)-m_i(t)\right)-\frac{m_i(t)}N.
\]
The last term is the $O(N^{-1})$ bias from omitting the diagonal. Hoeffding's inequality \cite{hoeffding1963} on a time grid of mesh $\Delta_N\le N^{-1}$, a union bound over the rows and grid points, and time interpolation give
\begin{equation}
\label{eq:indicator-row-residual-proof}
\rho_N:=\max_i\sup_{t\le T}|\rho_i^N(t)|
\le C\left(\sqrt{\frac{\log N+z}{N}}+\frac1N\right)
\end{equation}
outside an event of probability at most $Ce^{-cz}$.

Let $F_{N,s}$ and $F_s$ be the empirical and population distribution functions of $X_s(U)$ at a grid time $s$. The Dvoretzky--Kiefer--Wolfowitz inequality \cite{massart1990} and a union bound give, on a common event of the same probability,
\[
\max_{s\text{ on the grid}}\sup_x|F_{N,s}(x)-F_s(x)|
\le\epsilon_N,
\qquad
\epsilon_N\le C\sqrt{\frac{\log N+z}{N}}.
\]
Set $h_N=L_T\Delta_N$. For any $t$, choose a grid point $s$ with $|t-s|\le\Delta_N$. Then, for every $a\ge0$,
\begin{equation}
\label{eq:indicator-empirical-tube-alltime}
\frac1N\sum_i\1_{\{|X_t(U_i)|\le a\}}
\le\frac1N\sum_i\1_{\{|X_s(U_i)|\le a+h_N\}}
\le\omega(a+h_N)+2\epsilon_N.
\end{equation}
The factor $2$ accounts for the two endpoints of the tube. Increase $C_0$ in \eqref{eq:indicator-sampled-kappa} so that
\[
\rho_N+h_N\le\kappa_N(z),\qquad 2\epsilon_N\le\kappa_N(z).
\]
Subtracting the state equations yields
\[
|X_i^{N,+}(t)-X_t(U_i)|\le M_WD_N(t)+\rho_N.
\]
An indicator mismatch lies in the target tube of this radius. Applying \eqref{eq:indicator-empirical-tube-alltime} and integrating gives
\[
D_N(t)\le\int_0^t
\left[\omega\!\left(M_WD_N(s)+\kappa_N(z)\right)+\kappa_N(z)\right]\,\dd s.
\]
Scalar comparison proves \eqref{eq:indicator-sampled-osgood-H}--\eqref{eq:indicator-sampled-osgood-X}. The event is defined entirely by the sampled labels and target paths, independently of the chosen finite solution, so the same bounds hold for every finite hard solution on that event.

For the state laws, couple the finite and sampled-target states through their labels. Their Wasserstein distance is bounded by \eqref{eq:indicator-sampled-osgood-X}. The target states lie in $[-B_T,B_T]$, with $B_T:=\norm{x_0}_{L^\infty}+TL_T$. At grid times,
\[
\Wone\left(\frac1N\sum_i\delta_{X_s(U_i)},\Law(X_s(U))\right)
\le2B_T\epsilon_N.
\]
Time interpolation adds at most $2h_N$, proving \eqref{eq:indicator-sampled-osgood-W1}.

For the threshold fraction, write $q(t):=\lambda\{X_t\le0\}$. Bracketing $\{X_t\le0\}$ by $\{X_s\le-h_N\}$ and $\{X_s\le h_N\}$ gives
\[
\sup_{t\le T}\left|\frac1N\sum_i\1_{\{X_t(U_i)\le0\}}-q(t)\right|
\le\omega(h_N)+\epsilon_N.
\]
Adding the finite-versus-target mismatch from \eqref{eq:indicator-empirical-tube-alltime}, and using $\rho_N+h_N\le\kappa_N(z)$, gives \eqref{eq:indicator-sampled-osgood-fraction}. These supremum errors are measurable. Indeed, the tube bound and time Lipschitz continuity give $|q(t)-q(s)|\le\omega(L_T|t-s|)$, so $q$ is continuous. The finite closed-threshold count is jointly upper semicontinuous in time and in the continuous state path. Its maximum excess over $q$ on the compact time interval is therefore upper semicontinuous as a path functional. The opposite signed error is lower semicontinuous in time, so its supremum equals the supremum over a countable dense set of times and is Borel. Their maximum is the required absolute-error functional.

Finally, for $M_W>0$, separation of variables gives
\[
\int_\kappa^{\Gamma_\kappa(T)}\frac{\dd a}{\omega(a)+\kappa}=M_WT.
\]
The Osgood integral forces $\Gamma_\kappa(T)\to0$. Conversely, if that integral is finite, choose $b>0$ with $\int_0^b\dd a/\omega(a)<M_WT$; then $\Gamma_\kappa(T)>b$ for every sufficiently small $\kappa$. When $M_W=0$, $\mathcal R_\kappa(T)=T[\omega(\kappa)+\kappa]\to0$. For $\omega(a)=1\wedge La$, linear comparison gives $\mathcal R_\kappa(T)\le C_T\kappa$.
\end{proof}

\begin{proof}[Proof of \Cref{prop:indicator-obstructions}]
\label{proof:prop:indicator-obstructions}
For (i), let $W\equiv w>0$, $x_0=0$, $r=0$, and $0<\mu<w$. Both
\[
H_t\equiv0,\quad X_t=\mu t,
\qquad\text{and}\qquad
H_t\equiv t,\quad X_t=(\mu-w)t
\]
are hard solutions.

For (ii), take $W\equiv w>0$, $\mu=r=0$, and $x_0=0$. Since $H\ge0$, the equation forces $H_t=t$ and $X_t=-wt$, so the hard solution is unique. If $x_0^n\equiv1/n$, activation would require $H$ to reach $1/(nw)$ while its derivative is still zero; hence the unique solution is $H_t^n=0$ and $X_t^n=1/n$. Thus both cumulative distress and state have an order-one jump despite uniform convergence of the initial profiles.

For (iii), take $W\equiv1$, $\mu=r=0$, and $x_0(u)=u^2$. With $\beta(t)=\int_IH_t(u)\,\dd u$, the equation becomes
\[
\dot\beta(t)=\lambda\{u:u^2\le\beta(t)\}=\sqrt{\beta(t)},
\qquad \beta(0)=0,
\]
as long as $\beta\le1$. For every delay $\tau\in[0,T)$, a solution on $[0,T]$ is
\[
\beta_\tau(t)=
\begin{cases}
0,&t\le\tau,\\
(t-\tau)^2/4,&\tau<t\le\tau+2,\\
t-\tau-1,&t>\tau+2.
\end{cases}
\]
The last branch uses $\dot\beta=1$ after all institutions have crossed the threshold. The initial law is atomless, but its square-root tube modulus fails the Osgood condition.

For (iv), take $W\equiv-1$, $x_0=0$, $r=0$, and $\mu=-1/2$. The state is independent of $u$, so every solution would satisfy
\[
H'(t)=\1_{\{H(t)-t/2\le0\}}.
\]
For $y(t)=H(t)-t/2$ this reads
\[
y'=\frac12\ \text{on }\{y\le0\},
\qquad
y'=-\frac12\ \text{on }\{y>0\}.
\]
An absolutely continuous function has derivative zero a.e. on the level set $\{y=0\}$, so that set must be null. The chain rule then gives $(\abs y)'=-1/2$ a.e. and hence $\abs{y(t)}=-t/2$, a contradiction. Thus no hard solution exists.
\end{proof}

\begin{remark}[Relaxed signed-kernel closure]
For a bounded signed kernel, Lipschitz ramps admit a natural compact closure in the relaxed occupation-rate formulation
\[
q_t(u)\in
\begin{cases}
\{1\},&X_t(u)<0,\\
[0,1],&X_t(u)=0,\\
\{0\},&X_t(u)>0,
\end{cases}
\qquad
H_t=\int_0^tq_s\,\dd s,
\qquad X=x_0+\mu t+rtR_W-WH.
\]
The rates $q^\varepsilon$ have a weak-star convergent subsequence in $L^\infty$ with limit $q$. Set $H_t=\int_0^tq_s\,\dd s$. The paths $H^\varepsilon$ are uniformly bounded and Lipschitz in time in $L^2(I)$, and converge weakly to $H_t$ at each time. The Hilbert--Schmidt operator $W$ maps each bounded set into a relatively compact set, so Arzel\`a--Ascoli gives a further subsequence with $WH^\varepsilon\to WH$ in $C([0,T];L^2(I))$. Hence $X^\varepsilon\to X$ strongly, and almost everywhere after another subsequence. Away from $X=0$, the ramp rates converge to the hard indicator; on $X=0$, their weak-star limit remains in $[0,1]$. Hence $(q,X)$ satisfies the displayed inclusion. If the space--time contact set has measure zero, the limit is a hard solution. In the signed counterexample above, the positive-side ramp gives
\[
X_t^\varepsilon=\frac{\varepsilon}{2}(1-e^{-t/\varepsilon}),
\qquad H_t^\varepsilon=\frac t2+X_t^\varepsilon.
\]
Its limit is the relaxed path $H_t=t/2$, $X_t=0$.
\end{remark}

\section{Additional representations and approximation results}
\label{app:theory-supplement}

This appendix gives the density representation of the factor dynamics, the empirical estimates used for sampling, and constructions that satisfy the assumptions of the approximation theorems.

\subsection{Transport PDE for the limiting density}

For fixed loadings, the state map translates each initial buffer by the same amount. We use this translation to write the joint state--factor density explicitly and identify its transport equation. The representation also applies to a discontinuous loss function.

\begin{theorem}
\label{thm:transport-pde}
Assume \cref{ass:bounded-types}, let $\ell:\R\to[0,\ell_\ast]$ be bounded and measurable, and let $\beta$ solve \eqref{eq:Xlimit-map}--\eqref{eq:beta-limit-correct} on $[0,T]$. Suppose that $\mu_0$ has a compactly supported $C^1$ density $\Phi$ on $\mathcal Z$.

Then \(\beta\) is absolutely continuous, with \(\dot\beta_k(t)=m_k(t)\) for a.e. \(t\), and for every \(t\in[0,T]\) the measure \(\widehat\nu_t=(\Xi_t)_\#\mu_0\) admits a density \(f(t,\cdot)\) on \(\mathcal Z\), given by
\begin{equation}
\label{eq:density-explicit}
f(t,x,a,b)=\Phi\!\left(x-\Theta_t(a,b),a,b\right),
\end{equation}
where
\[
\Theta_t(a,b):=\mu t+\frac{rt}{K}\sum_{k=1}^K(a_k\bar b_k-b_k\bar a_k)-\frac{1}{K}\sum_{k=1}^K a_k\beta_k(t).
\]
Moreover, \(f\) solves the transport equation
\begin{equation}
\label{eq:transport-pde}
\partial_t f+\partial_x(v_f f)=0,\qquad f(0,x,a,b)=\Phi(x,a,b),
\end{equation}
in the distributional sense on \((0,T)\times\mathcal Z\), with velocity field
\begin{equation}
\label{eq:velocity-field}
v_f(t,a,b)=\mu+\frac{r}{K}\sum_{k=1}^K(a_k\bar b_k-b_k\bar a_k)-\frac{1}{K}\sum_{k=1}^K a_km_k(t),
\end{equation}
where
\begin{equation}
\label{eq:mk-density}
m_k(t)=\int_{\mathcal Z}\tilde b_k\,\ell(y)\,f(t,y,\tilde a,\tilde b)\,\dd y\,\dd\tilde a\,\dd\tilde b.
\end{equation}
Equivalently, for every \(\varphi\in C_c^1([0,T)\times\mathcal Z)\),
\begin{multline}
\label{eq:transport-weak-form}
\int_0^T\!\!\int_{\mathcal Z}
f(t,x,a,b)\Bigl(\partial_t\varphi(t,x,a,b)
+v_f(t,a,b)\partial_x\varphi(t,x,a,b)\Bigr)
\,\dd x\,\dd a\,\dd b\,\dd t\\
+\int_{\mathcal Z}\Phi(x,a,b)\varphi(0,x,a,b)\,\dd x\,\dd a\,\dd b=0.
\end{multline}
For every $k$, the feedback $m_k$ is continuous in time. Thus $f$ also solves \eqref{eq:transport-pde} classically.
\end{theorem}

\begin{proof}[Proof sketch]
Translation in the initial-state coordinate gives \eqref{eq:density-explicit}. Continuity of these translated densities in $L^1$ makes $m_k$ continuous for bounded measurable $\ell$. Differentiating the formula gives the transport equation; see \cref{app:proofs}.
\end{proof}
\begin{corollary}[Indicator-loss transport equation]
\label{cor:indicator-transport}
Under the assumptions of \cref{thm:indicator}, if $\mu_0$ has a compactly supported $C^1$ density $\Phi$, the unique indicator solution has density \eqref{eq:density-explicit} and satisfies \eqref{eq:transport-pde} classically, with $\ell(y)=\mathbf1_{\{y\le0\}}$ in \eqref{eq:mk-density}.
\end{corollary}

\begin{proof}
Apply \cref{thm:transport-pde} to the solution from \cref{thm:indicator}.
\end{proof}

\subsection{Empirical threshold estimates}

The proof of \cref{thm:indicator-finiteN} compares sampled and population threshold averages over the whole feedback cube. We first consider a fixed projection, then give the uniform estimate required for that comparison.

\begin{remark}
\label{rem:indicator-empirical}
We estimate the sampled mass near zero by comparing the empirical distribution function with that of $\mu_0$ after projection. Fix $T,C>0$, $t\in[0,T]$, and $\beta\in[-C,C]^K$. For i.i.d.\ samples $Z_1,\dots,Z_N$ from $\mu_0$, set $Y_i:=\Psi_t(Z_i,\beta)$. Under \cref{ass:density}, the Dvoretzky--Kiefer--Wolfowitz inequality gives, for $\varepsilon,\eta>0$,
\[
\Prob\!\left( \sup_{x\in\R} \frac1N\sum_{i=1}^N \1_{\{\abs{Y_i-x}\le \varepsilon\}} > 2M_\rho(T,C)\varepsilon + 2\eta \right)
\le 2e^{-2N\eta^2}.
\]
Thus the projected empirical mass in an $\varepsilon$-neighborhood of the threshold is typically of order $\varepsilon+N^{-1/2}$. The quantile-matched constructions in \cref{sec:numerics} replace random sampling by deterministic quantiles. For the sampled feedback equation, the analysis requires a uniform bound over the entire feedback cube, which is provided by the next lemma.
\end{remark}

\begin{lemma}[Uniform empirical process bound for the fixed-rank indicator classes]
\label{lem:indicator-vc}
Assume \cref{ass:bounded-types}. Let $Z_1,\dots,Z_N$ be i.i.d.\ with law $\mu_0$, write $\mathbb P_N:=N^{-1}\sum_{i=1}^N \delta_{Z_i}$ and $\mathbb P f:=\int f\,\dd\mu_0$, and fix $T>0$. Set $\mathcal C_T:=[-MT,MT]^K$. For $\alpha,\gamma\in[-M,M]^K$, $t\in[0,T]$, $\beta\in\mathcal C_T$, and $1\le k\le K$, define
\[
g_{t,\beta,\alpha,\gamma,k}(z):= b_k\,\1\!\left\{x+\mu t+\frac{rt}{K}\sum_{j=1}^K (a_j\gamma_j-b_j\alpha_j)-\frac{1}{K}\sum_{j=1}^K a_j\beta_j\le 0\right\},
\qquad z=(x,a,b)\in\mathcal Z.
\]
Let $\mathcal G_{T,K}$ be the class of these functions restricted to
\[
\mathcal Z_M:=\{(x,a,b):\max_{1\le j\le K}(|a_j|\vee|b_j|)\le M\},
\]
which contains the support of $\mu_0$. Then $\mathcal G_{T,K}$ is a bounded VC-subgraph class with envelope $M$. For every $N\ge2$, there is a measurable random variable $Z_{N,K}$ such that
\[
\sup_{g\in\mathcal G_{T,K}} |(\mathbb P_N-\mathbb P)g|\le Z_{N,K},
\qquad
\E Z_{N,K}\le C(M)\sqrt{\frac{K\log N}{N}}.
\]
\end{lemma}

\begin{proof}[Proof sketch]
The thresholds form a class of affine halfspaces in $\R^{2K+1}$. A countable halfspace class dominates the empirical supremum, and the halfspace trace bound controls the coordinate-weighted functions. Symmetrization then gives the estimate; see \cref{app:proofs}.
\end{proof}

\subsection{Fixed-width regularization and smoothing}

We distinguish sampling a regularized system from replacing the hard loss by a ramp. At a fixed width, \cref{thm:finite-rank-w1} bounds the sampling error. A uniform density bound then lets us compare the regularized and hard paths as the width tends to zero.

\begin{corollary}[Sampled stability for a fixed regularization]
\label{cor:regularized-sampled}
Fix $\varepsilon>0$ and let $\ell_\varepsilon$ be the Lipschitz regularization in \eqref{eq:elleps}. Assume \cref{ass:bounded-types} and $\E|X_0|<\infty$. The deterministic and sampled rank-$K$ feedback equations with loss $\ell_\varepsilon$ have unique global solutions, denoted by $\beta^\varepsilon$ and $\beta^{N,\varepsilon}$. The map
\[
(Z_1,\ldots,Z_N)\longmapsto \beta^{N,\varepsilon}\in C([0,T];\R^K)
\]
is Borel measurable. If $\mu_0^N$ is the empirical law of $N$ i.i.d.\ samples from $\mu_0$, then, for every $T>0$,
\[
\sup_{0\le t\le T}
\Wone\!\left(\widehat\nu_t^{N,\varepsilon},\widehat\nu_t^{\varepsilon}\right)
\le C_{T,\varepsilon}\,\Wone(\mu_0^N,\mu_0)
\qquad\text{a.s.},
\]
where
\[
\widehat\nu_t^{N,\varepsilon}:=(\Xi_t^{N,\varepsilon})_\#\mu_0^N,
\qquad
\widehat\nu_t^{\varepsilon}:=(\Xi_t^{\varepsilon})_\#\mu_0
\]
are the joint state--factor laws defined as in \cref{thm:finite-rank-w1}. The constant depends only on $T$, the factor bound, $|r|$, and the Lipschitz constant $1/\varepsilon$. In particular, the same estimate holds after taking expectations whenever $\E\Wone(\mu_0^N,\mu_0)<\infty$.
\end{corollary}

\begin{proof}
For fixed $\varepsilon>0$, the sampled and deterministic vector fields are globally Lipschitz. Picard iteration gives unique solutions and, since every Picard iterate is a Borel function of the sampled types, also gives Borel measurability of the sampled solution map. The stability estimate is \cref{thm:finite-rank-w1} applied to $\ell_\varepsilon$.
\end{proof}

\begin{remark}[Indicator loss along graphon truncations]
\label{rem:graphon-indicator}
For $\ell(x)=\1_{\{x\le0\}}$, \cref{thm:indicator-one-sided-stability,cor:indicator-strong-bridge} replace Lipschitz continuity of the loss by a bound on the mass near zero along the target path. They give the exponent $p/(p+1)$ for aligned $L^p$ data errors, $p<\infty$, and a linear bound at $p=\infty$. The next proposition concerns a separate error: replacing the indicator by a ramp of width $\varepsilon$. A density bound uniform over the rank-$K$ family makes this error $O(\varepsilon)$.
\end{remark}

\begin{proposition}[Conditional uniform smoothing error along rank-$K$ truncations]
\label{prop:graphon-indicator-smoothing}
For $\varepsilon>0$ let $\ell_\varepsilon$ be the piecewise-linear regularization of the threshold indicator defined in \eqref{eq:elleps}. In the setting of \cref{thm:bridge}, assume the rank-$K$ type laws $\mu_0^{(K)}$ satisfy \cref{ass:density} on $[0,T]$ with a constant $M_\rho(T,MT)$ independent of $K$. Let $X^{K,\varepsilon}$ and $X^{K,\mathrm{ind}}$ denote the rank-$K$ graphon solutions associated with $(W^{(K)},x_0^{(K)})$ and losses $\ell_\varepsilon$ and $\1_{\{x\le 0\}}$, respectively. Then there exists $C_T<\infty$, independent of $K$ and $\varepsilon$, such that
\[
\sup_{0\le t\le T}\norm{X_t^{K,\varepsilon}-X_t^{K,\mathrm{ind}}}_{L^1(I)}\le C_T\varepsilon,
\qquad
\sup_{0\le t\le T}\Wone\bigl(\Law(X_t^{K,\varepsilon}(U)),\Law(X_t^{K,\mathrm{ind}}(U))\bigr)\le C_T\varepsilon.
\]
\end{proposition}

\begin{proof}[Proof sketch]
Since $|\ell_\varepsilon(y)-\1_{\{y\le0\}}|\le\1_{\{0<y<\varepsilon\}}$, the uniform density hypothesis makes the smoothed and indicator feedback fields differ by at most $MM_\rho(T,MT)\,\varepsilon$ on the feedback cube, uniformly in $K$; the indicator field is Lipschitz there with a $K$-uniform constant, so Gronwall's inequality yields an $O(\varepsilon)$ bound on the feedback paths, which the state map transfers to $L^1$ and to $\Wone$.
\end{proof}

The proof applies verbatim to the negative-side ramp used in \cref{sec:real-data}, with the tube $\{-\varepsilon<y\le0\}$ replacing $\{0<y<\varepsilon\}$ in the bound above.

\subsection{Constructing finite-rank approximations}

We first state the path convergence implied by the kernel stability estimate, then give spectral and block constructions of the approximating kernels. The distinction between bounded kernels and uniformly bounded factors matters when a continuum approximation is also sampled.

\begin{corollary}
\label{cor:graphon-approx}
Let $\ell:\R\to[0,\ell_\ast]$ be bounded and Lipschitz, and let $x_0\in L^\infty(I)$. Suppose $W_K\in L^\infty(I^2)$ and $x_0^K\in L^\infty(I)$ satisfy
\[
\sup_{K\ge 1}\norm{W_K}_{L^\infty}<\infty,
\qquad
W_K\to W\text{ in }L^1(I^2),
\qquad
x_0^K\to x_0\text{ in }L^1(I).
\]
If $X^K$ and $X$ solve \eqref{eq:graphon-equation} with data $(W_K,x_0^K)$ and $(W,x_0)$, respectively, then
\[
\sup_{0\le t\le T}\norm{X_t^K-X_t}_{L^1(I)}\to 0
\qquad \text{for every }T>0.
\]
\end{corollary}

\begin{proof}
Immediate from \cref{thm:graphon-wellposed}.
\end{proof}

\begin{remark}[SVD truncations and rank selection]
\label{rem:svd}
If $W$ is Hilbert--Schmidt, a singular-value expansion
\[
W(u,v)=\sum_{m=1}^\infty \sigma_m\phi_m(u)\psi_m(v)
\]
gives the finite-rank partial sums
\[
W^{(K)}(u,v)=\sum_{m=1}^K \sigma_m\phi_m(u)\psi_m(v).
\]
The partial sums converge in $L^2$, hence in $L^1(I^2)$,
\[
\norm{W-W^{(K)}}_{L^1(I^2)}
\le \norm{W-W^{(K)}}_{L^2(I^2)}
=\Bigl(\sum_{m>K}\sigma_m^2\Bigr)^{1/2}.
\]
For the bounded kernels considered here, each finite partial sum is also bounded: the singular-function identities and Cauchy--Schwarz give
\[
\norm{\phi_m}_{L^\infty},\ \norm{\psi_m}_{L^\infty}
\le \frac{\norm{W}_{L^\infty}}{\sigma_m},\qquad \sigma_m>0.
\]
The sampled bridge uses truncations with the uniform factor bound stated in \cref{rem:bridge-factor-gap}. Signed partial sums enter the regularized indicator bridge of \cref{cor:indicator-strong-bridge}; nonnegative approximants also admit greatest hard solutions.

We can apply \cref{thm:bridge} at each chosen truncation level. Choosing that level requires a criterion: one may control the discarded singular values, compare trajectories, or retain channels specified by a financial scenario. In \cref{sec:full-sample}, we compare spectral truncation with channels chosen from a prescribed shock, measuring both by the same trajectory error.
\end{remark}

\begin{proposition}[Constructive block-constant $L^1$ approximation]
\label{prop:block-L1-approx}
Let $W\in L^\infty(I^2)$. For each $K\in\mathbb N$, partition $I$ into the equal subintervals
\[
I_j^{(K)}:=\Bigl[\frac{j-1}{K},\frac{j}{K}\Bigr),\qquad 1\le j\le K,
\]
and define the block-constant conditional expectation
\[
\Pi_K W(u,v)
:= \sum_{i,j=1}^K \1_{I_i^{(K)}}(u)\1_{I_j^{(K)}}(v)
\frac{1}{|I_i^{(K)}||I_j^{(K)}|}
\int_{I_i^{(K)}\times I_j^{(K)}} W(s,t)\,\dd s\,\dd t.
\]
Then $\norm{\Pi_KW}_{L^\infty}\le \norm{W}_{L^\infty}$, the integral operator associated with $\Pi_KW$ has rank at most $K$, and
\[
\norm{\Pi_KW-W}_{L^1(I^2)}\to 0
\qquad\text{as }K\to\infty.
\]
If in addition $W$ is piecewise $C^1$ on finitely many rectangles and the partial derivatives are bounded on each smooth piece, then
\[
\norm{\Pi_KW-W}_{L^1(I^2)} = O(K^{-1}).
\]
\end{proposition}

\begin{proof}[Proof sketch]
The $L^\infty$ bound follows from Jensen's inequality, and the range of the associated operator is contained in the $K$-dimensional space of functions constant on the $u$-partition. For $L^1$ convergence, approximate $W$ by a continuous function, use that $\Pi_K$ is an $L^1$ contraction, and apply uniform continuity to the continuous approximation. In the piecewise-$C^1$ case, cells interior to a smooth rectangle contribute $O(K^{-1})$, while cells meeting the finitely many rectangle boundaries have total area $O(K^{-1})$.
\end{proof}

\begin{remark}
\label{rem:block-L1-approx}
The block averages in \cref{prop:block-L1-approx} approximate every bounded kernel and retain the uniform kernel bound required by \cref{thm:graphon-wellposed,cor:graphon-approx}. If $W\ge0$, its block averages remain nonnegative, so \cref{cor:indicator-strong-bridge} applies to their greatest hard solutions whenever the target path satisfies the tube bound. To use these approximations in the sampled bridge, we also need the factor bounds discussed in \cref{rem:bridge-factor-gap}.
\end{remark}

\begin{remark}[Uniform factor bounds for sampled approximations]
\label{rem:bridge-factor-gap}
The finite-$N$ bridge theorem, \cref{thm:bridge}, uses a representation
\[
W^{(K)}(u,v)=\frac1K\sum_{k=1}^K a_k^{(K)}(u)b_k^{(K)}(v)
\]
whose factor maps are uniformly bounded in $K$. These bounds keep the constants in the population comparison uniform as the number of factors grows. One way to satisfy this requirement is to start from a fixed finite representation
\[
W(u,v)=\frac1R\sum_{r=1}^R \alpha_r(u)\gamma_r(v),
\]
with fixed $R$ and uniformly bounded factors. For $K$ a multiple of $R$, duplicate each pair $K/R$ times. This gives the same kernel with $K$ uniformly bounded factors, as in the exact low-rank benchmarks of \cref{sec:numerics-exact-low-rank}.
\end{remark}

\subsection{Factorized kernels and the infinite decomposition model}

A continuous collection of channels leads to the same feedback construction, with integration over the channel index replacing the finite sum. The following proposition writes this model as a factorized graphon.

\begin{proposition}
\label{prop:factorized-kernel}
Let $(\Theta,\nu)$ be a finite measure space and assume that
\[
W(u,v)=\int_\Theta a(u,\theta)b(v,\theta)\,\nu(\dd\theta)
\]
with bounded measurable coefficients $a$ and $b$. Then the graphon equation \eqref{eq:graphon-equation} is equivalent to
\begin{equation}
\label{eq:factorized-graphon}
X_t(u)=x_0(u)+\mu t+r tR_W(u)-\int_\Theta a(u,\theta)C_t(\theta)\,\nu(\dd\theta),
\end{equation}
where
\begin{equation}
\label{eq:Ctheta}
C_t(\theta)=\int_0^t\int_I b(v,\theta)\ell(X_s(v))\,\dd v\,\dd s.
\end{equation}
\end{proposition}

\begin{proof}
Insert the factorization of $W$ into \eqref{eq:graphon-equation} and apply Fubini's theorem:
\begin{align*}
\int_0^t\int_I W(u,v)\ell(X_s(v))\,\dd v\,\dd s
&= \int_0^t\int_I\int_\Theta a(u,\theta)b(v,\theta)\ell(X_s(v))\,\nu(\dd\theta)\,\dd v\,\dd s \\
&= \int_\Theta a(u,\theta)\left(\int_0^t\int_I b(v,\theta)\ell(X_s(v))\,\dd v\,\dd s\right)\nu(\dd\theta).
\end{align*}
Substitution yields \eqref{eq:factorized-graphon}--\eqref{eq:Ctheta}.
\end{proof}

\begin{remark}
\label{rem:factorized-sigma-finite}
The algebraic reformulation in \cref{prop:factorized-kernel} extends to $\sigma$-finite measure spaces $(\Theta,\nu)$ whenever Fubini's theorem is justified, for example under
\[
\int_\Theta \norm{a(\cdot,\theta)}_{L^\infty(I)}\norm{b(\cdot,\theta)}_{L^\infty(I)}\,\nu(\dd\theta)<\infty.
\]
The indicator theorem below uses a finite measure $\nu$ and bounded factors.
\end{remark}

\begin{assumption}
\label{ass:factorized-density}
Let $T>0$ and $C>0$. There exists $M_{\rho}^{\Theta}(T,C)<\infty$ such that for every $t\in[0,T]$ and every measurable $c:\Theta\to\R$ with $\norm{c}_{L^\infty(\Theta,\nu)}\le C$, the scalar random variable
\begin{equation}
\label{eq:factorized-density-proj}
\Psi_t(U,c):=x_0(U)+\mu t+r tR_W(U)-\int_\Theta a(U,\theta)c(\theta)\,\nu(\dd\theta),
\qquad U\sim\lambda,
\end{equation}
admits a density bounded by $M_{\rho}^{\Theta}(T,C)$.
\end{assumption}

\begin{proposition}[Uniform transversality implies \cref{ass:factorized-density}]
\label{prop:factorized-density-transverse}
Assume that for every $t\in[0,T]$ and every measurable $c:\Theta\to\R$ with $\norm{c}_{L^\infty(\Theta,\nu)}\le C$, the profile $u\mapsto \Psi_t(u,c)$ is piecewise $C^1$, admits at most $J$ monotone branches, and satisfies $\abs{\partial_u \Psi_t(u,c)}\ge m>0$ on each branch. Then \cref{ass:factorized-density} holds with $M_{\rho}^{\Theta}(T,C)=J/m$.
\end{proposition}

\begin{proof}
Fix $t\in[0,T]$ and $c$ with $\norm{c}_{L^\infty(\Theta,\nu)}\le C$. Let $I_1,\dots,I_n$ be a partition of $I$ into monotone $C^1$ branches of $u\mapsto \Psi_t(u,c)$, with $n\le J$. On each branch the restriction $g_j:=\Psi_t(\cdot,c)|_{I_j}$ is strictly monotone, so by the one-dimensional change-of-variables formula its pushforward of Lebesgue measure has density
\[
\rho_j(y)=\1_{g_j(I_j)}(y)\,\frac{1}{\abs{g_j'(g_j^{-1}(y))}}\le \frac{1}{m}.
\]
The law of $\Psi_t(U,c)$ for $U\sim\lambda$ is the sum of these branchwise pushforwards, hence it has density
\[
\rho(y)=\sum_{j=1}^n \rho_j(y)\le \frac{n}{m}\le \frac{J}{m}.
\]
Therefore \cref{ass:factorized-density} holds with $M_{\rho}^{\Theta}(T,C)=J/m$. Piecewise-smooth block, core--periphery, and trigonometric specifications satisfying the stated branchwise transversality conditions fall into this regime.
\end{proof}

\begin{theorem}
\label{thm:factorized-indicator}
Assume the factorized representation of \cref{prop:factorized-kernel}, let $\ell(x)=\1_{\{x\le 0\}}$, and set
\[
B_\ast:=\norm{b}_{L^\infty(I\times\Theta)}.
\]
Fix $T>0$. If $B_\ast>0$, suppose \cref{ass:factorized-density} holds with radius parameter $TB_\ast$; if $B_\ast=0$, no density assumption is needed. Then there exists a unique
\[
C\in C\bigl([0,T];L^1(\Theta,\nu)\bigr)\cap L^\infty([0,T]\times\Theta)
\]
solving the integral feedback equation
\begin{equation}
\label{eq:factorized-feedback-indicator}
C_t(\theta)=\int_0^t F_s(C_s)(\theta)\,\dd s,
\qquad
F_t(c)(\theta):=\int_I b(v,\theta)\1_{\{\Psi_t(v,c)\le 0\}}\,\dd v,
\end{equation}
for every $t\in[0,T]$ and a.e. $\theta$. This solution satisfies $\abs{C_t(\theta)}\le tB_\ast$. Consequently the profile \(X\) defined by \eqref{eq:factorized-graphon} belongs to \(C([0,T];L^\infty(I))\) and is the unique indicator solution of \eqref{eq:graphon-equation}.
\end{theorem}

\begin{proof}[Proof sketch]
The density bound makes $F_t(c)$ jointly continuous in $(t,c)$ and Lipschitz in the $L^1(\Theta,\nu)$ distance on the order interval $|c|\le TB_\ast$. Picard iteration preserves $|C_t|\le tB_\ast$ and gives the unique feedback path; \eqref{eq:factorized-graphon} then recovers $X$. See \cref{app:proofs}.
\end{proof}

\subsection{Explicit threshold moduli and rates}

We specialize the criteria of \cref{sec:threshold-criteria} to logarithmic Osgood moduli. This allows us to evaluate the scalar comparison bounds and obtain explicit perturbation and sampling rates.

\begin{remark}[A condition involving only kernel sections]
\label{rem:indicator-primitive-kernel-only}
Since
\[
\abs{R_W(v)-R_W(u)}
\le\norm{W(v,\cdot)-W(u,\cdot)}_{L^1}
+\norm{W(\cdot,v)-W(\cdot,u)}_{L^1},
\]
condition \eqref{eq:indicator-primitive-network} follows from
\[
T(1+|r|)\norm{W(v,\cdot)-W(u,\cdot)}_{L^1}
+|r|T\norm{W(\cdot,v)-W(\cdot,u)}_{L^1}
\le\theta_j\chi_j(v-u).
\]
Thus the criterion depends only on $x_0$ and the row and column sections of $W$, up to null sets.
\end{remark}

\begin{corollary}[A log--Osgood class]
\label{cor:indicator-primitive-log-osgood}
Assume \cref{prop:indicator-primitive-osgood}.  Suppose that for each branch
there are $c_j>0$, $s_j>0$, $A_j\ge e|I_j|$, and
$\gamma_j\in[0,1]$ such that
\begin{equation}
\label{eq:indicator-primitive-log-separation}
\chi_j(s)\ge c_j\frac{s}{[\log(A_j/s)]^{\gamma_j}},
\qquad0<s\le\min\{s_j,|I_j|\}.
\end{equation}
With $\bar\gamma:=\max_j\gamma_j$, there are $C,A,a_*>0$ such that
\begin{equation}
\label{eq:indicator-primitive-log-tube}
\omega_{\rm pr}(a)\le Ca[\log(A/a)]^{\bar\gamma},
\qquad0<a\le a_*.
\end{equation}
Hence \eqref{eq:indicator-primitive-osgood} holds. If $W\ge0$, the hard solution is unique and satisfies the tube assumption in the stability and convergence results below.
\end{corollary}

\begin{example}[Log--Osgood uniqueness beyond bounded density]
\label{ex:indicator-rank-one-log-osgood}
Fix $0<\gamma\le1$, $c>0$, and $b\in\R$, and set
\[
x_0(u):=b+c\int_0^u[\log(e/s)]^{-\gamma}\,\dd s,
\qquad W(u,v):=w(v),\qquad r=0,
\]
where $w\in L^\infty(I)$ is nonzero and nonnegative.  The kernel rows are
identical and the increasing integrand gives
\[
x_0(v)-x_0(u)
\ge c\int_0^{v-u}[\log(e/s)]^{-\gamma}\,\dd s
\ge c_*\frac{v-u}{[\log(e/(v-u))]^\gamma}.
\]
Thus \cref{cor:indicator-primitive-log-osgood} applies with one branch and $\theta_1=0$. The density of $x_0(U)$ is unbounded:
\[
p_0(x)=\frac1c\left[\log\!\left(
\frac{e}{x_0^{-1}(x)}\right)\right]^\gamma\longrightarrow\infty
\quad\text{as }x\downarrow b.
\]
The feedback translates the entire profile by the same amount, preserving the separation bound.
\end{example}

\begin{corollary}[Explicit log--Osgood stability]
\label{cor:indicator-explicit-log-osgood}
Assume \cref{thm:indicator-osgood-stability}, let $M_1>0$, and suppose
that for some $a_0>0$, $A\ge ea_0$, $C_\omega\ge1$, and
$\gamma\in[0,1]$,
\begin{equation}
\label{eq:indicator-log-osgood-upper}
\omega(a)\le C_\omega a\Lambda(a)^\gamma,
\qquad \Lambda(a):=\log\frac Aa,
\qquad0<a\le a_0.
\end{equation}
For $\delta=\delta_p\in(0,1)$ set
\[
q_{p,\delta}:=
\begin{cases}
\delta^{p/(p+1)},&1\le p<\infty,\\
\delta,&p=\infty,
\end{cases}
\qquad K:=M_1(2C_\omega+1),
\]
and define
\begin{equation}
\label{eq:indicator-explicit-log-flow}
\mathcal Q_{\gamma,K,t}(q):=
\begin{cases}
A\exp\!\left(-\left([\log(A/q)]^{1-\gamma}
-(1-\gamma)Kt\right)^{1/(1-\gamma)}\right),&0\le\gamma<1,\\[1ex]
A^{1-e^{-Kt}}q^{e^{-Kt}},&\gamma=1.
\end{cases}
\end{equation}
The first line is used where the expression in parentheses is positive.  If
$\delta$ is sufficiently small that
$\mathcal Q_{\gamma,K,T}(q_{p,\delta})\le a_0/2$, then
\begin{align}
\label{eq:indicator-explicit-log-H}
M_1\sup_{t\le T}\norm{H_t^1-H_t^2}_{L^1}
&\le\mathcal Q_{\gamma,K,T}(q_{p,\delta})-q_{p,\delta},\\
\label{eq:indicator-explicit-log-X}
\sup_{t\le T}\norm{X_t^1-X_t^2}_{L^p}
&\le\mathcal Q_{\gamma,K,T}(q_{p,\delta})-q_{p,\delta}+\delta,\\
\label{eq:indicator-explicit-log-q}
\sup_{t\le T}\norm{\1_{\{X_t^1\le0\}}-\1_{\{X_t^2\le0\}}}_{L^1}
&\le(2C_\omega+1)\mathcal Q_{\gamma,K,T}(q_{p,\delta})
\Lambda\!\left(\mathcal Q_{\gamma,K,T}(q_{p,\delta})\right)^\gamma.
\end{align}
For $0\le\gamma<1$ the first two errors are bounded by
$\delta_p^{p/(p+1)-o(1)}$ at finite $p$.  At $\gamma=1$, with
$\vartheta=e^{-KT}$, they are
$O(\delta_p^{\frac p{p+1}\vartheta})$ for finite $p$ and
$O(\delta_\infty^\vartheta)$ at $p=\infty$.
\end{corollary}

\begin{corollary}[Convergence rates under a linear tube]
\label{cor:indicator-strong-bridge}
Suppose the target solution for $(W,x_0)$ satisfies \eqref{eq:indicator-reference-tube}. First let $W^{(n)}\ge0$ be bounded kernels and let $x_0^{(n)}$ be bounded profiles. If, for some $1\le p\le\infty$,
\[
W^{(n)}\to W\text{ in }L^p(I^2),
\qquad x_0^{(n)}\to x_0\text{ in }L^p(I),
\]
then the greatest hard solutions of the approximants converge to the target at the rate in \cref{thm:indicator-one-sided-stability}. This also applies to equal-block step representatives of deterministic finite nonnegative networks.

For the signed kernels and ramp solutions in \cref{cor:indicator-osgood-bridge}, with the same $\delta_{n,p}$,
\[
\sup_{t\le T}\norm{X_t^{n,\varepsilon_n}-X_t}_{L^p}
\le C_{T,p}\left(\varepsilon_n+\delta_{n,p}^{p/(p+1)}+\delta_{n,p}\right)
\]
for finite $p$, with the linear endpoint at $p=\infty$. This includes signed spectral approximations.
\end{corollary}

\begin{proof}
The nonnegative statement is \cref{thm:indicator-one-sided-stability} with the canonical existence supplied by \cref{thm:indicator-greatest}. For the ramp, if $d_t=\norm{H_t-H_t^{n,\varepsilon_n}}_{L^1}$ and $A_t$ is the data-error term, then
\[
\left|\1_{\{X_t\le0\}}-\ell_{\varepsilon_n}(X_t^{n,\varepsilon_n})\right|
\le
\1_{\{\abs{X_t}\le \norm W_{L^\infty}d_t+\abs{A_t}+\varepsilon_n\}}.
\]
The tube bound and Gronwall give the additional error $O(\varepsilon_n)$. A finite network is exactly the graphon equation for its equal-block step data.
\end{proof}

\begin{corollary}[Explicit sampled rates for log--Osgood tubes]
\label{cor:indicator-sampled-log-rates}
Assume \cref{thm:indicator-sampled-label} and
\eqref{eq:indicator-log-osgood-upper}, with $M_W>0$, and set
$K_{\rm sam}:=M_W(C_\omega+1)$. If
$\mathcal Q_{\gamma,K_{\rm sam},T}(\kappa_N(z))$ is defined and at most $a_0$, then on the
high-probability event of that theorem,
\begin{align*}
\sup_{t\le T}D_N(t)
&\le\frac{\mathcal Q_{\gamma,K_{\rm sam},T}(\kappa_N(z))
-\kappa_N(z)}{M_W},\\
\max_{i\le N}\sup_{t\le T}\abs{X_i^{N,+}(t)-X_t(U_i)}
&\le\mathcal Q_{\gamma,K_{\rm sam},T}(\kappa_N(z)),\\
\sup_{t\le T}\Wone\left(\frac1N\sum_i\delta_{X_i^{N,+}(t)},
\Law(X_t(U))\right)
&\le C\mathcal Q_{\gamma,K_{\rm sam},T}(\kappa_N(z)),\\
\sup_{t\le T}\left|\frac1N\sum_i\1_{\{X_i^{N,+}(t)\le0\}}
-\lambda\{X_t\le0\}\right|
&\le C\mathcal Q_{\gamma,K_{\rm sam},T}(\kappa_N(z))
\Lambda\!\left(\mathcal Q_{\gamma,K_{\rm sam},T}
(\kappa_N(z))\right)^\gamma.
\end{align*}
Thus $0\le\gamma<1$ gives $N^{-1/2+o(1)}$ for the first three
quantities.  At $\gamma=1$, if
$\vartheta_{\rm sam}:=e^{-K_{\rm sam}T}$, they are
$O_{\Prob}((\log N/N)^{\vartheta_{\rm sam}/2})$; the threshold-fraction
bound has one additional $O(\log N)$ factor.
\end{corollary}

\begin{proof}
For $\Gamma_\kappa=M_W\mathcal R_\kappa+\kappa$, the scalar comparison
equation and $\Gamma_\kappa\ge\kappa$ give locally
\[
\Gamma_\kappa'
=M_W[\omega(\Gamma_\kappa)+\kappa]
\le K_{\rm sam}\Gamma_\kappa\Lambda(\Gamma_\kappa)^\gamma.
\]
Apply \eqref{eq:indicator-explicit-log-flow} and then the four estimates in
\cref{thm:indicator-sampled-label}.
\end{proof}

\subsection{Density and transversality criteria}

We use \cref{ass:graphon-indicator-density} first to compare two initial profiles for a fixed kernel. We then give a derivative condition on the initial profile and kernel that verifies the density assumption before the feedback equation is solved.

\begin{corollary}[Fixed-graphon stability with respect to the initial profile]
\label{cor:graphon-indicator-initial-stability}
Let $W\in L^\infty(I^2)$ be fixed and let $x_0^1,x_0^2\in L^\infty(I)$. Suppose that the two data pairs $(W,x_0^1)$ and $(W,x_0^2)$ jointly satisfy \cref{ass:graphon-indicator-density} on $[0,T]$ with a common bound $M_\rho^W(T,T)$; that is, the density bound holds uniformly for both initial profiles and all admissible $h$. Let $H^i$ and $X^i$ be the corresponding solutions from \cref{thm:graphon-indicator-general}. Then there exists $C_T<\infty$, depending only on $T$, $|r|$, $\norm{W}_{L^\infty}$, and $M_\rho^W(T,T)$, such that
\[
\sup_{0\le t\le T}\norm{H_t^1-H_t^2}_{L^1(I)}
\le C_T\norm{x_0^1-x_0^2}_{L^\infty(I)},
\]
and consequently
\[
\sup_{0\le t\le T}\norm{X_t^1-X_t^2}_{L^1(I)}
\le \bigl(1+\norm{W}_{L^\infty}C_T\bigr)\norm{x_0^1-x_0^2}_{L^\infty(I)}.
\]
\end{corollary}

\begin{proof}[Proof sketch]
The density bound gives a linear tube with constant $2M_\rho^W(T,T)$. Apply the $p=\infty$ case of \cref{thm:indicator-one-sided-stability} with $W_1=W_2=W$.
\end{proof}

\begin{proposition}[Branchwise $L^1$--$C^1$ transversality implies \cref{ass:graphon-indicator-density}]
\label{prop:graphon-indicator-transverse}
Fix $T>0$ and $C>0$. Let $0=u_0<u_1<\cdots<u_J=1$ and set $I_j:=(u_{j-1},u_j)$. Assume that $x_0$ is $C^1$ on each $I_j$. For every branch $I_j$, assume that
\[
u\longmapsto W(u,\cdot),
\qquad
u\longmapsto W(\cdot,u)
\]
are $C^1$ maps from $I_j$ into $L^1(I)$. Denote their $L^1$ derivatives by $D_1W(u,\cdot)$ and $D_2W(\cdot,u)$, respectively, and suppose
\[
A_1:=\sup_{u\in\cup_j I_j}\norm{D_1W(u,\cdot)}_{L^1(I)}<\infty,
\qquad
A_2:=\sup_{u\in\cup_j I_j}\norm{D_2W(\cdot,u)}_{L^1(I)}<\infty.
\]
If, on each branch,
\[
\abs{x_0'(u)}\ge m_0 > C A_1 + |r|T(A_1+A_2),
\qquad u\in I_j,
\]
then, for every $t\in[0,T]$ and every $h\in L^\infty(I)$ with $\norm{h}_{L^\infty}\le C$, the map $u\mapsto\Psi_t(u,h)$ is $C^1$ on each branch, has at most $J$ monotone branches, and satisfies
\[
\abs{\partial_u\Psi_t(u,h)}
\ge m:=m_0-CA_1-|r|T(A_1+A_2)>0.
\]
Consequently \cref{ass:graphon-indicator-density} holds with $M_\rho^W(T,C)=J/m$.
\end{proposition}

\begin{proof}[Proof sketch]
Differentiating the $L^1$ pairings bounds the network contribution to $\partial_u\Psi_t$ by $CA_1+|r|T(A_1+A_2)$. The remaining margin is $m>0$, so each branch is monotone. Change of variables then gives the density bound $J/m$.
\end{proof}

\begin{remark}[Economic meaning of $A_1$ and $A_2$]
\label{rem:A1A2-economic}
The constants $A_1$ and $A_2$ measure how exposure rows and columns vary with institutional type. The condition requires variation in initial buffers to dominate the variation induced by the network over $[0,T]$. A smaller margin $m$ permits more mass near the distress threshold, as quantified by $J/m$.
\end{remark}

\subsection{Indicator approximation under uniform transversality}
\label{sec:graphon-indicator-bridge}

We compare kernels and initial profiles that satisfy a common branchwise slope condition. This condition gives existence and uniqueness for each member of the family, including signed kernels, and bounds the mass near every threshold by the same linear modulus. We can therefore bound path differences by the $L^\infty$ data error with constants uniform across the family.

\begin{definition}[Uniformly transverse family]
\label{def:uniformly-transverse-family}
Fix $T>0$. A family $\mathfrak F=\{(W^{\alpha},x_0^{\alpha})\}_{\alpha\in\mathcal A}\subset L^\infty(I^2)\times L^\infty(I)$ is called $(J,m)$-\emph{uniformly transverse} on $[0,T]$ if there exists a partition $0=u_0<u_1<\cdots<u_J=1$ such that for every $\alpha\in\mathcal A$, every $t\in[0,T]$, and every $h\in L^\infty(I)$ with $\norm{h}_{L^\infty}\le T$, the threshold profile
\[
\Psi_t^{\alpha}(u,h):=x_0^{\alpha}(u)+\mu t+r tR_{W^{\alpha}}(u)-\int_I W^{\alpha}(u,v)h(v)\,\dd v
\]
is piecewise $C^1$ on the branch partition $\{I_j\}_{j=1}^J$, has at most $J$ monotone branches, and satisfies
\[
\abs{\partial_u\Psi_t^{\alpha}(u,h)}\ge m
\qquad\text{for all }u\in \bigcup_{j=1}^J I_j.
\]
\end{definition}

\begin{proposition}[Geometric level-set control under transversality]
\label{prop:level-set-control}
Let $f:I\to\R$ be piecewise $C^1$ on a partition with at most $J$ monotone branches, and assume that $\abs{f'(u)}\ge m>0$ on each branch. Then, for every $\eta\ge0$,
\[
\lambda\bigl(\{u\in I:\abs{f(u)}\le \eta\}\bigr)\le \frac{2J}{m}\,\eta.
\]
Consequently, if $g:I\to\R$ is measurable and $\norm{f-g}_{L^\infty(I)}\le \eta$, then
\[
\norm{\1_{\{f\le 0\}}-\1_{\{g\le 0\}}}_{L^1(I)}\le \frac{2J}{m}\,\eta.
\]
\end{proposition}

\begin{proof}
The inverse of $f$ on each branch is $1/m$-Lipschitz, so the preimage of $[-\eta,\eta]$ has length at most $2\eta/m$. Summing over the branches proves the first bound. A mismatch of the indicators implies $|f(u)|\le|f(u)-g(u)|\le\eta$ almost everywhere, which gives the second.
\end{proof}

\begin{theorem}[Indicator graphon systems on uniformly transverse families]
\label{thm:uniformly-transverse-stability}
Let
\[
\mathfrak F=\{(W^{\alpha},x_0^{\alpha})\}_{\alpha\in\mathcal A}
\]
be a $(J,m)$-uniformly transverse family on $[0,T]$. Set
\[
M_W:=\sup_{\alpha\in\mathcal A}\norm{W^{\alpha}}_{L^\infty(I^2)}<\infty.
\]
Then, for every $\alpha\in\mathcal A$, the indicator-loss graphon equation \eqref{eq:graphon-equation} with data $(W^{\alpha},x_0^{\alpha})$ has a unique solution $(H^{\alpha},X^{\alpha})$ with
\[
H^{\alpha}\in C([0,T];L^1(I))\cap L^\infty([0,T]\times I),
\qquad
X^{\alpha}\in C([0,T];L^\infty(I)).
\]
Moreover, for any $\alpha_1,\alpha_2\in\mathcal A$ there exists $C_T<\infty$, depending only on $T$, $J$, $m$, $M_W$, and $|r|$, such that
\begin{equation}
\label{eq:two-dataset-H-stability}
\sup_{0\le t\le T}\norm{H_t^{\alpha_1}-H_t^{\alpha_2}}_{L^1(I)}
\le C_T\Bigl(\norm{x_0^{\alpha_1}-x_0^{\alpha_2}}_{L^\infty(I)}+\norm{W^{\alpha_1}-W^{\alpha_2}}_{L^\infty(I^2)}\Bigr),
\end{equation}
and consequently
\begin{equation}
\label{eq:two-dataset-X-stability}
\sup_{0\le t\le T}\norm{X_t^{\alpha_1}-X_t^{\alpha_2}}_{L^1(I)}
\le C_T\Bigl(\norm{x_0^{\alpha_1}-x_0^{\alpha_2}}_{L^\infty(I)}+\norm{W^{\alpha_1}-W^{\alpha_2}}_{L^\infty(I^2)}\Bigr).
\end{equation}
\end{theorem}

\begin{proof}[Proof sketch]
Branchwise change of variables gives a density bound $J/m$, so \cref{thm:graphon-indicator-general} supplies existence and uniqueness. For two members of the family, \cref{prop:level-set-control} bounds the indicator difference by the initial-profile error, the kernel error, and $M_W\norm{H_t^{\alpha_1}-H_t^{\alpha_2}}_{L^1}$. Gronwall's lemma gives the two estimates.
\end{proof}

\begin{theorem}[Indicator bridge under uniform transversality]
\label{thm:restricted-indicator-bridge}
Let $\{(W^{(K)},x_0^{(K)})\}_{K\ge 1}\cup\{(W,x_0)\}$ be a $(J,m)$-uniformly transverse family on $[0,T]$. Assume in addition that
\[
M_W:=\sup_{K\ge1}\norm{W^{(K)}}_{L^\infty(I^2)}\vee\norm{W}_{L^\infty(I^2)}<\infty.
\]
Let $(H^{K,\mathrm{ind}},X^{K,\mathrm{ind}})$ and $(H^{\mathrm{ind}},X^{\mathrm{ind}})$ denote the corresponding indicator-loss graphon solutions. Then there exists $C_T<\infty$, depending only on $T$, $J$, $m$, $M_W$, and $|r|$, such that
\begin{equation}
\label{eq:restricted-indicator-bridge-H}
\sup_{0\le t\le T}\norm{H_t^{K,\mathrm{ind}}-H_t^{\mathrm{ind}}}_{L^1(I)}
\le C_T\Bigl(\norm{x_0^{(K)}-x_0}_{L^\infty(I)}+\norm{W^{(K)}-W}_{L^\infty(I^2)}\Bigr),
\end{equation}
and
\begin{equation}
\label{eq:restricted-indicator-bridge-X}
\sup_{0\le t\le T}\norm{X_t^{K,\mathrm{ind}}-X_t^{\mathrm{ind}}}_{L^1(I)}
\le C_T\Bigl(\norm{x_0^{(K)}-x_0}_{L^\infty(I)}+\norm{W^{(K)}-W}_{L^\infty(I^2)}\Bigr).
\end{equation}
In particular, if $U\sim\mathrm{Unif}(I)$, then
\[
\sup_{0\le t\le T}\Wone\bigl(\Law(X_t^{K,\mathrm{ind}}(U)),\Law(X_t^{\mathrm{ind}}(U))\bigr)
\le C_T\Bigl(\norm{x_0^{(K)}-x_0}_{L^\infty(I)}+\norm{W^{(K)}-W}_{L^\infty(I^2)}\Bigr).
\]
\end{theorem}

\begin{proof}
Apply \cref{thm:uniformly-transverse-stability} with $\alpha_1=K$ and $\alpha_2=\infty$. The Wasserstein estimate follows by coupling both laws through the same $U\sim\mathrm{Unif}(I)$.
\end{proof}

\begin{remark}[$L^p$ approximation]
\label{rem:restricted-bridge-scope}
The common transversality bound controls path errors linearly in the $L^\infty$ data error. For nonnegative approximants converging in aligned $L^p$, $1\le p<\infty$, \cref{cor:indicator-strong-bridge} instead gives the sharp exponent $p/(p+1)$ from a linear tube along the target solution.
\end{remark}

\begin{corollary}[Sampled finite-\(N\) to graphon bridge on admissible uniformly transverse families]
\label{cor:indicator-sampled-bridge}
Assume \cref{thm:restricted-indicator-bridge} and the following sampling conditions.

\begin{enumerate}[label={\rm(\roman*)},leftmargin=2.4em]
\item Each \(W^{(K)}\) admits a bounded-factor representation
\[
W^{(K)}(u,v)=\frac1K\sum_{k=1}^K a_k^{(K)}(u)b_k^{(K)}(v),
\qquad
\sup_{K\ge1}\max_{1\le k\le K}
\left(
\norm{a_k^{(K)}}_{L^\infty}
+
\norm{b_k^{(K)}}_{L^\infty}
\right)<\infty.
\]

\item If $U\sim\mathrm{Unif}(I)$, then
\[
\mu_0^{(K)}=\Law\bigl(x_0^{(K)}(U),a_1^{(K)}(U),\ldots,a_K^{(K)}(U),b_1^{(K)}(U),\ldots,b_K^{(K)}(U)\bigr).
\]
This pushforward law identifies the rank-$K$ limit with the graphon solution for $(W^{(K)},x_0^{(K)})$.

\item The fixed-rank sampled indicator estimate of \cref{thm:indicator-finiteN} holds for the laws \(\mu_0^{(K)}\) with constants independent of \(K\). For instance, this is ensured by the conditional-density condition \eqref{eq:conditional-density} of \cref{rem:density-sufficient} with constants uniform in \(K\), together with the uniform factor bounds above.

\item For $N\ge2$, let $\mu_0^{N,K}$ be the empirical measure of $N$ i.i.d.\ samples from $\mu_0^{(K)}$, and let $\nu_t^{N,K,\ind}$ be the state law associated with a measurable selection satisfying the hypotheses of \cref{thm:indicator-finiteN}. Let $\nu_t^{\ind}$ denote the state law of the target graphon indicator solution.
\end{enumerate}

Then there exists \(C_T<\infty\), independent of \(N\) and \(K\), such that
\[
\begin{aligned}
\E\left[
\sup_{0\le t\le T}
\Wone\left(\nu_t^{N,K,\ind},\nu_t^{\ind}\right)
\right]
&\le
C_T\sqrt{\frac{K\log N}{N}}
+
C_T\E\Wone(\mu_0^{N,K},\mu_0^{(K)})
\\
&\quad+
C_T\left(
\norm{x_0^{(K)}-x_0}_{L^\infty(I)}
+
\norm{W^{(K)}-W}_{L^\infty(I^2)}
\right).
\end{aligned}
\]
\end{corollary}

\begin{proof}
Let \(\nu_t^{K,\ind}:=\Law(X_t^{K,\ind}(U))\) be the deterministic rank-\(K\) graphon indicator state law. By assumption, this is the same deterministic rank-\(K\) limit law associated with the type distribution \(\mu_0^{(K)}\). Projecting the joint-law estimate from \cref{thm:indicator-finiteN} onto the state coordinate gives
\[
\E\left[
\sup_{0\le t\le T}
\Wone\left(\nu_t^{N,K,\ind},\nu_t^{K,\ind}\right)
\right]
\le
C_T\sqrt{\frac{K\log N}{N}}
+
C_T\E\Wone(\mu_0^{N,K},\mu_0^{(K)}),
\]
with \(C_T\) independent of \(K\) by the uniform hypotheses. On the other hand, \cref{thm:restricted-indicator-bridge} gives
\[
\sup_{0\le t\le T}
\Wone\left(\nu_t^{K,\ind},\nu_t^{\ind}\right)
\le
C_T\left(
\norm{x_0^{(K)}-x_0}_{L^\infty(I)}
+
\norm{W^{(K)}-W}_{L^\infty(I^2)}
\right).
\]
The claim follows from the triangle inequality.
\end{proof}
\begin{remark}[Sampling schemes and growth of \(K\)]
\label{rem:indicator-sampled-bridge-sampling}
The sampled bridge requires a density bound for initial buffers conditional on the loadings. When the factors determine the latent label up to finitely many values, that conditional law is atomic. In the following construction, the factors are constant on each block and the initial buffers vary within it, giving a bounded conditional density.

Fix a partition $I=\bigcup_{j=1}^J I_j$ up to its endpoints and set $\delta:=\min_j\lambda(I_j)>0$. Let $A_j,B_j:[0,1]\to[0,M]$ be continuous. For $u\in I_j$, define $a(u,\theta)=A_j(\theta)$ and $b(u,\theta)=B_j(\theta)$, and let $\theta_k^{(K)}=(k-\tfrac12)/K$. Set
\[
a_k^{(K)}(u)=a(u,\theta_k^{(K)}),\qquad
b_k^{(K)}(u)=b(u,\theta_k^{(K)}),\qquad
W^{(K)}(u,v)=\frac1K\sum_{k=1}^K a_k^{(K)}(u)b_k^{(K)}(v).
\]
Because there are finitely many block pairs, midpoint Riemann sums give uniform convergence to
\[
W(u,v)=\int_0^1 a(u,\theta)b(v,\theta)\,\dd\theta.
\]
Let $x_0^{(K)}=x_0$, where $x_0\in L^\infty(I)$ is piecewise $C^1$, strictly monotone on each $I_j$, and $|x_0'|\ge m_0>0$ on every branch. The kernel-dependent terms in the threshold profile are constant on each $I_j$; hence the family is $(J,m_0)$-uniformly transverse. If $U\sim\mathrm{Unif}(I)$ and $\mu_0^{(K)}$ is the pushforward law specified in assumption~(ii), then the factor vector is constant on each partition interval. Conditional on a factor value, the law of $U$ is normalized Lebesgue measure on the union of intervals carrying that value. Write $A^{(K)}=(a_1^{(K)}(U),\ldots,a_K^{(K)}(U))$ and $B^{(K)}=(b_1^{(K)}(U),\ldots,b_K^{(K)}(U))$, and let $\rho_K=\Law(A^{(K)},B^{(K)})$. A branchwise change of variables gives
\[
\left\|
(\alpha,\gamma)\longmapsto
\bigl\|f_{X_0\mid A^{(K)},B^{(K)}}(\cdot\mid\alpha,\gamma)\bigr\|_{L^\infty(\R)}
\right\|_{L^\infty(\rho_K)}
\le \frac{1}{\delta m_0},
\]
uniformly in $K$. Indeed, if a factor value occurs on $s$ intervals, their total length is at least $s\delta$, while the unnormalized density is at most $s/m_0$. The nonnegative factors also give nonnegative sampled exposure matrices, so \cref{cor:indicator-canonical-sample} supplies a Borel selection. Thus all hypotheses of \cref{cor:indicator-sampled-bridge} hold for this construction.

If \(K=K(N)\) grows with \(N\), the bound in \cref{cor:indicator-sampled-bridge} yields convergence provided
\[
\sqrt{\frac{K(N)\log N}{N}}\to0,
\qquad
\E\Wone(\mu_0^{N,K(N)},\mu_0^{(K(N))})\to0,
\]
and
\[
\norm{x_0^{(K(N))}-x_0}_{L^\infty(I)}
+
\norm{W^{(K(N))}-W}_{L^\infty(I^2)}
\to0.
\]

\end{remark}

\begin{corollary}[Block kernels with smooth trigonometric corrections]
\label{cor:trigonometric-family}
Let $0=u_0<u_1<\cdots<u_J=1$ and write
\[
W(u,v)=W^{\mathrm{blk}}(u,v)+W^{\mathrm{sm}}(u,v),
\qquad
W^{\mathrm{blk}}(u,v)=\sum_{i,j=1}^J c_{ij}\1_{I_i}(u)\1_{I_j}(v),
\]
where $W^{\mathrm{sm}}$ is $C^1$ on each rectangle $I_i\times I_j$ and admits a factorized trigonometric expansion
\[
W^{\mathrm{sm}}(u,v)=\sum_{m=1}^\infty \sigma_m\phi_m(u)\psi_m(v)
\]
with $\phi_m,\psi_m\in C^1(I)$ and
\[
\sum_{m=1}^\infty |\sigma_m|\Bigl(\norm{\phi_m}_{L^\infty}\norm{\psi_m}_{L^\infty} + \norm{\phi_m'}_{L^\infty}\norm{\psi_m}_{L^\infty} + \norm{\phi_m}_{L^\infty}\norm{\psi_m'}_{L^\infty}\Bigr)<\infty.
\]
Define the partial sums
\[
W^{(K)}(u,v):=W^{\mathrm{blk}}(u,v)+\sum_{m=1}^K \sigma_m\phi_m(u)\psi_m(v).
\]
If $x_0\in L^\infty(I)$ is piecewise $C^1$ on the same branch partition and satisfies
\[
\inf_{u\in \cup_j I_j}\abs{x_0'(u)}
> T A_1^\ast + |r|T(A_1^\ast+A_2^\ast),
\]
where
\[
A_1^\ast:=\sup_{K\ge 1}\sup_{u\in\cup_j I_j}\int_I \abs{\partial_1 W^{(K)}(u,v)}\,\dd v,
\qquad
A_2^\ast:=\sup_{K\ge 1}\sup_{u\in\cup_j I_j}\int_I \abs{\partial_2 W^{(K)}(v,u)}\,\dd v,
\]
then $\{(W^{(K)},x_0)\}_{K\ge 1}\cup\{(W,x_0)\}$ is a $(J,m)$-uniformly transverse family for some $m>0$. Hence \cref{thm:restricted-indicator-bridge} applies. The periodic Gaussian example in \cref{sec:numerics-piecewise-indicator,sec:numerics-piecewise-indicator-trig-bridge} has this form.
\end{corollary}

\begin{proof}
The series assumptions imply uniform convergence of $W^{(K)}$ to $W$ in $L^\infty(I^2)$ and uniform convergence of the branchwise first derivatives on each smooth rectangle. In particular, the quantities $A_1^\ast$ and $A_2^\ast$ are finite and dominate the corresponding derivative-integral bounds for both the limit kernel and all its partial sums. Applying \cref{prop:graphon-indicator-transverse} with the common partition and these uniform bounds yields the claim.
\end{proof}

\subsection{Other exposure networks}\label{sec:other-applications}

We can interpret the state and factors in other dense exposure networks using the same distinction between a buffer and the channels that transmit losses. An insurer or reinsurer holds a solvency buffer and receives payments on ceded-risk or retrocession claims, which determine its exposures to counterparties. A firm in a trade-credit network holds working capital and depends on payments from firms in other trading tiers. Participants in centrally cleared commodity and energy markets share clearing venues or margin channels, giving a factor structure analogous to the multiple-CCP model.

\section{Numerical implementation and supporting experiments}\label{sec:appendix-robustness}\label{app:additional}

We compute the finite-rank paths and graphon approximations in \cref{sec:numerics} from the population and exposure specifications below. Integration and grid-refinement checks distinguish numerical discrepancies from the population, smoothing, and kernel-approximation errors reported in the main text. The sovereign-portfolio calculations separate changes in channel coverage from changes in the scenario buffers, while resampling measures variation around the fixed empirical type law.

\subsection{Parameters and integration for the factor examples}\label{app:lowrank-parameters}
The four Gaussian examples use the group parameters in \cref{tab:lowrank-parameters} and the normalization $e_{ij}=K^{-1}a_i\cdot b_j$. Their respective $(\mu,r,T)$ values are $(-0.005,0,2)$, $(-0.008,0,2.5)$, $(-0.008,0,2.3)$, and $(-0.006,0,2.2)$. The dimension is $K=1$ in the rank-one model and $K=2$ in the other three.
\begin{table}[H]\centering\footnotesize
\setlength{\tabcolsep}{4pt}
\caption{Complete Gaussian-mixture parameters. Initial buffers conditional on group $g$ have law $N(m_g,\sigma_g^2)$. Vectors list factor coordinates in their fixed order.}
\label{tab:lowrank-parameters}
\begin{tabular}{llrrrcc}\toprule
Architecture & Group & $p_g$ & $m_g$ & $\sigma_g$ & $a_g$ & $b_g$\\\midrule
Rank-one mean field & Population & 1 & 0.24 & 0.18 & $(0.7)$ & $(0.7)$ \\
Core-periphery & Core & 0.25 & 0.3 & 0.18 & $(0.8,0.2)$ & $(0.8,0.3)$ \\
 & Periphery & 0.75 & 0.23 & 0.18 & $(0.3,0.7)$ & $(0.2,0.7)$ \\
Overlapping CCP channels & CCP 1 & 0.4 & 0.22 & 0.16 & $(0.9,0.1)$ & $(0.8,0.1)$ \\
 & CCP 2 & 0.35 & 0.25 & 0.18 & $(0.1,0.9)$ & $(0.1,0.8)$ \\
 & Overlap & 0.25 & 0.2 & 0.16 & $(0.6,0.6)$ & $(0.6,0.6)$ \\
Bank-NBFI multiplex & Banks & 0.7 & 0.28 & 0.18 & $(0.8,0.25)$ & $(0.7,0.3)$ \\
 & NBFIs & 0.3 & 0.24 & 0.16 & $(0.35,0.9)$ & $(0.25,0.9)$ \\
\bottomrule\end{tabular}\end{table}

For quantile-matched populations, we obtain group counts by largest-remainder rounding of $Np_g$. The $n_g$ initial buffers in group $g$ are
\[
m_g+\sigma_g\Phi_{\mathrm N}^{-1}\!\left(\frac{j-1/2}{n_g}\right),\qquad j=1,\ldots,n_g.
\]
Independent populations are sampled from the same complete-type mixture, with the generator and seed specified in \cref{sec:appendix-robustness-mc}.

We integrate the limiting Gaussian feedback equation \eqref{eq:numerical-beta} with DOP853, relative tolerance $10^{-10}$, and absolute tolerance $10^{-12}$. For finite populations, the velocity stays constant between successive threshold events, allowing integration by affine segments.

For a finite population, let $d_g$ be the number of distressed institutions in group $g$ on an interval between events. Then $\dot\beta_k=\sum_g b_{g,k}d_g/N$ and each group buffer has constant velocity $\mu-K^{-1}a_g\cdot\dot\beta<0$. The next event is the earliest zero crossing among the next solvent buffers in the sorted groups. Updating its group count and repeating integrates the finite dynamics by affine segments. Strictly negative velocities exclude tangencies and recoveries in these four population-comparison examples. The separate four-institution example in \cref{sec:recovery-example} updates both entries and exits; its five crossing times are $0.4$, $1$, $1.6$, $3$, and $3.7$, with nonzero crossing velocities.

To compute uniform threshold-fraction errors for the monotone finite populations, we evaluate both sides of every event and $t=0,T$, comparing the finite and limiting paths at the same times. The limiting threshold fraction is increasing, so these evaluations give the maximum difference on every interval where the finite fraction is constant. Other trajectory maxima for these factor experiments use $1001$ equally spaced times unless specified otherwise.

The directed example has two equally weighted groups with $a=(2.0,0.6)$, $b=(0.6,2.0)$, mean buffers $(0.50,0.45)$, common standard deviation $0.18$, outside drift $\mu=0.015$, and horizon $T=2.5$. We compare $r=0$ and $r=0.05$. The imbalance terms are $R_1=1.82$ and $R_2=-1.82$, and the initial net-lender fraction is $\Phi_{\mathrm N}(-0.50/0.18)=0.274\%$. The same Gaussian-feedback integrator is used, with $\mu+rR_g$ in place of $\mu$ for group $g$.

For the multiplex heatmap, only $a_{g,2}$ and the NBFI mean are changed; $K=2$, the first layer, the group probabilities, and all other parameters stay fixed. The grid uses $\omega_2\in[0,3]$ in steps of $0.05$ and $\Delta_{\mathrm{NBFI}}\in[0,0.16]$ in steps of $0.004$. The factor normalization is unchanged even when the second layer is set to zero.

\subsection{Independent sampling and quantile matching}\label{sec:appendix-robustness-mc}
For each of the four architectures, we draw $100$ independent populations of size $N=1600$ from the complete-type mixture in \cref{tab:lowrank-parameters} and integrate each finite system between threshold events. \Cref{tab:mc-summary} gives terminal summaries; the intervals are empirical $2.5\%$ and $97.5\%$ quantiles, computed by linear interpolation between order statistics. \Cref{fig:mc-benchmark} shows how the sampling ranges change along the trajectories.

Random draws use NumPy's PCG64 generator with seed $20261026$. A single stream is used for the four architectures in \cref{tab:lowrank-parameters} order, followed by the rank-one convergence experiment with sample sizes in increasing order. Within each replication, group counts are drawn multinomially, followed by independent standard-normal draws for each group in table order.
\begin{table}[tbp]\centering\small
\caption{Terminal threshold fractions for $100$ independent populations at $N=1600$. Sampling ranges describe variation between populations.}
\label{tab:mc-summary}
\begin{tabular}{lrrrr}\toprule
Example & Limit & Mean & Median & 95\% sampling range\\\midrule
Rank-one mean field & 0.4062 & 0.4067 & 0.4066 & [0.3153,\,0.5124] \\
Core-periphery & 0.2226 & 0.2281 & 0.2256 & [0.1852,\,0.2822] \\
Overlapping CCP channels & 0.2417 & 0.2423 & 0.2438 & [0.2040,\,0.2943] \\
Bank-NBFI multiplex & 0.1304 & 0.1304 & 0.1291 & [0.1022,\,0.1613] \\
\bottomrule\end{tabular}\end{table}
\begin{figure}[tbp]\centering
\includegraphics[width=\textwidth]{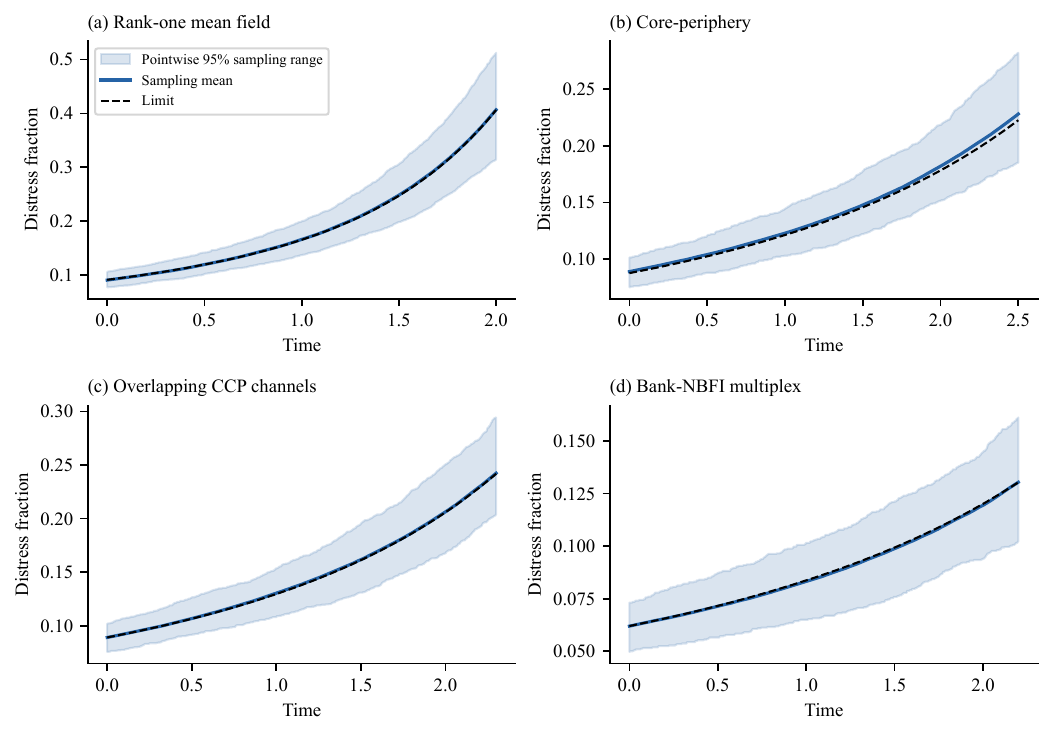}
\caption{Finite-population sampling for the four exact low-rank models. Solid curves show the mean of $100$ independent threshold-fraction trajectories; shaded regions are pointwise empirical $95\%$ sampling ranges; dashed curves show the Gaussian limits.}
\label{fig:mc-benchmark}\end{figure}

\subsection{Integration accuracy for the finite-rank examples}
\begin{table}[tbp]\centering\small
\caption{Finite-rank integration comparisons. The Euler columns give maximum state differences from the event solution at common Euler times, for $N=1600$. The final column is the maximum change in the limiting threshold fraction when DOP853 tolerances are reduced from $(10^{-10},10^{-12})$ to $(10^{-12},10^{-14})$, measured at $1001$ equally spaced times.}
\label{tab:exact-dt-by-example}
\begin{tabularx}{\textwidth}{>{\raggedright\arraybackslash}Xrrr}\toprule
Example & \shortstack{Euler\\$\Delta t=5\cdot10^{-4}$} & \shortstack{Euler\\$\Delta t=2.5\cdot10^{-4}$} & \shortstack{Limit\\tolerance change}\\\midrule
Rank-one mean field & $7.21\times10^{-5}$ & $3.45\times10^{-5}$ & $6.64\times10^{-10}$ \\
Core-periphery & $1.31\times10^{-5}$ & $6.49\times10^{-6}$ & $8.42\times10^{-12}$ \\
Overlapping CCP channels & $1.62\times10^{-5}$ & $8.80\times10^{-6}$ & $9.32\times10^{-12}$ \\
Bank-NBFI multiplex & $7.00\times10^{-6}$ & $3.72\times10^{-6}$ & $5.85\times10^{-12}$ \\
\bottomrule\end{tabularx}\end{table}
The finite-population figures and error table use the event solutions. The Euler state discrepancies in \cref{tab:exact-dt-by-example} decrease on step refinement. Threshold fractions jump by $1/N$ at individual crossings, so state differences provide a continuous measure of integration error; the reported population errors retain the threshold-fraction metric.

\subsection{Empirical fixed-rank indicator convergence rate}\label{sec:appendix-poc-rate}
To measure convergence under independent sampling, we generate $100$ populations in the rank-one Gaussian benchmark at each $N\in\{100,200,400,800,1600,3200\}$. For each replication we evaluate $\max_t|\beta_t^N-\beta_t|$ on $1001$ equally spaced times, with the finite path obtained by event integration and the limiting path by DOP853. We then average these errors and fit a line to the logarithm of the mean error against $\log N$.
The fixed bounded loadings and bounded conditional Gaussian density satisfy the hypotheses of \cref{thm:indicator-finiteN}, which gives an $O(\sqrt{\log N/N})$ bound on the expected feedback-path error. The fitted slope of $-0.465$ over the displayed sample sizes is consistent with this bound (\cref{fig:indicator-poc-rate}).

\begin{figure}[tbp]\centering
\includegraphics[width=.75\textwidth]{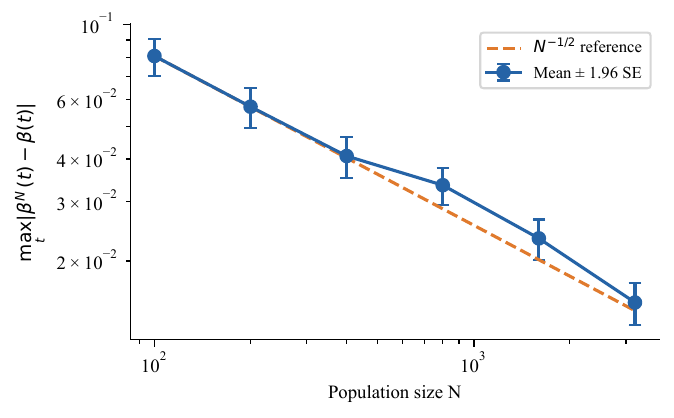}
\caption{Mean feedback-path error under independent sampling in the rank-one benchmark. Bars show $1.96$ estimated standard errors of the Monte Carlo mean from $100$ replications at each $N$. The dashed line has slope $-1/2$.}
\label{fig:indicator-poc-rate}\end{figure}

\subsection{Multiple-CCP and multiplex examples}\label{app:extra-examples}
\begin{figure}[t]\centering
\includegraphics[width=\textwidth]{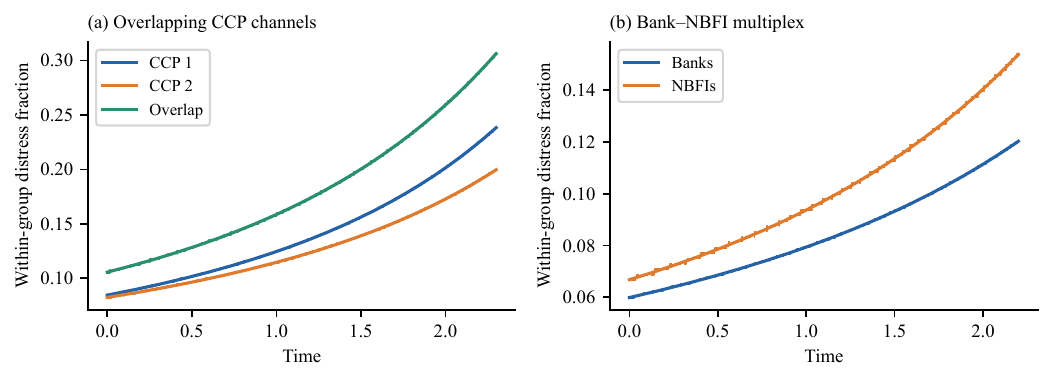}
\caption{Within-group threshold fractions for the overlapping-CCP and bank--NBFI models in \cref{tab:lowrank-parameters}. Solid curves are Gaussian limits; dashed curves are quantile-matched finite trajectories at $N=1600$.}
\label{fig:numerics-2}\end{figure}
In the CCP example, the overlap group contributes to and receives losses through both channels. Its limiting terminal fraction is $30.62\%$, compared with $23.83\%$ and $19.96\%$ for the two specialist groups. Both the initial buffers and the factor loadings differ across these groups, so their terminal fractions reflect both differences. In the multiplex example, banks and NBFIs have terminal fractions $12.03\%$ and $15.39\%$, respectively. The two-dimensional feedback reproduces the separate group paths in \cref{fig:numerics-2}. The second-layer intensity comparison in \cref{fig:numerics-4} keeps the normalization fixed.

\subsection{Graphon discretization and scalar regularization}
\label{app:graphon-methods}

The piecewise graphon uses $n=2000$ midpoint labels with weights $1/n$ and 2001 output times on $[0,2]$. Its nonnegative weights and drift $\mu=-0.06$ make every state strictly decreasing. Between successive threshold crossings the hard system is affine; the algorithm advances to the next crossing and adds the newly distressed labels to the loss vector. The periodic component is applied by discrete Fourier convolution. Smoothed paths use DOP853 with relative tolerance $10^{-9}$, absolute tolerance $10^{-11}$, and maximum step $0.01$.

For the scalar Osgood experiment, the positive-ramp vector field is evaluated through the exact antiderivative of each distribution function. DOP853 uses relative tolerance $2\times10^{-12}$, absolute tolerance $2\times10^{-14}$, maximum step $0.002$, and 2001 output times on $[0,1]$. Fitted exponents use widths $2^{-9},\ldots,2^{-12}$.

\begin{table}[tbp]
\centering\small
\caption{Fourier approximation of the hard-threshold graphon. The block component remains exact. Errors compare paths on the common midpoint grid and output times.}
\label{tab:piecewise-indicator-bridge}
\begin{tabular}{rrrrr}\toprule
$M$ & Kernel $L^\infty$ error & State $\max_t L^1$ error & $\max_t|h^{(M)}-h|$ & $h^{(M)}(T)$\\\midrule
1 & $7.82\times10^{-2}$ & $6.26\times10^{-3}$ & $3.50\times10^{-3}$ & 0.2230\\
2 & $4.04\times10^{-2}$ & $2.74\times10^{-3}$ & $3.00\times10^{-3}$ & 0.2225\\
4 & $7.13\times10^{-3}$ & $2.11\times10^{-4}$ & $5.00\times10^{-4}$ & 0.2195\\
8 & $5.42\times10^{-5}$ & $8.38\times10^{-7}$ & 0 & 0.2195\\
16 & $5.76\times10^{-11}$ & $6.15\times10^{-13}$ & 0 & 0.2195\\\bottomrule
\end{tabular}
\end{table}

To distinguish ODE error from spatial quadrature, we tighten the tolerances and double the spatial grid in separate comparisons (\cref{tab:graphon-resolution}). We compare piecewise states by linear interpolation within each branch. For the 30-factor model, we reconstruct the coarse and fine feedback paths at the same 4000 labels using the analytic factors.

\begin{table}[tbp]
\centering\small
\caption{Graphon solver and grid refinement. Tightened ODE tolerances are a factor 100 smaller. Spatial refinement doubles the midpoint grid from 2000 to 4000 labels.}
\label{tab:graphon-resolution}
\begin{tabular}{lrr}\toprule
Comparison & State $\max_t L^1$ difference & Threshold difference\\\midrule
Piecewise ramp, $\varepsilon=0.01$: tighter ODE tolerance & $1.41\times10^{-8}$ & 0\\
Piecewise hard rule: spatial refinement & $1.15\times10^{-7}$ & $2.50\times10^{-4}$\\
30-factor ramp: tighter ODE tolerance & $1.22\times10^{-8}$ & 0\\
30-factor ramp: common-label spatial refinement & $1.90\times10^{-8}$ & 0\\\bottomrule
\end{tabular}
\end{table}

\subsection{The graphon truncation experiment}
\label{app:graphon-truncation}

Consider the directed positive series
\[
W(u,v)=\sum_{m\ge1}0.35m^{-1.8}a_m(u)b_m(v),
\]
\[
a_m(u)=1+0.45\sin(2\pi mu)+0.15\cos(2\pi(m+1)u),
\]
\[
b_m(v)=1+0.40\cos(2\pi mv+0.6)+0.12\sin(2\pi(m+2)v).
\]
The factors are positive and have unit integral. We set $x_0(u)=0.13+0.08\sin(2\pi u)-0.04\cos(6\pi u)$, $\mu=-0.006$, $r=0$, and $T=2$. The loss is the positive-side ramp of width $0.02$. The reference retains 30 factors at 2000 midpoint labels. Its feedback coordinates are integrated by DOP853 with relative tolerance $10^{-10}$, absolute tolerance $10^{-12}$, maximum step $0.01$, and the output times used above.

\begin{figure}[tbp]
\centering
\includegraphics[width=\textwidth]{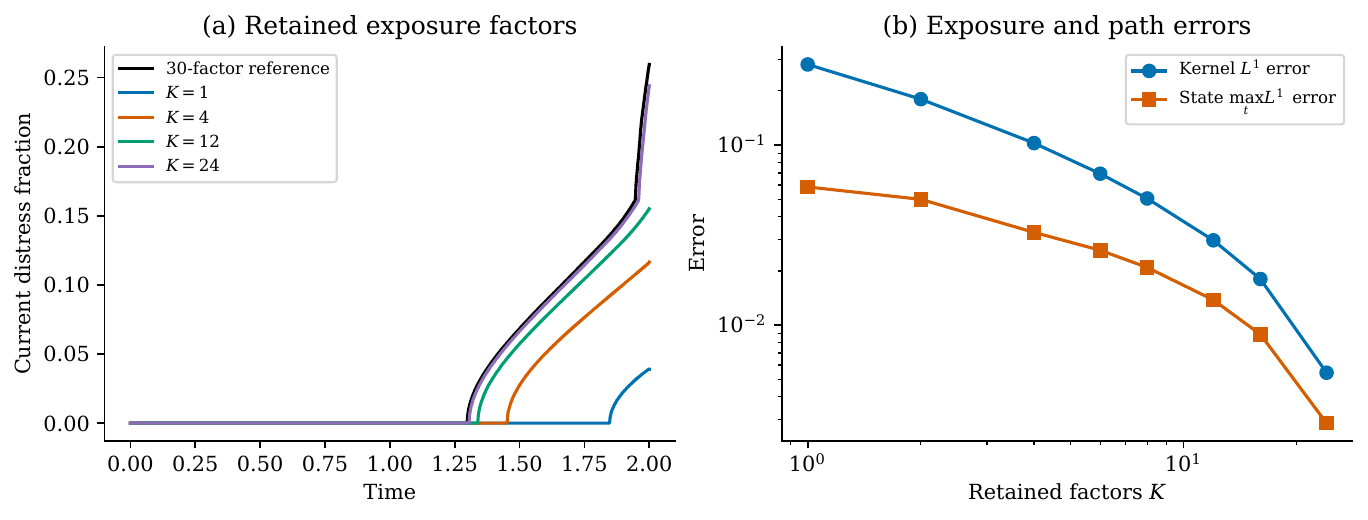}
\caption{Positive-factor truncation with a bounded Lipschitz loss. (a) Current distress fractions under the 30-factor reference and selected truncations. (b) Kernel $L^1$ errors and pathwise state errors relative to that reference.}
\label{fig:numerics-3}
\end{figure}

Removing positive exposure factors weakens transmission, delays activation, and reduces the subsequent loss flow (\cref{fig:numerics-3}). At $K=4,12,24$, the kernel errors relative to the 30-factor reference are $0.103$, $0.0296$, and $0.00544$, and the state errors are $0.0327$, $0.0138$, and $0.00285$. The reference itself omits an infinite tail with exact $L^1$ mass $0.35\zeta(1.8,31)\approx0.0284$. Including that tail in the error calculation gives kernel errors $0.131$, $0.0580$, and $0.0339$ relative to the infinite series. The aligned estimate in \cref{thm:graphon-wellposed} also controls the state discrepancy associated with this remaining tail.

\subsection{Sensitivity to the transversality margin}
\label{sec:numerics-piecewise-sensitivity}

Keeping the piecewise kernel fixed, let the right branch of $x_0$ have slope $m_0$ and the left branch slope $m_0+0.20$, with $x_0(0)=-0.18$ and jump $0.025$ at $1/2$. The sufficient margin is $m=m_0-TA_1$, where $TA_1\approx0.640$. Reducing the slopes increases the initial population near distress and lowers this margin.

\begin{table}[tbp]
\centering\small
\caption{Branch-slope sensitivity at positive-ramp width $\varepsilon=0.01$. The hard terminal fraction and smoothing error both increase as the threshold layer becomes more concentrated.}
\label{tab:piecewise-sensitivity}
\begin{tabular}{rrrr}\toprule
$m_0$ & $m_0-TA_1$ & Hard $h(T)$ & State $\max_t L^1$ smoothing error\\\midrule
1.75 & 1.11 & 0.2195 & $2.35\times10^{-3}$\\
0.90 & 0.26 & 0.4915 & $4.86\times10^{-3}$\\
0.70 & 0.06 & 0.6040 & $5.34\times10^{-3}$\\\bottomrule
\end{tabular}
\end{table}

Along this family, the state discrepancy more than doubles as the margin falls from $1.11$ to $0.06$ (\cref{tab:piecewise-sensitivity}). A given buffer displacement can move more institutions across zero when the initial profile is flatter, in agreement with the $J/m$ density bound.

\subsection{EBA portfolios and numerical implementation}
\label{app:eba-construction}

The source is the sovereign-exposure file of the 2025 EBA transparency exercise, together with its data dictionary and institution and country metadata. We select item 2520810, period 202506, accounting-portfolio code zero, and maturity code eight, the reported total over residual maturities. We exclude the aggregate \emph{All other banks} record, replace five negative rounding residuals of at most EUR~1,000 by zero, and sum each institution's disclosed country or regional holdings. Missing bucket entries contribute zero to this disclosed-book construction. Three named institutions have zero total holdings and are excluded. The resulting matrix has 116 rows and 44 disclosed buckets. Selecting buckets with at least $0.1\%$ of aggregate holdings gives the 38-factor matrix used in the text.

We integrate the smoothed systems through their exact feedback coordinates with the adaptive DOP853 method, relative tolerance $2\times10^{-10}$ and absolute tolerance $2\times10^{-12}$. Path comparisons use 601 equally spaced observation times on $[0,3]$. Reducing the tolerances by a factor of 100 in the 116-bank baseline changes its buffer paths by less than $5\times10^{-11}$; direct integration of all 116 states gives the same accuracy. For hard losses with $\mu=r=0$ and nonnegative exposures, buffers are nonincreasing. The active loss set remains fixed between successive threshold hits, so we integrate the buffer paths as affine segments and activate each newly distressed bank at its hitting time.

\paragraph{The six-bank construction.}
\Cref{tab:holdings} reports the portfolio entries used to form the four-channel illustration. Calculations use the unrounded disclosed amounts. The corresponding shares are shown in \cref{fig:shares}.

\begin{table}[htbp]
\centering\small
\caption{Six-bank holdings in million EUR (rounded for display) and initial buffers at $x_\ast=0.40$.}
\label{tab:holdings}
\begin{tabular}{lrrrrr}
\toprule
Bank & France & Germany & Italy & Netherlands & Initial buffer\\
\midrule
BNP & 54,638 & 17,862 & 23,455 & 608 & 0.06459\\
DB & 17,962 & 14,776 & 22,414 & 677 & 0.06643\\
UCG & 7,121 & 17,405 & 45,760 & 16 & 0.03449\\
SAN & 7,898 & 1,642 & 17,284 & 221 & 0.00068\\
ING & 7,288 & 9,514 & 2,339 & 4,630 & 0.17864\\
ISP & 12,810 & 2,175 & 42,616 & 1,108 & -0.00952\\
\bottomrule
\end{tabular}
\end{table}

BNP denotes BNP Paribas, DB Deutsche Bank, UCG UniCredit, SAN Banco Santander, ING ING Groep, and ISP Intesa Sanpaolo. The four countries cover between $12.7\%$ and $47.1\%$ of these banks' full disclosed sovereign books; \cref{tab:coverage} gives the institution-level coverage.

\begin{table}[htbp]
\centering\small
\caption{Coverage of the four selected countries and the Italian portfolio share.}
\label{tab:coverage}
\begin{tabular}{lrrr}
\toprule
Bank & Four-country coverage & Italy/full book & Italy/selected four\\
\midrule
BNP & 26.4\% & 6.4\% & 24.3\%\\
DB & 30.9\% & 12.4\% & 40.1\%\\
UCG & 47.1\% & 30.6\% & 65.1\%\\
SAN & 12.7\% & 8.1\% & 63.9\%\\
ING & 23.4\% & 2.3\% & 9.8\%\\
ISP & 44.0\% & 31.9\% & 72.6\%\\
\bottomrule
\end{tabular}
\end{table}

\begin{figure}[htbp]
\centering
\includegraphics[width=.70\textwidth]{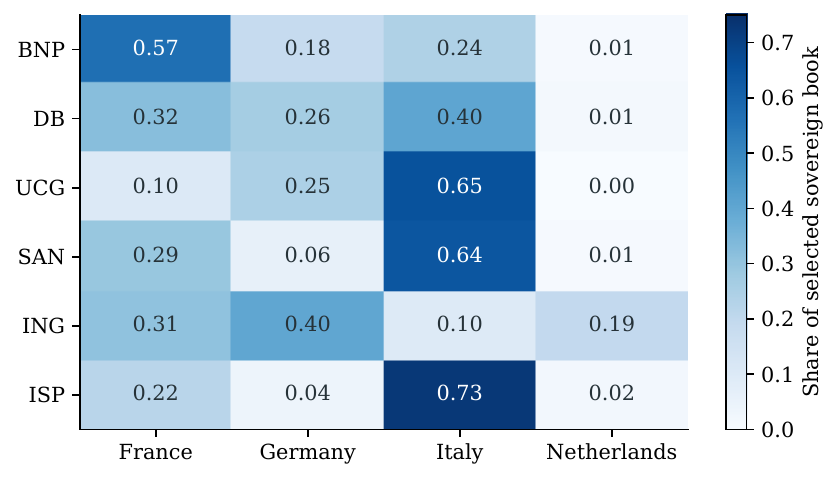}
\caption{Shares within each bank's selected four-country sovereign book. The full-book coverage is reported in \cref{tab:coverage}.}
\label{fig:shares}
\end{figure}

\Cref{tab:six-bank-intensity} varies only the common channel intensity. The initial buffers and shock weights stay fixed, so the additional threshold crossings come from stronger transmission.
\begin{table}[htbp]
\centering\small
\caption{Six-bank sensitivity to channel intensity.}
\label{tab:six-bank-intensity}
\begin{tabular}{rrrr}
\toprule
Intensity multiplier & $\lambda_E$ & Terminal mean distress rate & Terminal count\\
\midrule
0.5 & 0.0375 & 0.0751 & 2/6\\
1.0 & 0.0750 & 0.1696 & 2/6\\
2.0 & 0.1500 & 0.5691 & 5/6\\
\bottomrule
\end{tabular}
\end{table}

The leading singular values of the six-bank kernel are $0.4890$, $0.0949$, $0.0487$, and $0.0167$. Its rank-one, rank-two, and rank-three SVD truncations are all nonnegative. Starting each system from the same initial buffers gives the paths in \cref{tab:truncation,fig:distress}. Every approximation ends with two banks below zero, but the mean loss rates reveal differences that this common count does not capture.

\begin{table}[htbp]
\centering\small
\caption{SVD truncations of the six-bank kernel. Path gap is the maximum absolute difference in mean distress rates.}
\label{tab:truncation}
\begin{tabular}{rrrr}
\toprule
Rank & Relative Frobenius error & Terminal mean distress rate & Path gap\\
\midrule
1 & 0.2155 & 0.1204 & 0.0492\\
2 & 0.1027 & 0.1566 & 0.0129\\
3 & 0.0334 & 0.1567 & 0.0129\\
4 & 0.0000 & 0.1696 & 0.0000\\
\bottomrule
\end{tabular}
\end{table}

\begin{figure}[htbp]
\centering
\includegraphics[width=.98\textwidth]{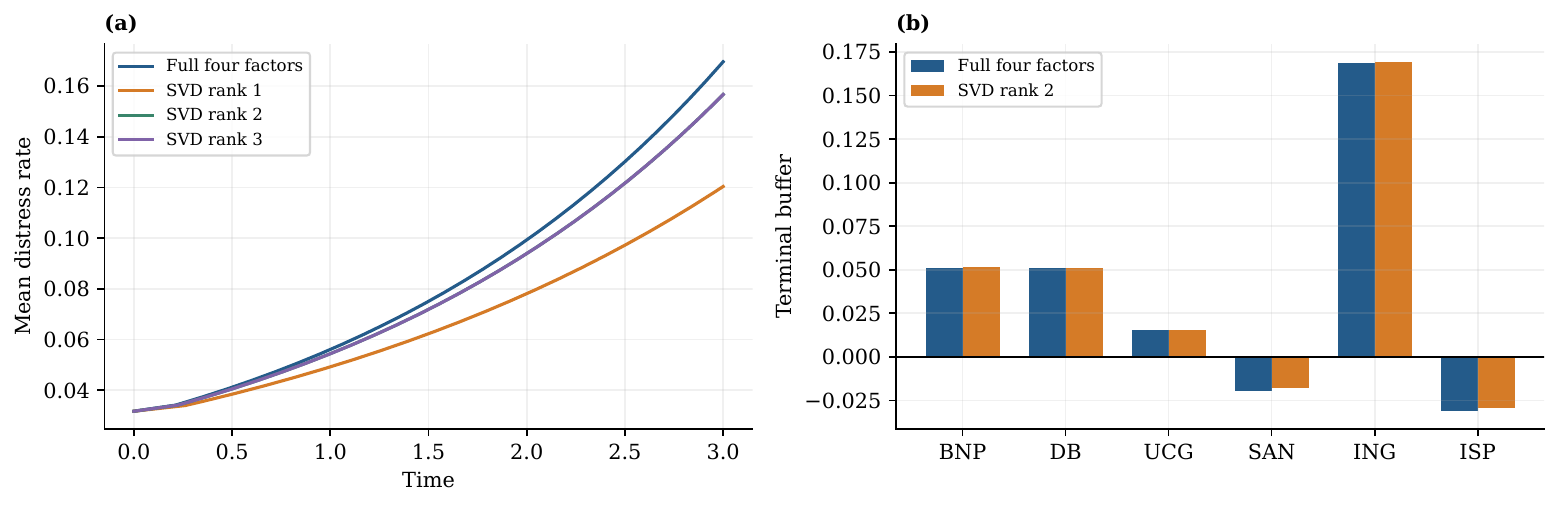}
\caption{Six-bank dynamics with identical initial buffers. (a) Mean distress-rate paths for the full four-factor kernel and its SVD truncations. (b) Terminal buffers under the full kernel and rank-two SVD.}
\label{fig:distress}
\end{figure}

\subsection{Scenario, coverage, and scale sensitivity}
\label{app:eba-robustness}

For the 116-bank construction, retained-book coverage has minimum $8.61\%$, fifth percentile $92.81\%$, median $100\%$, and mean $97.45\%$. We include the excluded holdings in a residual bucket, giving 39 factors. Recomputing both portfolio shares and initial buffers gives a terminal mean distress rate of $0.02459$ and four banks below zero, with maximum buffer-path difference $0.0503$. Keeping the original initial buffers in this 39-factor calculation reduces the maximum path difference to $0.00121$. The first comparison changes both the portfolio composition and the scenario applied to it; the second isolates the change in exposures.

The six-bank denominator comparison in \cref{tab:six-bank-denominator} separates these two effects in the same way. Using the full sovereign book as the denominator reduces the receiver loadings of the four retained channels. Recomputing buffers from these smaller shares leaves no bank initially distressed. Keeping the original buffers instead leaves two banks distressed at the terminal time, but weakens the loss flow and lowers the mean distress rate.

\begin{table}[htbp]
\centering\small
\caption{Six-bank receiver normalization and initial-buffer choices. Sender loadings and the four sovereign channels remain fixed.}
\label{tab:six-bank-denominator}
\begin{tabular}{lrr}
\toprule
Specification & Terminal mean distress rate & Terminal count\\
\midrule
Conditional shares, original buffers & 0.1696 & 2/6\\
Full-book shares, recomputed buffers & 0.0000 & 0/6\\
Full-book shares, fixed original buffers & 0.0530 & 2/6\\
\bottomrule
\end{tabular}
\end{table}

Changing from four to 38 factors at fixed $c_E=0.30$ would lower $\lambda_E$ from $0.075$ to $0.00789$. In the full population, this changes the terminal mean distress rate from $0.02462$ to $0.01378$, while the terminal threshold count stays four. We therefore hold $\lambda_E$ fixed when constructing the complete retained books, so that the number of represented channels does not itself change the row-mean interaction strength.

\paragraph{Buffer and intensity grid.}
The common buffer level takes the values $x_\ast\in\{0.36,0.38,0.40,0.42\}$. For each value, we use three smoothing widths, $\varepsilon\in\{0.025,0.05,0.10\}$, and three interaction strengths, $\lambda_E\in\{0.0375,0.075,0.15\}$. \Cref{tab:eba-robustness} reports the slice $\varepsilon=0.05$.
\begin{table}[htbp]
\centering\small
\caption{Full-population sensitivity at $\varepsilon=0.05$. Each cell gives the terminal threshold count and mean distress rate, with $N=116$.}
\label{tab:eba-robustness}
\begin{tabular}{rrrrr}
\toprule
$x_\ast$ & Initial count & $\lambda_E=0.0375$ & $\lambda_E=0.075$ & $\lambda_E=0.15$\\
\midrule
0.36 & 5 & 7; 0.0432 & 7; 0.0541 & 8; 0.0616\\
0.38 & 4 & 5; 0.0320 & 6; 0.0407 & 7; 0.0585\\
0.40 & 4 & 4; 0.0177 & 4; 0.0246 & 6; 0.0406\\
0.42 & 1 & 2; 0.0058 & 2; 0.0069 & 2; 0.0101\\
\bottomrule
\end{tabular}
\end{table}

\paragraph{Changing the stressed channels.}
To test the channel ordering under different stresses, we reassign the four shock magnitudes before constructing initial buffers. In France--Germany--Italy--Netherlands order, the three assignments are $(0.35,0.15,0.45,0.05)$, $(0.15,0.45,0.05,0.35)$, and $(0.45,0.05,0.15,0.35)$. The two largest shock weights determine the retained channels. Within each comparison, the full kernel, the two-channel kernel, and both SVD approximations start from the same buffers. \Cref{tab:eba-scenarios} gives their maximum buffer-path errors. Generic SVD truncations can contain negative weights: these account for $3.0\%$ of entries at rank two and $16.3\%$ at rank twelve. Their dynamics use the signed-kernel Lipschitz formulation; channel truncation preserves nonnegative exposures.

\begin{table}[htbp]
\centering\small
\caption{Maximum buffer-path errors under three preassigned shock vectors and three buffer levels. Counts are terminal numbers below zero in the full model.}
\label{tab:eba-scenarios}
\begin{tabular}{lrrrrr}
\toprule
Retained channels & $x_\ast$ & Count & Two channels & Rank-2 SVD & Rank-12 SVD\\
\midrule
IT--FR & 0.32 & 11 & 0.00469 & 0.05804 & 0.03481\\
IT--FR & 0.36 & 7 & 0.00298 & 0.05074 & 0.02734\\
IT--FR & 0.40 & 4 & 0.00072 & 0.02341 & 0.01122\\
DE--NL & 0.32 & 11 & 0.00770 & 0.10174 & 0.01967\\
DE--NL & 0.36 & 6 & 0.00446 & 0.03799 & 0.01497\\
DE--NL & 0.40 & 2 & 0.00026 & 0.00194 & 0.00075\\
FR--NL & 0.32 & 7 & 0.01044 & 0.09363 & 0.01432\\
FR--NL & 0.36 & 4 & 0.00373 & 0.03738 & 0.00750\\
FR--NL & 0.40 & 3 & 0.00117 & 0.01249 & 0.00314\\
\bottomrule
\end{tabular}
\end{table}

\subsection{A portfolio perturbation and its Lipschitz envelope}
\label{sec:sensitivity-envelope}

For the six-bank example, rescale every Italian holding by
\begin{equation}
\label{eq:italy-rescaling}
s_{i,\mathrm{IT}}\longmapsto(1+\delta)s_{i,\mathrm{IT}},\qquad -0.20\le\delta\le0.20,
\end{equation}
and recompute the receiver shares and scenario buffers. The sender loading for Italy stays unchanged because its numerator and cross-sectional mean have the same multiplier. With $\Delta x=\widetilde x-x$, $\Delta e=\widetilde e-e$, define
\[
\bar d=\max_i\frac1N\sum_j|\Delta e_{ij}|,
\qquad \bar e=\max_i\frac1N\sum_j|e_{ij}|=\lambda_E.
\]
The Lipschitz estimate yields
\begin{equation}
\label{eq:stability-envelope}
\max_i|\widetilde X_t^i-X_t^i|
\le(\|\Delta x\|_\infty+t\bar d)\exp(\lambda_E t/\varepsilon).
\end{equation}
\Cref{fig:stress} compares the resulting buffer changes with this envelope. The exponential factor grows from $e^{1.5}$ at $t=1$ to $e^{4.5}$ at $t=3$, making the global bound substantially looser at the longer horizon.

\begin{figure}[htbp]
\centering
\includegraphics[width=.98\textwidth]{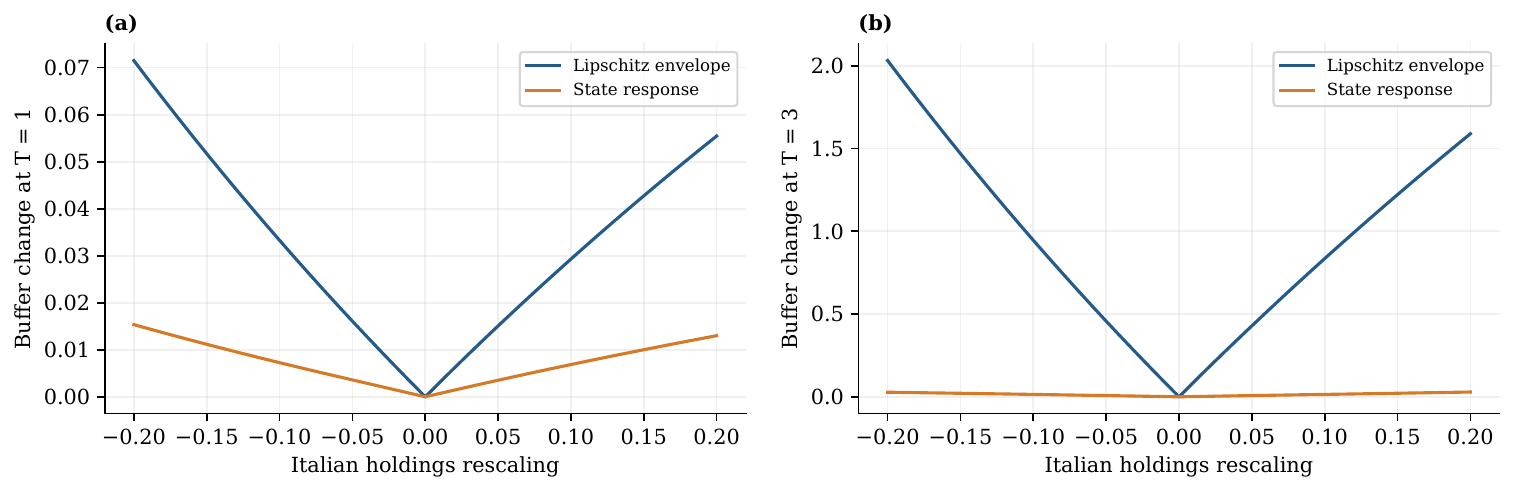}
\caption{Uniform Italian-book rescaling in the six-bank system. The receiver shares and initial buffers are recomputed; the common intensity stays $\lambda_E=0.075$. (a) State response and the Lipschitz envelope at $T=1$. (b) Their values at $T=3$.}
\label{fig:stress}
\end{figure}

\subsection{Resampling the disclosed population}
\label{app:eba-resampling}

We treat the uniform distribution on the 116 observed types as a fixed empirical reference law and draw 200 independent resamples at each $N\in\{8,16,32,64,128,256,512\}$, using seed 20261026. Each sampled type keeps its original factor loadings and initial buffer. The smoothed calculation uses DOP853 with tolerances $10^{-8}$ and $10^{-10}$; the hard calculation uses event integration. On 301 equally spaced times, we record the maximum feedback-coordinate error and the terminal Wasserstein distance between the resampled and reference state laws. Feedback coordinates use the paper's normalization $\beta_k=K\lambda_E\int b_k\ell\,\dd\varrho\,\dd t$.

The smoothed feedback and state-law slopes are $-0.455$ and $-0.506$; the hard-loss slopes are $-0.465$ and $-0.499$ (\cref{fig:fullsample-rate}). Error bars show $1.96$ Monte Carlo standard errors of the estimated mean. These fits describe convergence over the sampled population sizes; for the smoothed loss, we can also derive an expectation bound from the finite support of the reference law.

To obtain the $N^{-1/2}$ expectation bound for the smoothed loss, write $m=116$, $z_i=(x_i,a_i^{\mathrm{emp}},K\lambda_E b_i^{\mathrm{emp}})$, and $\mu_0=m^{-1}\sum_{i=1}^m\delta_{z_i}$, using the factor normalization of \cref{thm:finite-rank-w1}. Let $D=\max_{i,j}|z_i-z_j|$ in the $\ell^1$ type metric, and let $p_i^N$ be the resampled proportion of type $i$. Coupling the common mass at each atom and using $\operatorname{Var}(p_i^N)=(1/m)(1-1/m)/N$ gives
\begin{equation}
\label{eq:eba-finite-support-w1}
\E\Wone(\mu_0^N,\mu_0)
\le \frac D2\sum_{i=1}^m\E\left|p_i^N-\frac1m\right|
\le \frac D2\sqrt{\frac{m-1}{N}}.
\end{equation}
For fixed $\varepsilon>0$, \cref{thm:finite-rank-w1} controls the joint state--factor laws. The weighted loss $b_k\ell_\varepsilon^-(x)$ is uniformly Lipschitz on their bounded factor support, so integrating its expectation also controls the feedback paths. Consequently,
\begin{equation}
\label{eq:eba-smoothed-sampling-rate}
\E\left[\sup_{t\le T}\Wone(\widehat\nu_t^N,\widehat\nu_t)
+\sup_{t\le T}\norm{\beta^N(t)-\beta(t)}_\infty\right]
\le C_{T,\varepsilon,\mu_0}N^{-1/2}.
\end{equation}
The constant depends on the fixed observed type law, interaction strength, horizon, and smoothing width, and is independent of $N$. Projection onto the buffer coordinate gives the same bound for the terminal state-law error plotted here.

\begin{figure}[htbp]
\centering
\includegraphics[width=.98\textwidth]{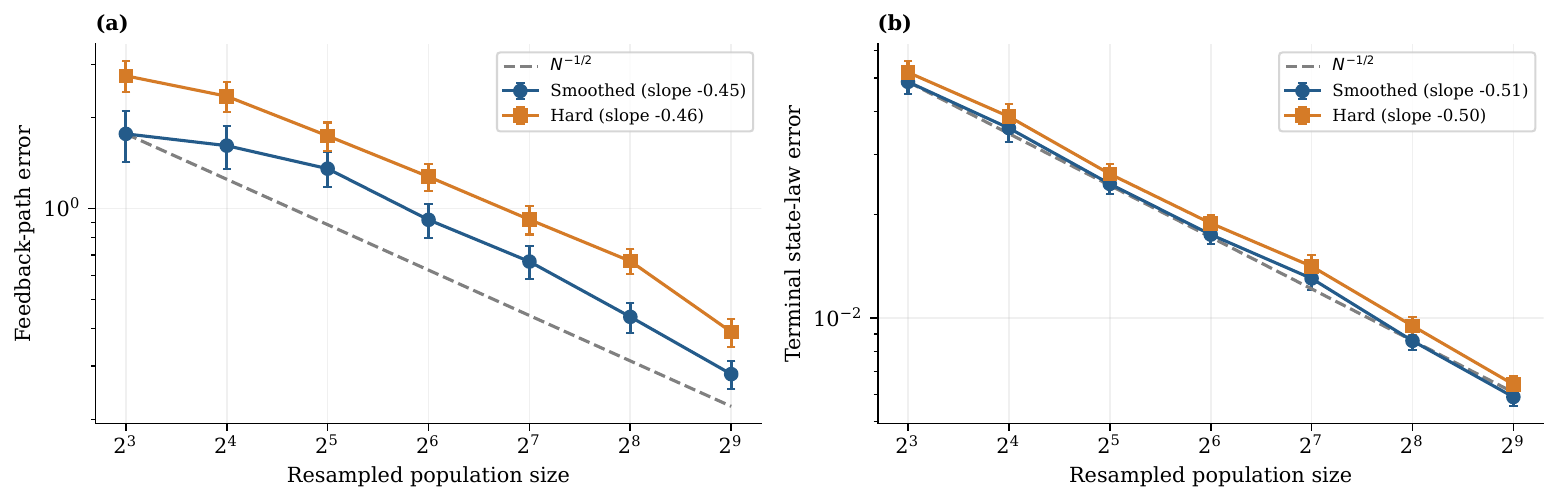}
\caption{Resampling from the uniform empirical law on 116 observed bank types, with 200 replications per population size. (a) Mean feedback-path error. (b) Mean terminal state-law Wasserstein error. Error bars are $1.96$ standard errors of the Monte Carlo means. The dashed line has slope $-1/2$.}
\label{fig:fullsample-rate}
\end{figure}

\end{document}